\documentclass[11pt]{llncs}
\usepackage{tabularx,booktabs,multirow,delarray,array}
\usepackage{graphicx,amssymb,amsmath,amssymb}
\usepackage{latexsym}
\usepackage[absolute]{textpos}

\usepackage[in]{fullpage}

\usepackage{lineno}

\def\calP{\mathcal{P}}
\def\calH{\mathcal{H}}

\def\calV{\mathcal{V}}
\def\calC{\mathcal{C}}
\def\calB{\mathcal{B}}
\def\st{$s$-$t$}

\newcommand{\vd}{\mbox{$V\!D$}}

\newenvironment{proof}{\noindent {\textbf{Proof:}}\rm}{\hfill $\Box$\rm}
\newtheorem{observation}{Observation}

\begin{document}

\title{Bicriteria Rectilinear Shortest Paths among Rectilinear Obstacles in the Plane}

\author{
Haitao Wang
}

\institute{
Department of Computer Science\\
Utah State University, Logan, Utah 84322, USA\\
\email{haitao.wang@usu.edu}
}

\maketitle

\pagenumbering{arabic}
\setcounter{page}{1}

\vspace*{-0.2in}
\begin{abstract}
Given a rectilinear domain $\calP$ of $h$ pairwise-disjoint rectilinear obstacles with a total of $n$ vertices in the plane, we study the problem of computing bicriteria rectilinear shortest paths between two points $s$ and $t$ in $\calP$. Three types of bicriteria rectilinear paths are considered: minimum-link shortest paths, shortest minimum-link paths, and minimum-cost paths where the cost of a path is a non-decreasing function of both the number of edges and the length of the path. The one-point and two-point path queries are also considered. Algorithms for these problems have been given previously. Our contributions are threefold.
First, we find a critical
error in all previous algorithms. Second, we correct the error
in a not-so-trivial way. Third, we further improve the algorithms so
that they are even faster than the previous (incorrect) algorithms
when $h$ is relatively small.
For example, for the minimum-link shortest paths, we obtain the following results. Our algorithm computes a minimum-link shortest \st\ path
in $O(n+h\log^{3/2} h)$ time.
For the one-point queries, we build a data structure of size $O(n+
h\log h)$ in $O(n+h\log^{3/2} h)$ time for a source point $s$, such that given any query point $t$, a minimum-link shortest \st\ path can be computed in $O(\log n)$ time.
For the two-point queries, with $O(n+h^2\log^2 h)$ time and space preprocessing,
a minimum-link shortest \st\ path can be computed in $O(\log n+\log^2 h)$ time for any two query points $s$ and $t$; alternatively, with $O(n+h^2\cdot \log^{2} h \cdot 4^{\sqrt{\log h}})$ time and $O(n+h^2\cdot \log h \cdot 4^{\sqrt{\log h}})$ space preprocessing, we can answer each two-point query in $O(\log n)$ time. Note that $h^2\cdot \log^{2} h \cdot 4^{\sqrt{\log h}}=O(h^{2+\epsilon})$ for any $\epsilon>0$. These results are particularly interesting when $h$ is relatively small. For example, if $h=O(n^{1/2-\epsilon})$ for any $\epsilon>0$, then all above results match the best results for the problems in simple rectilinear polygons, which are optimal. The complexities for the other two types of paths are slightly worse, but still linearly depend on $n$ (in addition to $g(h)$ for some functions $g(h)$ of $h$).
\end{abstract}

\section{Introduction}
\label{sec:intro}

Let $\calP$ be a {\em rectilinear domain} with a total of $h$ holes and $n$ vertices in the plane, i.e., $\calP$ is a multiply-connected region whose boundary is a union of $n$ axis-parallel line segments, forming $h+1$ closed polygonal cycles (i.e., $h$ holes plus an outer boundary). A simple rectilinear polygon is a special case of a rectilinear domain with $h=0$. A {\em rectilinear} path is a path consisting of only horizontal and vertical line segments.

For a rectilinear path $\pi$, we define its {\em length} as the total sum of the lengths of the segments of $\pi$, and we define its {\em link distance} as the number of edges of $\pi$ (each edge is also called a {\em link}). We use the {\em measure} of $\pi$ to refer to both its length and its link distance.
For any two points $s$ and $t$ in $\calP$, a {\em shortest rectilinear path} from $s$ to $t$ is a rectilinear path connecting $s$ to $t$ in $\calP$ with the minimum length, and a {\em minimum-link rectilinear path} is a rectilinear \st\ path with the minimum link distance. Among all shortest rectilinear \st\ paths, the one with the minimum link distance is called a {\em minimum-link shortest \st\ path}; among all minimum-link \st\ paths, the one with the minimum length is called a {\em shortest minimum-link \st\ path}.
We define the {\em cost} of $\pi$ as a non-decreasing function $f$ of both the length and the link distance of $\pi$. We assume that given the number of links of $\pi$ and the length of $\pi$, its cost can be computed in constant time. Depending on the context, the measure of $\pi$ may also refer to its cost. A {\em minimum-cost} path from $s$ to $t$ is a rectilinear \st\ path in $\calP$ with the minimum cost (with respect to the cost function $f$).

All the three types of paths discussed above (i.e.,  minimum-link
shortest paths, shortest minimum-link paths, and minimum-cost paths)
are called {\em bicriteria shortest paths}. In order to differentiate
between ``bicriteria shortest paths'' and ``shortest paths'', we will
use {\em optimal paths} to refer to these bicriteria shortest paths.
Since some observations and algorithmic schemes may be applicable to
all three types of optimal paths, unless otherwise stated, a statement
made to ``optimal paths'' should be applicable to all three types of
optimal paths.

In this paper, we study the problem of computing all three types of optimal paths between two points $s$ and $t$ in $\calP$. Their one-point and two-point queries are also considered.

\subsection{Previous Work}

These problems have been studied before. The following results are applicable to all three types of optimal paths.

Yang et al.~\cite{ref:YangOn92} first presented an $O(nr + n\log n)$ time algorithm, where $r$ is the number of extreme edges of $\calP$ (an edge $e$ of $\calP$ is {\em extreme} if its two adjacent edges lie on the same side of the line containing $e$; $r=\Omega(n)$ in the worst case). Later, Yang et al.~\cite{ref:YangRe95} proposed an algorithm of $O(n\log^2 n)$ time and $O(n\log n)$ space and another algorithm of $O(n\log^{3/2} n)$ time and space; Chen et. al.~\cite{ref:ChenOn01} improved the algorithm to $O(n\log^{3/2} n)$ time and $O(n\log n)$ space.

The {\em one-point optimal path query problem}, where $s$ is the
source and $t$ is a query point, was also studied. Based on the
algorithm of Yang et al.~\cite{ref:YangRe95}, Chen et.
al.~\cite{ref:ChenOn01} built a data structure of $O(n\log n)$ size in
$O(n\log^{3/2}n)$ time such that for each query point $t$, the measure
of the optimal \st\ path can be computed in $O(\log n)$ time and an
actual path can be output in additional time linear in the number of edges of the path. For simplicity, in the following, when we say that the query time of a data structure for finding a path is $O(g(n))$, we mean that the measure of the path can be computed in $O(g(n))$ time and an actual path can be output in additional time linear in the number of edges of the path. 

The {\em two-point optimal path query problem}, i.e., both $s$ and $t$ are query points, was also studied by Chen et. al.~\cite{ref:ChenOn01}, where a data structure of $O(n^2\log^2 n)$ size was built in $O(n^2 \log^{2}n)$ time such that each two-point query can be answered in $O(\log^2 n)$ time.

\subsection{Our Results}

We provide a comprehensive study on these problems. Our contributions are threefold.

First, we show that all the algorithms in the previous work mentioned
above are incorrect. More specifically, we find a critical error in
the algorithm of Yang et al.~\cite{ref:YangRe95}. Since the algorithms
and data structures of Chen et. al.~\cite{ref:ChenOn01} are all based
on the method of Yang et al.~\cite{ref:YangRe95}, the above results of
Chen et. al.~\cite{ref:ChenOn01} are not correct either. A similar
error also appears in the algorithm of  \cite{ref:YangOn92}. We should
point out that the technique of Chen et. al.~\cite{ref:ChenOn01},
which follows the similar idea in Chen et al.~\cite{ref:ChenSh00} for
computing $L_1$ shortest paths in arbitrary polygonal domains, would
work if it were based on a correct algorithm (for example, it still works in our new algorithm).

Second, we fix the error of Yang et al.~\cite{ref:YangRe95} in a
not-so-trivial way. However, the complexities are not the same as
before for all three types of optimal paths. Specifically, for
computing a minimum-link shortest path, our corrected algorithm runs
in $O(n\log^{3/2} n)$ time and $O(n\log n)$ space (with the help of
the technique of Chen et. al.~\cite{ref:ChenOn01} to
reduce a factor of $\log^{1/2}n$).
For the other two types of optimal paths, however, the complexities have one more $O(n)$ factor, i.e., $O(n^2\log^{3/2} n)$ time and $O(n^2\log n)$ space.

Third, we further improve the algorithms in the way that the complexities only linearly depend on $n$ (in addition to $g(h)$ for some functions $g(h)$ of $h$). For computing a minimum-link shortest path, our algorithm runs in $O(n+h\log^{3/2} h)$ time and $O(n+h\log h)$ space. For computing other two types of optimal paths, our algorithm runs in $O(n+h^2\log^2 h)$ time and $O(n+h^2\log h)$ space.
We also obtain data structures for one-point and two-point queries.
Our results are summarized in Table~\ref{tab:results}. Note that for two-point queries, we give two data structures for each problem with tradeoff between the preprocessing and the query time. We also consider the two-point query problem for minimum-link paths (without considering the lengths) since the problem was not studied before (but the one-point query problem has already been studied, as discussed below).

Our results are particularly interesting when $h$ is
relatively small. For example if $h=O(n^{1/2-\epsilon})$ for any
$\epsilon>0$, then for finding a single optimal path of any type, our
algorithm runs in $O(n)$ time, and our data structures for the
minimum-link shortest path and minimum-link path queries are also
optimal.

\begin{table}[t]
\begin{center}
{
\footnotesize
\begin{tabularx}{0.95\textwidth}{lllll}
\toprule
\addlinespace[1ex]
  & & One-Point Queries\ \ \  & \multicolumn{2}{c}{Two-Point Queries} \\
\addlinespace[1ex]
\midrule[0.01in]
 \multirow{3}{*}{Min-Link Shortest Paths\ \ } &Preprocess Time\ \ &$O(n + h\log^{3/2} h)$ & $O(n +  h^2\log^{2} h)$& $O(n+h^2\log^2 h4^{\sqrt{\log h}})$\\
 & Space & $O(n + h\log h)$ & $O(n+  h^2\log^{2} h)$  & $O(n+h^2\log h4^{\sqrt{\log h}})$ \\
 & Query Time & $O(\log n)$ \  & $O(\log n+\log^2 h)$ \ \  & $O(\log n)$\\
\midrule[0.008in]
 \multirow{3}{*}{Shortest Min-Link Paths\ \ } &Preprocess Time\ \ &$O(n + h^2\log^{3/2} h)$ & $O(n +  h^3\log^{2} h)$& $O(n+h^3\log^2 h4^{\sqrt{\log h}})$\\
 & Space & $O(n + h^2\log h)$ & $O(n+  h^3\log^{2} h)$  & $O(n+h^3\log h4^{\sqrt{\log h}})$ \\
 & Query Time & $O(\log n+\log^2 h)$ \  & $O(\log n+\log^3 h)$ \ \  & $O(\log n+\log^2 h)$\\
 \midrule[0.008in]
 \multirow{3}{*}{Minimum-Cost Paths\ \ } &Preprocess Time\ \ &$O(n + h^2\log^{3/2} h)$ & $O(n +  h^3\log^{2} h)$& $O(n+h^3\log^2 h4^{\sqrt{\log h}})$\\
 & Space & $O(n + h^2\log h)$ & $O(n+  h^3\log^{2} h)$  & $O(n+h^3\log h4^{\sqrt{\log h}})$ \\
 & Query Time & $O(\log n+h\log h)$ \  & $O(\log n+h\log^2 h)$ \ \  & $O(\log n+h\log h)$\\
 \midrule[0.008in]
 \multirow{3}{*}{Minimum-Link Paths\ \ } &Preprocess Time\ \ & & $O(n +  h^2\log^{2} h)$& $O(n+h^2\log^2 h4^{\sqrt{\log h}})$\\
 & Space &  & $O(n+  h^2\log^{2} h)$  & $O(n+h^2\log h4^{\sqrt{\log h}})$ \\
 & Query Time &  & $O(\log n+\log^2 h)$ \ \  & $O(\log n)$\\
\bottomrule
\end{tabularx}
\vspace*{0.1in}
\caption{\footnotesize Summary of our data structures on one-point and two-point optimal path queries. Note that $\log^{2} h \cdot 4^{\sqrt{\log h}}=O(h^{\epsilon})$ for any $\epsilon>0$.
}
\label{tab:results}
}
\end{center}
\vspace{-0.4in}
\end{table}

It is easy to see that the minimum-link shortest paths and the shortest minimum-link paths are special cases of minimum-cost paths, and we discuss them separately mainly because our results for the two special cases are better that those for the minimum-cost paths.
In fact, as the cost function $f$ is quite general, our algorithm for computing minimum-cost paths may find many applications. We give two examples below.

Polishchuk and Mitchell~\cite{ref:PolishchukkLink05} gave an
$O(kn\log^2 n)$ time algorithm for computing a shortest \st\ path with
at most $k$ links for a given integer $k$, which improves the $O(kn^2)$ time algorithm in \cite{ref:YangOn92}.
As indicated in~\cite{ref:PolishchukkLink05}, the problem can be solved using any algorithm that can find a minimum-cost path with the cost function defined as $f(a,b)=a$ if $b\leq k$ and $f(a,b)=\infty$ otherwise, where $a$ and $b$ are the length and the link distance of the path, respectively.
Partially due to this reason, Polishchuk and
Mitchell~\cite{ref:PolishchukkLink05} already suspected that there is
a misunderstanding on the algorithms of
\cite{ref:ChenOn01,ref:YangRe95} for computing minimum-cost paths. We
thus confirm their suspicion. On the other hand, applying our new (and
correct) algorithm for minimum-cost paths can solve the problem in
$O(n+h^2\log^{3/2} h)$ time, which is faster than the algorithm
in~\cite{ref:PolishchukkLink05} when $h$ is sufficiently small or when $k$ is relatively large.

As a dual problem, finding a minimum-link \st\ path with length at
most a given value $l$ was also studied in~\cite{ref:YangOn92}, where a
worst-case $O(n^2(r+\log n))$ time algorithm was given with
$r$ as the number of extreme edges of $\calP$. The problem can also be solved using any minimum-cost path algorithm by defining the cost function as $f(a,b)=b$ if $a\leq l$ and $f(a,b)=\infty$ otherwise. Hence, applying our algorithm for minimum-cost paths can solve the problem in $O(n+h^2\log^{3/2} h)$, which improves the algorithm of~\cite{ref:YangOn92} since it holds that $r\geq h$.

\subsection{Other Related Work}

If $\calP$ is a simple rectilinear polygon (i.e., $h=0$), then there always exists a rectilinear \st\ path that has both the minimum length and the minimum link distance for any two points $s$ and $t$ in $\calP$~\cite{ref:deBergOn91,ref:HershbergerCo94}. de Berg~\cite{ref:deBergOn91} built a data structure of $O(n\log n)$ size in $O(n\log n)$ time that can find such a path in $O(\log n)$ time for any two-point query. The preprocessing time and space were both reduced to $O(n)$ by Schuierer~\cite{ref:SchuiererAn96} (with $O(\log n)$ query time).

If $\calP$ is a general rectilinear domain with $h\neq 0$, then there may not exist a rectilinear path that is both a minimum-link path and a shortest path~\cite{ref:YangOn92}. The problems of finding only minimum-link paths or only shortest paths have been studied extensively.
Imai and Asano~\cite{ref:ImaiEf86} presented an $O(n\log n)$ time and space algorithm for finding a minimum-link \st\ path in $\calP$, and the space was reduced to $O(n)$~\cite{ref:DasGe91,ref:MitchellMi14,ref:SatoA87}. Recently, Mitchell et al.~\cite{ref:MitchellAn15} proposed an $O(n+h\log h)$ time and $O(n)$ space algorithm for the problem, after $\calP$ is triangulated (which can be done in $O(n\log n)$ time or $O(n+h\log^{1+\epsilon}h)$ time for any $\epsilon>0$~\cite{ref:Bar-YehudaTr94}). The algorithms in~\cite{ref:DasGe91,ref:MitchellMi14,ref:MitchellAn15} also construct an $O(n)$ size data structure that can answer each one-point minimum-link path query in $O(\log n)$ time.

For computing shortest \st\ paths in $\calP$, Clarkson et al.~\cite{ref:ClarksonRe87} gave an algorithm of $O(n\log^2 n)$  time and $O(n\log n)$ space, and as a tradeoff between time and space, they modified their algorithm so that it runs in $O(n\log^{3/2}n)$ time and space~\cite{ref:ClarksonRe88}. Wu et al.~\cite{ref:WuRe87} proposed an $O(n\log r +r^2\log t)$ time algorithm, where $r$ is the number of extreme edges of $\calP$, and the algorithm was later improved to $O(n\log r+r\log^{3/2}r)$ time~\cite{ref:YangRe95}. Mitchell~\cite{ref:MitchellAn89,ref:MitchellL192} solved the problem in $O(n\log n)$ time and $O(n)$ space, and Chen and Wang~\cite{ref:ChenA11,ref:ChenL113STACS} reduced the time to $O(n+h\log h)$ after $\calP$ is triangulated.

If $\calP$ is an arbitrary polygonal domain (i.e., not rectilinear), then the results from \cite{ref:ChenA11,ref:ChenL113STACS,ref:ClarksonRe87,ref:ClarksonRe88,ref:MitchellAn89,ref:MitchellL192} are also applicable to finding arbitrary shortest paths under $L_1$ metric. In addition, the algorithms in \cite{ref:ChenA11,ref:ChenL113STACS,ref:MitchellAn89,ref:MitchellL192} can be used to compute an $O(n)$ size data structure so that each one-point $L_1$ shortest path query can be answered in $O(\log n)$ time. For two-point $L_1$ shortest path queries, Chen et al.~\cite{ref:ChenSh00} constructed a data structure of size $O(n^2\log n)$ in $O(n^2\log^2 n)$ time that can answer each query in $O(\log^2 n)$ time. Recently, Chen et al.~\cite{ref:ChenTw16} reduced the query time to $O(\log n)$ by building a data structure of size  $O(n+h^2\cdot \log h \cdot 4^{\sqrt{\log h}})$ in $O(n+h^2\cdot \log^2 h \cdot 4^{\sqrt{\log h}})$ time.

To find a minimum-link \st\ path between two points $s$ and $t$ in an arbitrary polygonal domain $\calP$, Mitchell~\cite{ref:MitchellMi92} gave an $O(E\alpha(n)\log^2 n)$ time algorithm, where $\alpha(n)$ is the  inverse of Ackermann's function and $E$ is the size of the visibility graph of $\calP$ and $E=\Theta(n^2)$ in the worst case.
The one-point query problem was also studied in~\cite{ref:MitchellMi92}.

In the following, unless otherwise stated, a path always refers to a rectilinear path.

\subsection{Our Techniques}

Given two points $s$ and $t$ in the rectilinear domain $\calP$, to
find an optimal \st\ path, the algorithm of Yang et
al.~\cite{ref:YangRe95} first built a ``path-preserving'' graph $G$ of size $O(n\log n)$ by using the idea of Clarkson et al.~\cite{ref:ClarksonRe87}. Then, it
is shown that $G$ contains an \st\ path $\pi_G(s,t)$ that is homotopic
to an optimal \st\ path $\pi(s,t)$ in $\calP$ with the same
length, and further, $\pi(s,t)$ can be obtained from $\pi_G(s,t)$ by performing certain ``dragging'' operations. Motivated by this observation, Yang et
al.~\cite{ref:YangRe95} computed an optimal \st\ path by applying
Dijkstra's algorithm on $G$ and simultaneously performing the
dragging operations. We find a critical error in their way of applying
Dijkstra's algorithm. We fix the error
by using a ``path-based'' Dijkstra's algorithm and maintaining some additional  information, and we prove that our algorithm is correct. Due to that we need to maintain more information on computing shortest minimum-link paths and minimum-cost paths, our algorithm for them runs slower than that for computing minimum-link shortest paths.

To further reduce the running time (for small $h$), our main idea is to use a reduced graph $G_r$ of size $O(h\log h)$ instead of $G$. We show
that $G_r$ contains an \st\ path $\pi_{G_r}(s,t)$ that is homotopic
to an optimal \st\ path $\pi(s,t)$ in $\calP$ with the same
length, and further, $\pi(s,t)$ can be obtained from $\pi_{G_r}(s,t)$ by performing the dragging operations as in~\cite{ref:YangRe95} and a new kind of operations, called {\em through-corridor-path generating operations}.
The graph $G_r$ is built based on a corridor structure of $\calP$,
which was used to find minimum-link paths in~\cite{ref:MitchellAn15}.
More specifically, we decompose $\calP$ into $O(h)$ {\em junction rectangles} and $O(h)$ {\em corridors}. Each
corridor is a simple rectilinear polygon. Although each corridor may have $\Theta(n)$
vertices, we show that we only need to consider at most four points of each corridor to build the graph $G_r$. To this end, we make use of
the histogram partitions of rectilinear simple polygons~\cite{ref:SchuiererAn96}.

To solve the one-point queries, the approach of Chen et al.~\cite{ref:ChenOn01}
is to ``insert'' the query point $t$ to the graph $G$ to obtain
a set $V_g(t)$ of $O(\log n)$  vertices (called ``gateways'') of $G$ such that an optimal path can be obtained by performing the dragging operations from the gateways. We follow the similar scheme but on our reduced graph $G_r$, where only $O(\log h)$ gateways are necessary. Further, we also need to  utilize the techniques of Schuierer~\cite{ref:SchuiererAn96} for simple rectilinear polygons.

For the two-point queries, the approach of Chen et al.~\cite{ref:ChenOn01}
inserts both query points $s$ and $t$ to the graph $G$ to obtain a
 set $V_g(s)$ of $O(\log n)$ gateways for $s$ and a
set $V_g(t)$ of $O(\log n)$ gateways of $G$ for $t$, so that an optimal
\st\ path can be obtained by performing the dragging operations from
these gateways. The query time becomes $O(\log^2 n)$ because every pair of
points $(p,q)$ with $p\in V_g(s)$ and $q\in V_g(t)$ needs to be considered.
We again use the same scheme but on the graph $G_r$ with only $O(\log h)$ gateways for both $s$ and $t$, which reduces the query time to $O(\log n+ \log^2 h)$. To further reduce the query time to
$O(\log n)$, we follow the scheme in~\cite{ref:ChenTw16} for solving two-point $L_1$ shortest path queries in arbitrary polygonal domains. The main idea is to build a larger graph by adding more vertices
to $G_r$ such that $O(\sqrt{\log h})$ gateways are sufficient for each query point.

The rest of the paper is organized as follows. We introduce some notation and concepts in Section~\ref{sec:pre}.
In Section~\ref{sec:old}, we review the algorithm  given by Yang, Lee, and Wong~\cite{ref:YangRe95} (we refer to it as the YLW algorithm), point out the error, and correct it. 
In Section~\ref{sec:new}, we further improve the algorithm for finding a single optimal \st\ path. The one-point and two-point path query problems are discussed in Sections~\ref{sec:onepoint} and \ref{sec:twopoint}, respectively.

\section{Preliminaries}
\label{sec:pre}

In this section, we define notation and review some concepts.
Some terminologies are borrowed from the previous work, e.g.,
\cite{ref:ChenTw16,ref:ChenSh00,ref:ClarksonRe87,ref:YangRe95}

For any two points $p$ and $q$ of $\calP$,
if the line segment $\overline{pq}$ is in $\calP$, then we say that $p$ is {\em visible} to $q$.
Consider a vertical line $l$ and a point $p\in \calP$. Let $p'$ be the point on $l$ whose $y$-coordinate is the same as that of $p$. We call $p'$ the {\em horizontal projection} of $p$ on $l$. If $p$ is visible to $p'$, then we say that $p$ is {\em horizontally visible} to $l$.

For any two points $p$ and $q$, we use $R_{pq}$ to denote the rectangle with $\overline{pq}$ as a diagonal. A path in $\calP$ is {\em L-shaped} if it consists of a horizontal segment and a vertical segment (each of them may be empty).
A path is {\em U-shaped} if it consists of three segments $s_1$, $s_2$, and $s_3$ such that $s_1$ and $s_3$ are on the same side of the line containing $s_2$ (e.g., see Fig.~\ref{fig:staircase}). A path is called a {\em staircase path} if it does not contain a U-shaped subpath. Note that a staircase path is a shortest path.

\begin{figure}[t]
\begin{minipage}[t]{\linewidth}
\begin{center}
\includegraphics[totalheight=1.0in]{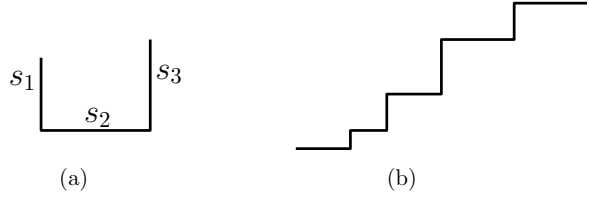}
\caption{\footnotesize (a) a U-shaped path; (b) a staircase path.
}
\label{fig:staircase}
\end{center}
\end{minipage}
\vspace*{-0.15in}
\end{figure}

Let $\calV$ denote the set of all vertices of $\calP$. We let $\calV$ also include the two points $s$ and $t$. We review a  path-preserving graph $G(\calV)$ on $\calV$, which was originally from \cite{ref:ClarksonRe87} and has been used elsewhere, e.g., \cite{ref:ChenTw16,ref:ChenSh00,ref:LeeSh91,ref:YangRe95}. The vertex set of $G(\calV)$ consists of the points of $\calV$ and {\em Steiner points} on some vertical lines, called {\em cut-lines}. The cut-lines and the Steiner points are defined as follows.

Let $v_m$ be the point of $\calV$ with the median $x$-coordinate. The vertical
line $l_m$ through $v_m$ is a {\em cut-line}. For each point $v\in \calV$, if $v$ is horizontally visible to $l_m$, then the horizontal projection of $v$ on $l_m$ is a {\em Steiner point}. Let $\calV_l$ (resp., $\calV_r$) be the points of $\calV$ on the left (resp., right) side of $l_m$.
The cut-lines and Steiner points on the left and right sides of $l_m$ are defined on $\calV_l$ and $\calV_r$, recursively. We use a binary tree $T(\calV)$ to represent the above recursive procedure, called {\em cut-line tree}.
Each node $u\in T(\calV)$ corresponds to a cut-line $l(u)$ and a subset $V(u)\subseteq \calV$. If $u$ is the root, then $l(u)$ is $l_m$ and $V(u)=\calV$. The left and right subtrees of the root are defined recursively on $\calV_l$ and $\calV_r$. Hence, $T(\calV)$ has $O(n)$ nodes and each point of $\calV$ can define a Steiner point on at most $O(\log n)$ cut-lines. Therefore, there are $O(n\log n)$ Steiner points in total.

The vertex set of $G(\calV)$ consists of all points of $\calV$ and all Steiner points defined above. The edges of the graph are defined as follows. First, if a point $v\in \calV$ defines a Steiner point $v'$ on a cut-line, then $G(\calV)$ has an edge $\overline{vv'}$. Second, for any two adjacent Steiner points $p_1$ and $p_2$ on each cut-line, if the two points are visible to each other, then $G(\calV)$ has an edge $\overline{p_1p_2}$.

Clearly, $G(\calV)$ has $O(n\log n)$ nodes and $O(n\log n)$ edges. Each edge of the graph is either horizontal or vertical. Each edge of $G(\calV)$ has a weight that is the length of the corresponding line segment.
The graph $G(\calV)$ can be built in $O(n\log^{2} n)$
time~\cite{ref:ClarksonRe87,ref:LeeSh91,ref:YangRe95}\footnote{The
graph $G(\calV)$ introduced in \cite{ref:ClarksonRe87} also includes
Steiner points on horizontal cut-lines and projection points of
$\calV$ on the boundary of $\calP$. However, in our problem, since
$\calP$ is rectilinear, by the similar analysis as in
\cite{ref:ClarksonRe87}, we can show that our graph $G(\calV)$ is also
a path-preserving graph. We will give analysis details when we prove a
similar observation on our reduced graph $G_r$ in
Section~\ref{sec:reduced} (i.e., Lemma~\ref{lem:target} and Corollary~\ref{coro:target}).}. The
following lemma will be useful later.

\begin{lemma}\label{obser:rectangle}{\em \cite{ref:ClarksonRe87,ref:LeeSh91,ref:YangRe95}}
For any two points $p$ and $q$ in $\calV$, if $R_{pq}$ is empty (i.e., $R_{pq}$ is in $\calP$), then $G(\calV)$ contains a staircase path from $p$ to $q$.
\end{lemma}

For any path $\pi$ in $\calP$, let $L_1(\pi)$ denote its length and
let $L_d(\pi)$ denote its link distance.
For any two points $a$ and $b$ on $\pi$, if the context is clear, we often use
$\pi(a,b)$ to denote the subpath of $\pi$ between $a$ and $b$.
For any two points $p$ and
$q$ in the plane, we say that $q$ is to the {\em northeast} of $p$ if
$q$ is in the first quadrant (including its boundary) with respect to
$p$.  Similarly, we define {\em northwest, southwest}, and {\em
southeast}.

\section{The YLW Algorithm and Our Correction}
\label{sec:old}

In this section, we first review the YLW algorithm \cite{ref:YangRe95} and then point out the error. Finally, we will fix the error and prove the correctness our new algorithm.

\subsection{The YLW Algorithm}

The YLW algorithm is
essentially based on the following observation.

\begin{lemma}\label{lem:homo} {\em (Yang et al.~\cite{ref:YangRe95})}
For any optimal path $\pi$ from $s$ to $t$ in $\calP$, there is path $\pi_G$ in $G(\calV)$ such that $L_1(\pi_G)=L_1(\pi)$ and $\pi_G$ is  homotopic to $\pi$ (i.e., $\pi_G$ can be continuously dragged to $\pi$ without going outside of $\calP$).
\end{lemma}

We briefly review the proof of Lemma~\ref{lem:homo} because it will help to understand the algorithm and also help us to prove the correctness of our new algorithm given later.

Let $\pi$ be any optimal path from $s$ to $t$. It is shown (Lemma 2.1
\cite{ref:YangRe95}) that $\pi$ can be divided into a sequence of
staircase subpaths, and the two endpoints of each such subpath are in
$\calV$. Hence, it is sufficient to prove the lemma for any
staircase subpath of $\pi$. In the following, we consider a staircase
subpath $\pi(p,q)$ of $\pi$ with $p$ and $q$ as the two endpoints. We
further obtain a {\em pushed staircase path} as follows. Without loss
of generality, we assume $q$ is to the northeast of $p$ and the
segment of $\pi(p,q)$ incident to $p$ is horizontal. We push the first
vertical segment of $\pi(p,q)$ rightwards until either it hits a
vertex of $\calV$ or it becomes collinear with the second vertical segment
of $\pi(p,q)$. If the latter case happens, then we merge the two
vertical segments and keep pushing the merged vertical segment
rightwards. If the first case happens, then we push the next
horizontal segment upwards in a similar way. The procedure stops until
we arrive at the segment incident to $q$. Let $\pi'$ denote the
resulting path.
Observe that $L_1(\pi')=L_1(\pi(p,q))$, $\pi'$ is homotopic to $\pi(p,q)$, and $\pi'$ is also a staircase path. $\pi'$ is called a {\em pushed staircase path}~\cite{ref:YangRe95}. Also note that each segment of $\pi'$ contains at least one vertex of $\calV$.

\paragraph{Remark.}
There are eight types of pushed staircase paths from $p$ to $q$ depending on which quadrant of $p$ the point $q$ lies in and also depending on whether the first segment of the path incident to $p$ is horizontal or vertical.
\paragraph{}

The vertices of $\calV$ partition $\pi'$ into subpaths.
To prove the lemma, it is sufficient to show the following {\em claim}: for any subpath $\pi'(p',q')$ of $\pi'$ between any two adjacent vertices $p'$ and $q'$ of $\calV$ on $\pi'$, there is a path $\pi_G(p',q')$ connecting $p'$ and $q'$ in $G(\calV)$ with the same length and the two paths are homotopic. Because every segment of $\pi'$ contains at least one vertex of $\calV$, $\pi'(p',q')$ must be an L-shaped path. Without loss of generality, we assume $q'$ is to the northwest of $p'$.
If the rectangle $R_{p'q'}$ is empty (this includes the case where $\pi'(p',q')$ is a single segment), then by Lemma~\ref{obser:rectangle}, the above claim is true. Otherwise, as shown in \cite{ref:YangRe95} (Lemma 4.5), there are some points of $\calV$ in $R_{p'q'}$ that can be ordered as $p'=v_0,v_1,\ldots,v_t=q'$ with $R_{v_{i-1}v_{i}}$ being empty and $v_i$ to the northwest of $v_{i-1}$ for each $1\leq i\leq t$, and further, $\pi'(p',q')$ is homotopic to the concatenation of $\overline{v_{i-1}v_i}$ for all $1\leq i\leq t$.
By Lemma~\ref{obser:rectangle}, for each $1\leq i\leq t$, $G(\calV)$ contains a staircase path connecting $v_{i-1}$ and $v_i$ and the path is in $R_{v_{i-1}v_i}$ (and thus is homotopic to $\overline{v_{i-1}v_i}$). Therefore, by concatenating the staircase paths from $v_{i-1}$ to $v_i$ for all $i=1,2,\ldots,t$, we obtain a staircase path from $p'$ to $q'$ and the path is homotopic to $\pi'(p',q')$. Note that the staircase path has the same length as $\pi'(p',q')$ since $\pi'(p',q')$ is an L-shaped path (and thus is also a shortest path).
The above claim thus follows.

\begin{figure}[t]
\begin{minipage}[t]{\linewidth}
\begin{center}
\includegraphics[totalheight=1.2in]{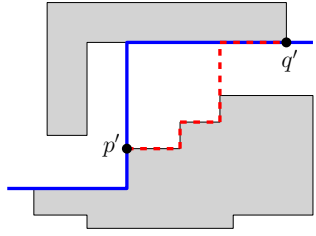}
\caption{\footnotesize Converting $\pi_G(p',q')$ (the dashed red path) to $\pi'(p',q')$ (the solid blue path between $p'$ and $q'$).
}
\label{fig:Lshaped}
\end{center}
\end{minipage}
\vspace*{-0.15in}
\end{figure}

This proves Lemma~\ref{lem:homo}. The proof actually constructs the path $\pi_G$ in $G(\calV)$ corresponding to the optimal path $\pi$, and $\pi_G$ is called a {\em target path}. 
Yang et al.~\cite{ref:YangRe95} also showed that $\pi$ can be obtained from $\pi_G$ by applying certain {\em dragging} operations during searching the graph $G(\calV)$. Before describing the details of the operation, we first give some intuition on how $\pi$ can be obtained from $\pi_G$. Based on the above constructive proof for Lemma~\ref{lem:homo}, we only need to show that for each L-shaped path $\pi'(p',q')$, it can be obtained from the corresponding staircase path $\pi_G(p',q')$ in $G(\calV)$. Without loss of generality, we assume that $q'$ is to the northeast of $p'$ and the segment incident to $p'$ in $\pi'(p',q')$ is vertical. Because $\pi_G(p',q')$ is homotopic to $\pi'(p',q')$, we can convert $\pi_G(p',q')$ to $\pi'(p',q')$ as follows (e.g., see Fig.~\ref{fig:Lshaped}). Starting from $p'$, for each horizontal segment of  $\pi_G(p',q')$, drag it upwards until either it hits the horizontal segment of $\pi'(p',q')$ or it becomes collinear with the next horizontal segment of $\pi_G(p',q')$. In the former case, we have obtained $\pi'(p',q')$. In the latter case, we continue to drag the new horizontal segment upwards in the same way as before.

In the sequel, we briefly review the dragging queries~\cite{ref:YangRe95}. This will make our paper self-contained and also help us to explain our new algorithm as well as the optimal path queries given later.

The YLW algorithm intends to search a target path in $G(\calV)$. The algorithm starts from $s$. When a vertex $p$ of $G(\calV)$ is processed, at most eight paths from $s$ to $p$ are stored at $p$ such that their last pushed staircase subpaths containing $p$ are different. Later the algorithm will advance these paths from $p$ to each neighboring vertex $q$ of $p$ in $G(\calV)$. Let $\pi(s,p)$ be such a path stored at $p$ and we want to advance it from $p$ to $q$ to obtain a path $\pi(s,q)$ from $s$ to $q$. Without loss of generality, we assume $p$ is to the northeast of $p'$, where $p'$ is the start point of the last staircase path of $\pi(s,p)$ containing $p$, and we also assume that the last segment $w$ of $\pi(s,p)$ is horizontal (i.e., $w$ is incident to $p$). Other cases are similar. Let $\pi'(s,q)=\pi(s,p)\cup\overline{pq}$. We obtain $\pi(s,q)$ from $\pi(s,q')$ by a dragging operation on $w$ as follows.
We say that $w$ is {\em fixed} if it borders an obstacle that is above $w$, in which case $w$ cannot be dragged upwards anymore, and is {\em floating} otherwise.
Since $\overline{pq}$ is an edge of $G(\calV)$, it is either vertical or horizontal.

\begin{figure}[t]
\begin{minipage}[t]{\linewidth}
\begin{center}
\includegraphics[totalheight=1.0in]{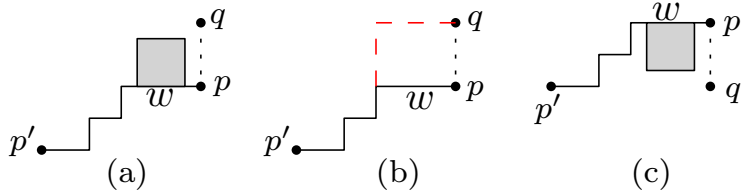}
\caption{\footnotesize Illustrating the dragging operations. The horizontal segment incident to $p$ is $w$.
}
\label{fig:drag}
\end{center}
\end{minipage}
\vspace*{-0.15in}
\end{figure}

\begin{enumerate}
\item
If $q$ is to the right of $p$, then $\pi(s,q)=\pi'(s,q)$ and the last segment of $\pi(s,q)$ is $w\cup\overline{pq}$.

\item
If $q$ is to the left of $p$, then we ignore the path $\pi(s,q)$ (i.e., the path will not be considered).


\item
If $q$ is above $p$ (i.e., $\overline{pq}$ is vertical), then depending on whether $w$ is fixed there are two subcases.

\begin{enumerate}
\item
If $w$ is fixed (e.g., see Fig.~\ref{fig:drag}(a)), then $\pi(s,q)=\pi'(s,q)$ and $\overline{pq}$ becomes the last segment of $\overline{pq}$. Note that $\overline{pq}$ is fixed if and only if it borders an obstacle that is on its right side.

\item
If $w$ is floating, then there are further two subcases.

\begin{enumerate}
\item
If $w$ can be dragged upwards to $q$ without hitting a point of $\calV$ (e.g., see Fig.~\ref{fig:drag}(b)), then we drag $w$ to $q$ and obtain $\pi(s,q)$, which has the dragged $w$ as its last segment.
\item
Otherwise, we ignore the path $\pi(s,q)$.

\end{enumerate}

\end{enumerate}
\item
If $q$ is below $p$, then there are further two subcases.
\begin{enumerate}
\item
If $w$ borders an obstacle below it (e.g., see Fig.~\ref{fig:drag}(c)), then a new U-shaped path is generated and a new pushed staircase subpath is also generated with $w$ being the first segment. We have $\pi(s,q)=\pi'(s,q)$ with $\overline{pq}$ as the last segment.

\item
Otherwise, we ignore the path $\pi(s,q)$.
\end{enumerate}
\end{enumerate}

If we apply the dragging operations on a target path from $s$ to $t$, then an optimal path can be eventually produced. This can be seen from the intuition we discussed earlier (refer to Lemma 4.6 of \cite{ref:YangRe95} for details). This motivates the YLW algorithm, which we describe below on finding  a  minimum-link shortest \st\ path (other two types of optimal paths are similar).

The YLW algorithm works by applying Dijkstra's algorithm according to the measure vector $(L_1(\pi),L_d(\pi))$ for a path $\pi$. Initially, all vertices of $G(\calV)$ are in a priority queue $Q$ with measure vectors $(\infty,\infty)$ except that the measure vector for $s$ is $(0,0)$. As long as $Q$ is not empty, the algorithm removes from $Q$ the vertex $p$ with the smallest measure vector (lexicographically, i.e., for two vectors $(a_1,b_1)$ and $(a_2,b_2)$, the first one is smaller than the second one if and only if $a_1<a_2$, or $a_1=a_2$ and $b_1<b_2$) and advance the paths stored at $p$ to each of $p$'s neighbor $q$ by using the dragging operations. Let $\pi(s,q)$ be a path obtained for $q$. There may be other paths that are already stored at $q$ and the types of the last staircase subpaths of these paths are also stored (recall that there are eight types of pushed staircase subpaths). The YLW algorithm relies on the following two rules to determine whether the new obtained path $\pi(s,q)$ should be stored at $q$, and if yes, whether some paths stored at $q$ should be removed. Let $\pi'(s,q)$ be any path that has already been stored at $q$.

\begin{enumerate}
\item[Rule($a$)]
If the measure vectors of $\pi(s,q)$ and $\pi'(s,q)$ are not the same, then  discard the one whose measure vector is strictly larger.

\item[Rule($b$)]
If $\pi(s,q)$ and $\pi'(s,q)$ have the same measure vector and of the same type, compare their last segments. If their last segments overlap, discard the path whose last segment is longer.
\end{enumerate}

It is claimed in \cite{ref:YangRe95} that once the point $t$ is processed, among all paths stored at $t$, the one with the smallest measure vector is an optimal \st\ path.

\begin{figure}[t]
\begin{minipage}[t]{\linewidth}
\begin{center}
\includegraphics[totalheight=1.3in]{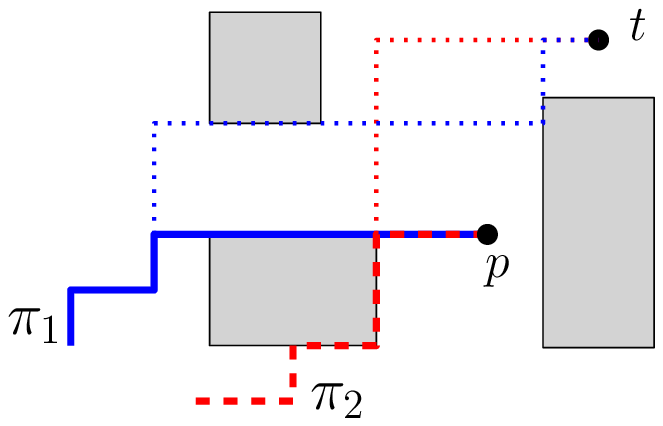}
\caption{\footnotesize Illustrating a counter example for the YLW algorithm.
}
\label{fig:counterex}
\end{center}
\end{minipage}
\vspace*{-0.15in}
\end{figure}

We find that the algorithm is not correct, mainly due to Rule($a$). Figure~\ref{fig:counterex} illustrates a counterexample. Assume that both $\pi_1$ and $\pi_2$ are paths from $s$ to $p$ with $L_1(\pi_1)=L_1(\pi_2)$ and $L_d(\pi_1)+1=L_d(\pi_2)$. Thus, the measure vector of $\pi_1$ is strictly smaller than that of $\pi_2$. According to Rule($a$), we should discard $\pi_2$.
Observe that we can obtain an \st\ path from $s$ to $t$ using $\pi_2$
without having any extra link. However, to obtain an \st\ path using $\pi_1$, we need at least two more links. Therefore, $\pi_2$ can lead to a better \st\ path than $\pi_1$, and thus, we should not discard $\pi_2$.
Notice that the reason this happens is that although the measure vector of $\pi_1$ is strictly smaller than that of $\pi_2$, the last segment of $\pi_2$ is shorter than that of $\pi_1$ (and thus it may be ``freely'' dragged upwards higher than that of $\pi_1$).

In fact, the most essential reason for this error to happen might be
the following. If $\pi$ is a shortest \st\ path, then for any two
points $p$ and $q$ in $\pi$, the subpath of $\pi$ between $p$ and $q$
is also a shortest path from $p$ to $q$. However, this may not be the
case for minimum-link paths. Namely, if $\pi$ is a minimum-link \st\
path, then it is possible that for two points $p$ and $q$ in $\pi$,
the subpath of $\pi$ between $p$ and $q$ is not a minimum-link path
from $p$ to $q$. Due to this reason, one can verify that the $O(nr +
n\log n)$ time algorithm given by Yang et al.~\cite{ref:YangOn92} for
computing optimal paths is not correct either. Indeed, the approach in \cite{ref:YangOn92} also applies Dijkstra's algorithm on a graph to search the optimal paths using the measure vectors like $(L_1(\pi),L_d(\pi))$.

\subsection{Our New Algorithm}
\label{sec:corrected}

To fix the error, we need to fix Rule($a$).
We first consider the minimum-link shortest paths. We replace Rule($a$) by the following Rule($a_1$), but still keep Rule($b$). (Recall that $\pi'(s,q)$ denotes any path that has already been stored at $q$.)

\begin{enumerate}
\item[Rule($a_1$)]
Let $\pi_1$ be one of $\pi'(s,q)$ and $\pi(s,q)$, and $\pi_2$ the other. If  $L_1(\pi_1)<L_1(\pi_2)$, or $L_1(\pi_1)=L_1(\pi_2)$ but $L_d(\pi_1)\leq L_d(\pi_2)-2$, then we discard $\pi_2$.
\end{enumerate}

By Rule($a_1$), we may need to store two paths $\pi_1$ and $\pi_2$ at
$q$ even if the measure vector of one path is strictly smaller than
that of the other, in which case $L_1(\pi_1)=L_1(\pi_2)$ and
$L_d(\pi_1)=L_d(\pi_2)\pm 1$.
Hence, unlike the YLW algorithm, each vertex $q$ of $G(\calV)$ may store paths with different measure vectors. Therefore, we cannot apply the same ``vertex-based'' Dijkstra's algorithm  as before. Instead, we propose a ``path-based'' Dijkstra's algorithm. Roughly speaking, we will process individual paths instead of vertices. Specifically, in the beginning there is only one path from $s$ to $s$ itself in the priority queue $Q$. In general, as long as $Q$ is not empty, we remove from $Q$ the path $\pi$ with the smallest measure vector. Assume that the endpoint of $\pi$ is $p$. Then, we advance $\pi$ from $p$ to each of $p$'s neighbors $q$. If $\pi(s,q)$ is stored at $q$ by our rules (i.e., both Rule($a_1$) and Rule($b$)), then we (implicitly) insert $\pi(s,q)$ to $Q$. The algorithm stops once $Q$ is empty. Since we process paths following the increasing measure order, the algorithm will eventually stop. Finally, among all paths stored at $t$, we return the one with the smallest measure as the optimal solution.

We will prove the correctness of the algorithm in Section~\ref{sec:correctold}. In terms of the running time, the YLW algorithm maintains at most eight paths at each vertex $p$ of $G(\calV)$. To see this, due to the Rule($b$), for each type of staircase paths, $p$ maintains at most one path. In our new algorithm, the paths maintained at $p$ always have the same length but their link distances differ by at most one. Hence, again due to Rule($b$), there are at most sixteen paths maintained at $p$. Clearly, this does not affect both the time and the space complexities of the algorithm asymptotically. Thus, the algorithm still runs in $O(n\log^2 n)$ time and $O(n\log n)$ space, as the YLW algorithm.

In addition, using another path-preserving graph $G^*(\calV)$ of
$O(n\log^{1/2}n)$  vertices and $O(n\log^{3/2}n)$ edges~\cite{ref:ClarksonRe88}, Yang et al.~\cite{ref:YangRe95} proposed
another $O(n\log^{3/2} n)$ time and space algorithm (see Section 4.2 of
\cite{ref:YangRe95}). Further, Chen et al.~\cite{ref:ChenOn01} reduced the space of the algorithm to $O(n\log n)$ with the same $O(n\log^{3/2} n)$ time (similar technique was also used in \cite{ref:ChenSh00}).
By applying the techniques of both \cite{ref:YangRe95} and
\cite{ref:ChenOn01} to our new method, we can also obtain an algorithm of $O(n\log^{3/2} n)$ time and $O(n\log n)$ space. We omit the details.

We proceed on the problem of finding a minimum-cost \st\ path. Recall that we have a cost function $f$. For any path $\pi$, we use $f(\pi)$ to denote the cost of the path. 
Our algorithm is the same as above with the following changes. First,
the paths $\pi$ in the priority $Q$ are prioritized by $f(\pi)$. Second, we replace both Rule($a_1$) and Rule($b$) by the following rule.

\begin{enumerate}
\item[Rule($a_2$)]
Let $\pi_1$ be one of $\pi'(s,q)$ and $\pi(s,q)$, and $\pi_2$ the other.
If the last segments of $\pi_1$ and $\pi_2$ are exactly the same and $f(\pi_1)\leq f(\pi_2)$, then we discard $\pi_2$.
\end{enumerate}

\begin{figure}[t]
\begin{minipage}[t]{\linewidth}
\begin{center}
\includegraphics[totalheight=1.5in]{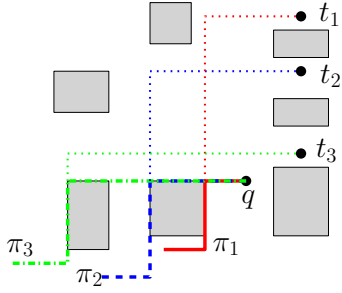}
\caption{\footnotesize Illustrating an example on why we need Rule($a_2$).
}
\label{fig:threepaths}
\end{center}
\end{minipage}
\vspace*{-0.15in}
\end{figure}

We give some intuition on why we use the above rule. Refer to Fig.~\ref{fig:threepaths}, where there are three paths $\pi_1$, $\pi_2$, and $\pi_3$ from $s$ to $q$. Let $s_i$ be the last segment of $\pi_i$ for each $1\leq i\leq 3$, and we assume that they overlap with $|s_1|<|s_2|<|s_3|$, where $|s_i|$ is the length of each $s_i$. We also assume that $L_d(\pi_1)=L_d(\pi_2)=L_d(\pi_3)$ and $L_1(\pi_1)>L_1(\pi_2)>L_1(\pi_3)$. In this case, we have to keep all three paths because any of them may lead to the best path from $s$ to $t$. For example, for each $1\leq i\leq 3$, the path $\pi_i$ may lead to the best path from $s$ to $t_i$. One can generalize the example so that a total of $\Omega(n)$ paths may need to be stored at $p$. However, $O(n)$ is the upper bound since the last segment of each such path starts from a different vertex of $G(\calV)$ in the horizontal line through $q$ and there are $O(n)$ such vertices. For this reason, their are $O(n^2\log n)$ paths stored in all $O(n\log n)$ vertices of $G(\calV)$. Hence, the running time of the algorithm becomes $O(n^2\log^2 n)$ and the space becomes $O(n^2\log n)$.
One may want to use some other rules to reduce the number of paths that need to be stored at $q$, e.g., Rule($b$); however, in the worst case, the number of paths stored at $q$ is still $\Theta(n)$.

We point out a detail about the algorithm implementation. Suppose we
have computed a new path $\pi(s,q)$ at $q$ and we want to apply
Rule$(a_2)$. Then, we need to know whether there is a path $\pi'(s,q)$
already stored at $q$ whose last segment is exactly the same as that
of $\pi(s,q)$. If we check every path stored at $q$, then this step
would cost $O(n)$ time, resulting in an overall $O(n^3\log n)$ time.
We can actually implement this step in $O(1)$ time, as follows. Let
$\overline{pq}$ be the last segment of $\pi(s,q)$ and suppose
$\overline{pq}$ is horizontal. Observe that $p$ must be a vertical
projection of a point in $\calV$ on the horizontal line through $q$.
We use an array $A$ of size $|\calV|$ such that $A[i]$ corresponds to
the $i$-th vertex of $\calV$ in the order of increasing
$x$-coordinate. Hence, if $p$ is the projection of the $i$-th vertex
of $\calV$, then we can simply check the path stored at $A[i]$ in
$O(1)$ time, and if $\pi(s,p)$ should be stored, then we simply store it at $A[i]$. Similarly, we also build another array for the horizontal projections of the vertices of $\calV$ on the vertical line through $q$. In this way, the overall running time of the algorithm is $O(n^2\log^2 n)$. The space complexity is still $O(n^2\log n)$ because the total size of the arrays at each vertex of the graph is $O(n)$.

Further, as for the minimum-link shortest paths, by using the graph $G^*(\calV)$ and the techniques in \cite{ref:ChenOn01,ref:YangRe95}, we can reduce the running time by a  factor of $\sqrt{\log n}$. 
We omit the details.

For computing a shortest minimum-link \st\ path, we use the same
algorithm as above for the minimum-cost paths but with the following
changes. First, we use the measure vector $(L_d(\pi),L_1(\pi))$
instead. Second, we use the following rule, which is similar to
Rule($a_2$).

\begin{enumerate}
\item[Rule($a_3$)]
Let $\pi_1$ be one of $\pi'(s,q)$ and $\pi(s,q)$, and $\pi_2$ the other.
If the last segments of $\pi_1$ and $\pi_2$ are exactly the same and the measure vector of $\pi_1$ is no larger than that of $\pi_2$, then we discard $\pi_2$.
\end{enumerate}

The time and space complexities are the same as the above for the
minimum-cost paths.

If we are looking for a minimum-link \st\ path (without considering the length), then we can use the following rule. 

\begin{enumerate}
\item[Rule($a_4$)]
Let $\pi_1$ be one of $\pi'(s,q)$ and $\pi(s,q)$, and $\pi_2$ the other.
If $L_d(\pi_1)\leq L_d(\pi_2)-2$, then we discard $\pi_2$. We also  discard $\pi_2$ if the following is true: $L_d(\pi_1)=L_d(\pi_2)$, the last segments of $\pi_1$ and $\pi_2$ overlap, and the last segment of $\pi_1$ is no longer than that of $\pi_2$.
\end{enumerate}

The rule makes sure that we only need to keep at most eight paths at
any vertex $q$ of $G(\calV)$: for each of the following four
directions of $q$: left, right, above, below, there are two paths
whose last segments are from that direction and their link distances differ by at most one.
Hence, similar to the minimum-link shortest paths, we can find a
minimum-link path in $O(n\log^{3/2}n)$ time and $O(n\log n)$ space. As
discussed in Section~\ref{sec:intro}, the problem of finding a
single minimum-link path and its one-point query problem have been solved
optimally~\cite{ref:MitchellAn15} (after $\calP$ is triangulated).
We discussed the above result mainly because we will use it to answer the
two-point queries in Section~\ref{sec:twopoint}.

The correctness of all above algorithms is proved in
Section~\ref{sec:correctold}.

\subsection{The Correctness of Our Algorithm}
\label{sec:correctold}

We first show the correctness of the algorithm for computing a
minimum-link shortest path. The analysis for other paths is very similar.

Let $\pi(s,t)$ be a minimum-link shortest \st\ path in $\calP$. Let
$\pi_{G(\calV)}(s,t)$ be the corresponding target path from $s$ to $t$ in the
graph $G(\calV)$.
For any vertex $p$ in the target path, let $\pi(s,p)$ be the path in
$\calP$ from $s$ to $p$ obtained by applying the dragging
operations on the subpath of $\pi_{G(\calV)}(s,t)$ from $s$ to $p$. To
prove the correctness of our algorithm, it is sufficient to show that
the paths of $\pi(s,p)$ for the vertices $p$ of $\pi_{G(\calV)}(s,t)$
from $s$ to $t$ will be computed and advanced following the vertex
order of  $\pi_{G(\calV)}(s,t)$ during our algorithm.
According to our analysis before, we only need to prove it for any
L-shaped subpath $\pi(p,q)$ between two adjacent vertices $p$ and $q$ of
$\pi(s,t)$.

We assume the path $\pi(s,p)$ has been computed and stored at $p$, and is about to
advance. Initially this is trivially true when $p=s$. Let $\pi_{G(\calV)}(p,q)$
be the subpath of $\pi_{G(\calV)}(s,t)$ between $p$ and $q$, and let
$p=v_0,v_1,\ldots,v_k=q$ be the vertices of $\pi_{G(\calV)}(p,q)$ in
order from $p$ to $q$. Recall that $\pi_{G(\calV)}(p,q)$ is a
staircase path. Without loss of generality, we assume $q$ is to the
northeast of $p$. If the path $\pi(s,v_1)$ is stored at $v_1$, then
our algorithm is correct. Otherwise, there must be a path
$\pi'(s,v_1)$ stored at $v_1$ that causes $\pi(s,v_1)$ not to be
stored. According to Rules $(a_1)$ and $(b)$, at least one of the
following cases must happen: (1) $L_1(\pi'(s,v_1))<L_1(\pi(s,v_1))$; (2)
$L_1(\pi'(s,v_1))=L_1(\pi(s,v_1))$ but $L_d(\pi'(s,v_1))\leq
L_d(\pi(s,v_1))-2$; (3) the measure vectors of the two paths are exactly the same, the last
staircase subpaths of both paths are of the same type, and the last
segment of $\pi'(s,v_1)$ is shorter than or equal to that of $\pi(s,v_1)$.

If Case (1) happens, then consider the following path $\pi'(s,t)$ from $s$ to $t$ (e.g., see Fig.~\ref{fig:correct}): the concatenation of $\pi'(s,v_1)$, a vertical segment from $v_1$ to a point $v_1'$ on the horizontal segment of the L-shaped subpath $\pi(p,q)$, and the subpath $\pi(v_1',t)$ of $\pi(s,t)$ between $v_1'$ and $t$. Note that $L_1(\pi(s,t))=L_1(\pi(s,v_1))+|\overline{v_1v_1'}|+L_1(\pi(v_1',t))$. Because $L_1(\pi'(s,v_1))<L_1(\pi(s,v_1))$ (i.e., Case (1)) and $L_1(\pi'(s,t))=L_1(\pi'(s,v_1))+|\overline{v_1v_1'}|+L_1(\pi(v_1',t))$, we obtain
$L_1(\pi'(s,t))< L_1(\pi(s,t))$, contradicting with that $\pi(s,t)$ is a shortest path.

\begin{figure}[t]
\begin{minipage}[t]{\linewidth}
\begin{center}
\includegraphics[totalheight=1.3in]{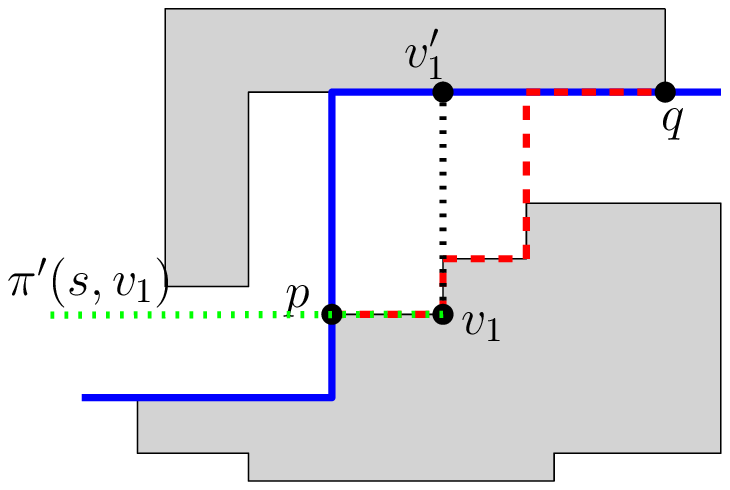}
\caption{\footnotesize Illustrating the definition of $v_1'$. The blue
solid path is $\pi(s,t)$, the red dashed path is
$\pi_{G(\calV)}(p,q)$, and the green dotted path is $\pi(s,v'_1)$.
}
\label{fig:correct}
\end{center}
\end{minipage}
\vspace*{-0.15in}
\end{figure}

If Case (2) happens, then we still consider the path $\pi'(s,t)$ obtained above. Observe that $L_d(\pi(s,t))\geq L_d(\pi(s,v_1))+L_d(\pi(v_1',t))-1$, where the minus 1 is due to that we may be able to drag the last segment of $\pi(s,v_1)$ so that it overlaps with the first segment $\pi(v_1',t)$ (and thus save one link). Also note that $L_d(\pi'(s,t))\leq L_d(\pi'(s,v_1))+L_d(\pi(v_1',t))+1$, where the plus 1 is due to the segment $\overline{v_1v_1'}$. Since $L_d(\pi'(s,v_1))\leq L_d(\pi(s,v_1))-2$ (i.e., Case (2)), we obtain that $L_d(\pi'(s,t))\leq L_d(\pi(s,t))$. Further, because $L_1(\pi'(s,v_1))=L_1(\pi(s,v_1))$, we obtain $L_1(\pi'(s,t))= L_1(\pi(s,t))$. This implies that using $\pi'(s,v_1)$ we can also obtain a minimum-link shortest \st\ path, and thus, $\pi(s,v_1)$ can be safely ignored.

If Case (3) happens, then similar to the proof of Lemma 4.7 in \cite{ref:YangRe95}, using  $\pi'(s,v_1)$ we can also obtain a minimum-link shortest \st\ path, and thus, $\pi(s,v_1)$ can be safely ignored.

The above proves that in any case our algorithm stores necessary paths at $v_1$ that can be used to eventually obtain a minimum-link shortest \st\ path. By the similar argument, we can show that this is true for $v_i$ for all $i=2,3,\ldots,k$. This establishes the correctness of our algorithm.

We proceed to show the correctness of our algorithm for computing a minimum-cost \st\ path. We follow the above analysis scheme and focus on proving that the path $\pi(s,v_1)$ will be stored at $v_1$ if necessary.
If $\pi(s,v_1)$ is not stored at $v_1$, then according to Rules $(a_2)$, this only happens because there is another path $\pi'(s,v_1)$ stored at $v_1$ such that the last segments of $\pi(s,v_1)$ and $\pi'(s,v_1)$ are exactly the same and $f(\pi'(s,v_1))\leq f(\pi(s,v_1))$.

First of all, we know that the last segment of $\pi(s,v_1)$ (i.e.,
$\overline{pv_1}$) is horizontal and we can drag it upwards freely
until the horizontal segment $e$ of the L-shaped path $\pi(p,q)$ to
obtain the path $\pi(s,t)$ (i.e., by concatenating with $\pi(v_1',t)$).
Since the last segments of $\pi(s,t)$ and
$\pi'(s,t)$ are exactly the same, regardless of whether $\pi(s,t)$ and
$\pi'(s,t)$ are of the same type, we can also drag the last segment of
$\pi'(s,v_1)$ upwards freely until $e$, so that we can obtain another
\st\ path $\pi'(s,t)$. Further, the above dragging on the last segment
of $\pi'(s,t)$ does not introduce any extra link and the amount of
length it introduces is the same as that introduced by dragging the
last segment of $\pi(s,v_1)$. As $f(\pi'(s,v_1))\leq f(\pi(s,v_1))$
and the cost function $f$ is non-decreasing in both the length and the
link distance of the
path, we can obtain that $f(\pi'(s,t))\leq f(\pi(s,t))$. Hence, we can
also obtain a minimum-cost \st\ path by using $\pi'(s,v_1)$, and thus
$\pi(s,v_1)$ can be safely ignored without being stored at $v_1$.
This establishes the correctness of the algorithm.

The correctness of our algorithm for computing shortest minimum-link paths follows the similar analysis as the above case for minimum-cost paths. We omit the details.

Finally, we show the correctness for computing a minimum-link \st\ path. We again follow the above scheme. If the path $\pi(s,v_1)$ is not stored at $v_1$, then according to Rule $(a_4)$, there must be another path $\pi'(s,v_1)$ stored at $v_1$ such that one of the following two cases happens: (1) $L_d(\pi'(s,v_1))\leq L_d(\pi(s,v_1))-2$; (2) $L_d(\pi'(s,v_1))= L_d(\pi(s,v_1))$, the last segments of both paths overlap, and the last segment of $\pi'(s,v_1)$ is no longer than that of the last segment of $\pi(s,v_1)$.

If Case (1) happens, then as in the analysis for minimum-link shortest paths, we consider the path $\pi'(s,t)$ obtained from $\pi'(s,v_1)$ by adding a vertex segment $\overline{v_1v_1'}$. We have shown above that $L_d(\pi'(s,t))\leq L_d(\pi(s,t))$ and thus it is safe to ignore $\pi(s,v_1)$. If Case (2) happens, we can follow the proof of Lemma 4.7 of \cite{ref:YangRe95} (or the similar analysis as the above for the minimum-cost paths) to show that $\pi'(s,v_1)$ can also lead to a minimum-link \st\ path.


\section{The Improved Algorithm}
\label{sec:new}

In this section, we improve our algorithm proposed in
Section~\ref{sec:old}, so that  in addition to $O(n)$, the
complexities of our improved algorithm only depend on $h$, i.e., the number
of holes of $\calP$. We first review the corridor structure of
$\calP$~\cite{ref:MitchellAn15} and the histogram partitions of
rectilinear simple polygons~\cite{ref:SchuiererAn96}.

\subsection{The Corridor Structure of $\calP$}

For ease of exposition, we make a general position assumption that no two edges of $\calP$ are collinear.
The {\em vertical visibility decomposition} of $\calP$, denoted by
$\vd(\calP)$, is obtained by extending each vertical edge of $\calP$
until it hits the boundary of $\calP$ (e.g., see Fig.~\ref{fig:vd}).
Each cell of $\vd(\calP)$ is a rectangle.
Each extension segment is called a {\em diagonal} of $\vd(\calP)$.

\begin{figure}[t]
\begin{minipage}[t]{0.49\linewidth}
\begin{center}
\includegraphics[totalheight=1.7in]{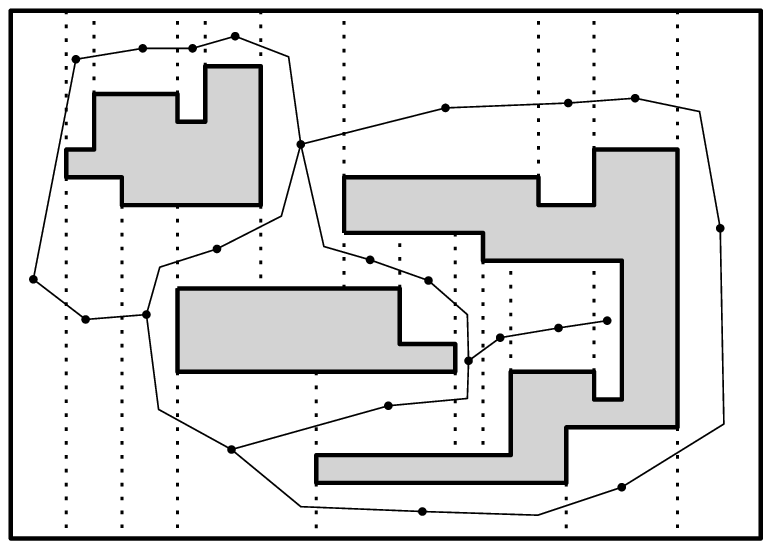}
\caption{\footnotesize Illustrating the vertical visibility decomposition $\vd(\calP)$ and its dual graph $G_{vd}$.
}
\label{fig:vd}
\end{center}
\end{minipage}
\hspace{0.05in}
\begin{minipage}[t]{0.49\linewidth}
\begin{center}
\includegraphics[totalheight=1.7in]{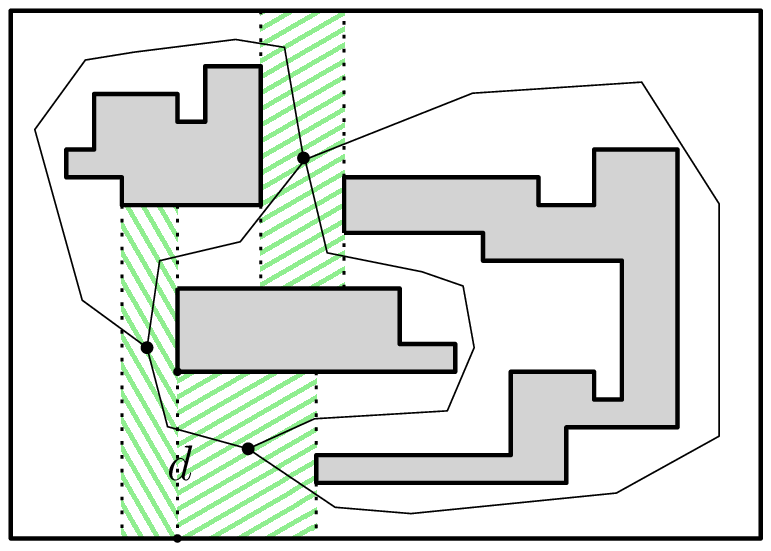}
\caption{\footnotesize Illustrating the corridor structure and the corridor graph $G_{cor}$ of three vertices. There are three junction rectangles, which are highlighted. Each connected white region is a corridor, which corresponds to an edge of $G_{cor}$. The diagonal $d$ forms a degenerated corridor.
}
\label{fig:reduceG}
\end{center}
\end{minipage}

\vspace*{-0.15in}
\end{figure}

The corridor structure of $\calP$ has been introduced before, e.g., see~\cite{ref:MitchellAn15}.
Let $G_{vd}$ be the dual graph of $\vd(\calP)$ (e.g., see Fig.~\ref{fig:vd}), i.e., each node of
$G_{vd}$ corresponds to a cell of $\vd(\calP)$ and two nodes have an
edge if the corresponding cells share an edge. Based on $G_{vd}$, we
obtain a {\em corridor graph} $G_{cor}$ as follows. First, we keep
removing every degree-one node from $G_{vd}$ along with its incident
edge until no such nodes remain. Second, we keep contracting every
degree-two node from $G_{vd}$ (i.e., remove the node and replace its
two incident edges by a single edge) until no such nodes remain. The
graph thus obtained is $G_{cor}$, which has $O(h)$ nodes and $O(h)$
edges~\cite{ref:MitchellAn15}. Refer to
Fig.~\ref{fig:reduceG} for an example. The cells of $\vd(\calP)$
corresponding to the nodes of $G_{cor}$ are called {\em junction
rectangles}. If we remove all junction rectangles from $\calP$, each
connected region is a simple rectilinear polygon, which is called a
{\em corridor}. Each corridor has two diagonals each of which is on a vertical side of a junction rectangle, and we call
them the {\em doors} of the corridor. For convenience, if
a diagonal $d$ bounds two junction rectangles (e.g., see
Fig.~\ref{fig:reduceG}), then we consider $d$ itself as a
``degenerate'' corridor whose two doors are both $d$. With the
degenerated corridors, each vertex of $\calP$ lies in a unique corridor.

The decomposition $\vd(\calP)$ can be computed in $O(n+h\log^{1+\epsilon}h)$ time for any $\epsilon>0$~\cite{ref:Bar-YehudaTr94}. After $\vd(\calP)$ is known, the corridor structure of $\calP$ (i.e., computing all corridors and junction rectangles) can be obtained in $O(n)$ time.


\subsection{The Histogram Partitions}

The histogram partition is a decomposition of a simple rectilinear polygon~\cite{ref:SchuiererAn96}. We will need to build the histogram partitions on the corridors of $\calP$. Below we review the partition and we follow the terminologies of \cite{ref:SchuiererAn96}.

A simple rectilinear polygon $H$ is called a {\em histogram} if its
boundary can be divided into an $x$- or $y$-monotone chain and a single line
segment; the single segment is called the base of $H$ (e.g., see
Fig.~\ref{fig:histogram}).

\begin{figure}[t]
\begin{minipage}[t]{0.49\linewidth}
\begin{center}
\includegraphics[totalheight=1.3in]{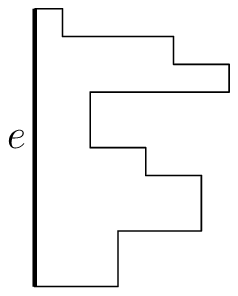}
\caption{\footnotesize Illustrating a histogram with base $e$.
}
\label{fig:histogram}
\end{center}
\end{minipage}
\hspace{0.05in}
\begin{minipage}[t]{0.49\linewidth}
\begin{center}
\includegraphics[totalheight=1.5in]{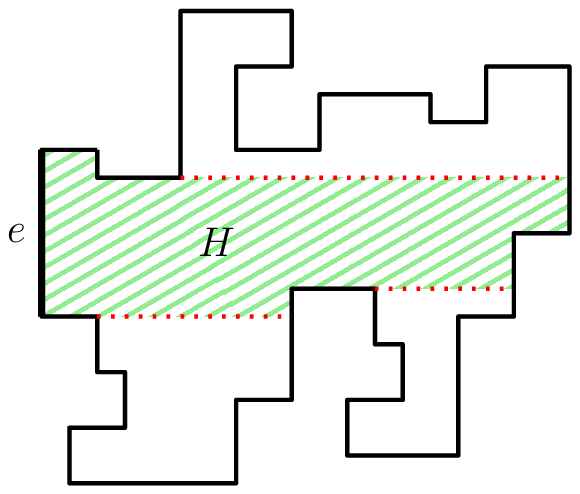}
\caption{\footnotesize Illustrating the maximal histogram $H$, which has three windows shown with (red) dotted segments.
}
\label{fig:histpar}
\end{center}
\end{minipage}

\vspace*{-0.15in}
\end{figure}

Consider a simple rectilinear polygon $Q$ (e.g., a corridor $\calC$ of the corridor structure of $\calP$) and let $e$ be an edge of $Q$ (e.g., a door of $\calC$). A histogram partition of $Q$ with respect to $e$, denoted by $\calH(Q,e)$, is defined as follows. Let $H$ be the {\em maximal histogram} with base $e$ in $Q$, i.e., there is no other histogram in $Q$ with base $e$ that can properly contain it (e.g., see Fig.~\ref{fig:histpar}). A {\em window} of $H$ is a maximal segment on the boundary of $H$ that is contained in the interior of $Q$ except its two endpoints (e.g., see Fig.~\ref{fig:histpar}). For each window $w$ of $H$, it divides $H$ into two subpolygons, and we let $Q(w)$ denote the one that does not contain $e$. If $H$ does not have a window, then we are done with the histogram partition of $Q$. Otherwise, for each window $w$, we perform the above partition on $Q(w)$ recursively with respect to $w$.

For any points $p$ and $q$ in $Q$, it is known that there exists a
path from $p$ to $q$ in $Q$ that is both a shortest path and a
minimum-link
path~\cite{ref:deBergOn91,ref:HershbergerCo94,ref:SchuiererAn96}, and
we call it a {\em smallest path}.


\subsection{A Reduced Path Preserving Graph}
\label{sec:reduced}
Recall that our algorithm in Section~\ref{sec:old} use a graph $G(\calV)$, which is built on the vertices of $\calV$ and has $O(n\log n)$ nodes and edges.
In this section, as a major tool for reducing the complexities of our algorithm, we propose a {\em reduced graph} of only $O(h\log h)$ nodes and edges.

The rest of this section is organized as follows. In Section~\ref{sec:backbone}, we introduce a set $\calB$ of {\em backbone points}, based on which we will define the reduced graph $G(\calB)$ in Section~\ref{sec:reduce}. In Section~\ref{sec:graphconstruct}, we compute $G(\calB)$. In Section~\ref{sec:compute}, we give an algorithm to compute optimal paths by using $G(\calB)$, and Section~\ref{sec:correct} proves its correctness. The algorithm in Section~\ref{sec:compute} is for the special case where both $s$ and $t$ are in junction rectangles. Section~\ref{sec:general} generalizes the approach to other cases.

\subsubsection{The Backbone Points}
\label{sec:backbone}

We introduce a set $\calB$ of $O(h)$ {\em backbone} points on
the doors of the corridors of $\calP$, which will be used to define our reduced graph later.

Consider a corridor $\calC$ of the corridor structure of $\calP$. Let $d_1$ and $d_2$ be the two doors of $\calC$. Note that both doors are vertical. The region of $\calC$ excluding the two doors is called the {\em interior} of $\calC$. If there exist a point $p_1\in d_1$ and a point $p_2\in d_2$ such that $\overline{p_1p_2}$ is horizontal and $\overline{p_1p_2}$ in $\calC$ then we say that $\calC$ is an {\em open corridor}; otherwise, it is {\em closed} (e.g., see Fig.~\ref{fig:opencan} and Fig.~\ref{fig:close}).

\begin{figure}[t]
\begin{minipage}[t]{0.49\linewidth}
\begin{center}
\includegraphics[totalheight=1.3in]{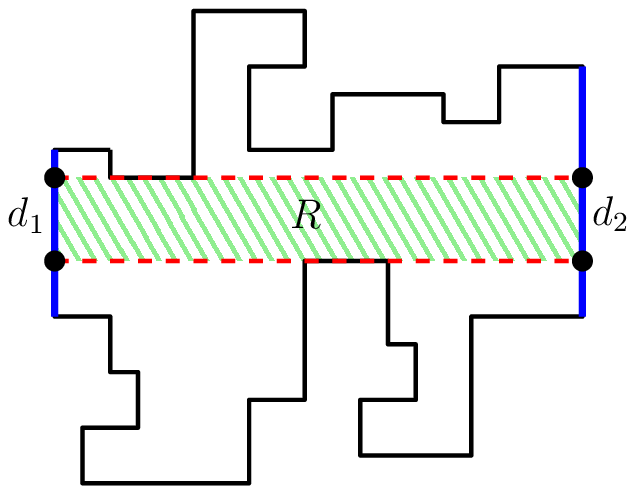}
\caption{\footnotesize Illustrating an open corridor: the canal $R$ and the two bridges are highlighted. The four points on the two doors are backbone points.
}
\label{fig:opencan}
\end{center}
\end{minipage}
\hspace{0.05in}
\begin{minipage}[t]{0.49\linewidth}
\begin{center}
\includegraphics[totalheight=1.3in]{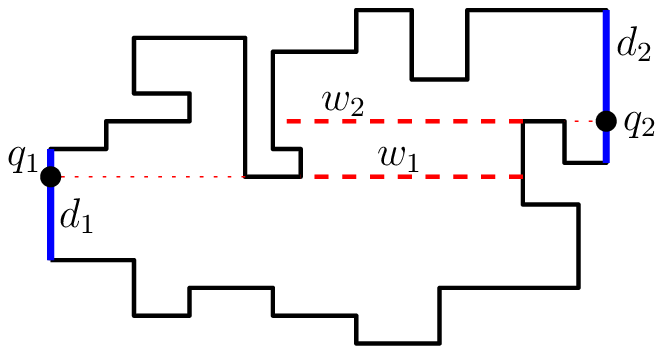}
\caption{\footnotesize Illustrating a closed corridor. The points $q_1$ and $q_2$ are backbone points on $d_1$ and $d_2$, respectively.
}
\label{fig:close}
\end{center}
\end{minipage}
\vspace*{-0.15in}
\end{figure}

Consider an open corridor $\calC$ (e.g., see Fig.~\ref{fig:opencan}). Let $p_1$ and $p_2$
be the points defined above. Imagine that we drag
$\overline{p_1p_2}$ vertically upwards (resp., downwards) until we hit
a vertex of $\calC$, then the current locations of $p_1$ and $p_2$ are
two {\em backbone points}. In this way, each door of $\calC$ has two
backbone points. Clearly, the rectangle $R$ with the four backbone
points as the vertices is in $\calC$ and we call $R$ the {\em canal}
of $\calC$. The two horizontal edges of $R$ are called {\em bridges}
of $\calC$. Further, the top edge of $R$ is the {\em upper} bridge and
the bottom edge is called the {\em lower} bridge.

If $\calC$ is a degenerate corridor, which is a single
diagonal $d$, then $\calC$ is also an open corridor and the upper
(resp., lower) bridge is degenerated to the upper (resp., lower) endpoint of $d$.


We have the following self-evident observation on open corridors.

\begin{observation}\label{obser:open}
Suppose $d_1$ and $d_2$ are the two doors of an open corridor $\calC$.
Consider any two points $p_1\in d_1$ and $p_2\in d_2$.
\begin{enumerate}
\item
If both $p_1$ and $p_2$ are on the boundary of the canal $R$, then $\overline{p_1p}\cup\overline{pp_2}$ is a shortest path in $\calC$ from $p_1$ to $p_2$, where $p$ is the horizontal projection of $p_2$ on $d_1$. 

\item
Otherwise, $\overline{p_1q_1}\cup \overline{q_1q_2}\cup \overline{q_2p_2}$ is a shortest path in $\calC$ from $p_1$ to $p_2$ for some bridge $\overline{q_1q_2}$ of $\calC$.
\end{enumerate}
In either case, we use $\pi(\calC,p_1,p_2)$ to denote the above shortest path between $p_1$ and $p_2$, and we call it a {\em canonical path}.
\end{observation}

Next, we consider the case where $\calC$ is closed (e.g., see
Fig.~\ref{fig:close}). Let $H_1$ be the maximal histogram in $\calC$
with base $d_1$. As $\calC$ is closed, $H_1$ has a window $w_1$ that
{\em separates} $d_1$ from $d_2$, that is, $w_1$ divides $\calC$ into two sub-polygons that contain $d_1$ and $d_2$, respectively. By the definition of windows, if we extend $w_1$ to $d_1$, the extension will hit $d_1$ at a point, denoted by $q_1$, before it goes out of $\calC$. Similarly, we define $H_2$, $w_2$, and $q_2$, with respect to the other door $d_2$. The two points $q_1$ and $q_2$ are {\em backbone points} of $\calC$.
The following is proved in \cite{ref:SchuiererAn96}.

\begin{observation}\label{obser:closed}{\em (Lemma 3.1 of \cite{ref:SchuiererAn96})}
Suppose $d_1$ and $d_2$ are the two doors of a closed corridor $\calC$.
For any two points $p_1\in d_1$ and $p_2\in d_2$, the concatenation of
$\overline{p_1q_1}$, a shortest path from $q_1$ to $q_2$ in $\calC$, and $\overline{q_2p_2}$ is a shortest path from $p_1$ to $p_2$ in $\calC$. We use $\pi(\calC,p_1,p_2)$ to denote the path, and we call it a {\em canonical path}.
\end{observation}

The above defines two backbone points on each door of every open corridor and one backbone point on each door of every closed corridor. Let $\calB$ denote the set of all such  backbone points. Since there are $O(h)$ corridors, the size of $\calB$ is $O(h)$.

\subsubsection{The Reduced Graph $G(\calB)$}
\label{sec:reduce}

In the following, we introduce the reduced graph, denoted by
$G(\calB)$, and we will use it to compute optimal paths.
We first consider the case where both $s$ and $t$ are in junction rectangles.
With a little abuse of notation, we let $\calB$ also contain both $s$ and $t$.

We build the graph $G(\calB)$ with respect to the points of $\calB$ in the
same way as $G(\calV)$ with respect to $\calV$ in
Section~\ref{sec:old}. Hence, $G(\calB)$ has $O(h\log h)$ vertices and
$O(h\log h)$ edges. In addition, we add the following $O(h)$ edges to
$G(\calB)$. Consider a closed corridor $\calC$ with the two backbone
points $q_1$ and $q_2$ on its two doors. Note that $q_1$ and $q_2$ are
also two vertices in $G(\calB)$. We add to $G(\calB)$ an edge
$e(q_1,q_2)$ to connect $q_1$ and $q_2$ with length equal to
$L_1(\pi(\calC,q_1,q_2))$, i.e., the length of the canonical path
$\pi(\calC,q_1,q_2)$. We call $e(q_1,q_2)$ a {\em corridor edge} of
$G(\calB)$, and call $\pi(\calC,q_1,q_2)$ a {\em corridor path} of $\calC$.
We do this for all closed corridors. This completes the construction
of $G(\calB)$. Since there are $O(h)$ corridors, $G(\calB)$ has $O(h)$
corridor edges. For differentiation, other edges of $G(\calB)$ that
are not corridor edges are called {\em ordinary edges}. Hence, $G(\calB)$ has $O(h\log h)$ edges in total.

Note that every path $\pi_{G(\calB)}$ in $G(\calB)$ corresponds to a path
$\pi$ in $\calP$ with
the same length in the sense that if the path $\pi_{G(\calB)}$
contains a corridor edge, then $\pi$ contains the corresponding
corridor path.
Similar to Lemma~\ref{obser:rectangle}, we have the following observation.

\begin{observation}\label{obser:recreduced}
For any two points $p$ and $q$ in $\calB$, if the rectangle $R_{pq}$
is empty, then $G(\calB)$ contains a staircase path connecting $p$ and
$q$.
\end{observation}

The following lemma is analogous to Lemma~\ref{lem:homo}, but on
the reduced graph $G(\calB)$. It explains why the graph $G(\calB)$ can help to find optimal paths.

\begin{lemma}\label{lem:target}
There exists a path $\pi_{G(\calB)}$ in $G(\calB)$ from $s$ to $t$ that is homotopic to an
optimal \st\ path and the two paths have the same length; we call $\pi_{G(\calB)}$ a {\em target path}.
\end{lemma}
\begin{proof}
Let $\pi$ be an optimal \st\ path in $\calP$. Since both $s$ and $t$ are in junction rectangles, an easy observation is that if $\pi$ contains an interior point of a corridor $\calC$, then $\pi$ must {\em travel through} $\calC$, i.e., $\pi$ enters $\calC$ through one door and leaves $\calC$ through the other.

We assume that $\pi$ travels through some closed corridors since
otherwise the analysis would be similar (but easier).
Consider each
such closed corridor $\calC$ with two doors $d_1$ and $d_2$. Let $q_1$
and $q_2$ be the two backbone points on $d_1$ and $d_2$, respectively. If we
traverse $\pi$ from $s$ to $t$, define $p_1$ to be the last point on $d_1$
we encounter and define $p_2$ to be the first point on $d_2$ we encounter.
Hence, the subpath of $\pi$ between $p_1$ and $p_2$, denoted by
$\pi(p_1,p_2)$, is in $\calC$. We obtain another \st\ path $\pi'$ by
replacing $\pi(p_1,p_2)$ with the canonical path $\pi(\calC, p_1,
p_2)$ in $\calC$. By Observation~\ref{obser:closed}, $\pi(\calC, p_1,
p_2)$ is a shortest path in $\calC$, and thus,
$L_1(\pi')=L_1(\pi)$. Since both $\pi(p_1,p_2)$ and $\pi(\calC, p_1,
p_2)$ are paths in $\calC$, which is simply connected, they are
homotopic to each other. Therefore, $\pi'$ is homotopic to $\pi$.
Note that $\pi'$ contains the two backbone points $q_1$ and $q_2$ and
the subpath of $\pi'$ between $q_1$ and $q_2$ is the corridor path of
$\calC$, which corresponds to a corridor edge $e(q_1,q_2)$ in $G(\calB)$.

We do the above for all such closed corridors that are traveled through by $\pi$. With a little abuse of notation, let $\pi'$ be the new \st\ path. By the above analysis, $L_1(\pi')=L_1(\pi)$ and $\pi'$ is homotopic to $\pi$.
Let $\pi_1$ be a maximal subpath of $\pi'$ that does not contain any corridor path. Note that $\pi_1$ does not contain an interior point of any closed corridor. Let $a$ and $b$ be the two endpoints of $\pi_1$. Clearly, $a$ and $b$ are in $\calB$. Because all corridor paths are in $G(\calB)$, to prove the lemma, it is sufficient to show that there is a path in $G(\calB)$ connecting $a$ and $b$ with the same length as $\pi_1$ and the path is homotopic to $\pi_1$.
We assume that $\pi_1$ travels through at least one open corridor since otherwise the analysis would be similar (but easier).

Suppose $\pi_1$ travels through an open corridor $\calC$. If we traverse on $\pi_1$ from $a$ to $b$, let $p_1$ be the first point and last point of $\pi_1\cap \calC$, respectively. Hence, $p_1$ is on a door of $\calC$ and $p_2$ is on the other door.
Let $\pi_1(p_1,p_2)$ be the subpath of $\pi_1$ between $p_1$ and
$p_2$. We obtain another path $\pi_1'$ from $a$ to $b$ by replacing
$\pi_1(p_1,p_2)$ with the canonical path $\pi(\calC,p_1,p_2)$. Since
$\pi(\calC,p_1,p_2)$ is a shortest path
between $p_1$ and $p_2$ in $\calC$, $L_1(\pi_1)=L_1(\pi_1')$ and
$\pi_1'$ is homotopic to $\pi_1$. By Observation~\ref{obser:open},
$\pi(\calC,p_1,p_2)$ consists of at most one horizontal segment and at most two vertical segments, and further, the two endpoints of the horizontal segment are on the two doors of $\calC$, respectively, and the two vertical segments are on the doors of $\calC$.

We do this for all such open corridors that are traveled through by
$\pi_1$. Let $\pi_2$ denote the new path, which still connects $a$ and
$b$. Based on the above discussion, $L_1(\pi_1)=L_1(\pi_2)$ and
$\pi_2$ is homotopic to $\pi_1$. Further, for each horizontal segment
of $\pi_2$, if it intersects the interior of a corridor $\calC$
(which is necessarily an open corridor), then it must intersect both
doors of $\calC$.

Suppose we traverse $\pi_2$ from $a$ to $b$. If $\pi_2$ intersects a junction rectangle $R$, then let $p_1$ and $p_2$ be the first and last points $\pi_2$ intersecting $R$, respectively. Let $\pi_2(p_1,p_2)$ be the subpath of $\pi_2$ between $p_1$ and $p_2$. We obtain another path $\pi_2'$ from $a$ to $b$ by replacing $\pi_2(p_1,p_2)$ with an L-shaped path connecting $p_1$ and $p_2$, which has the same length as $\pi_2(p_1,p_2)$ and is homotopic to $\pi_2(p_1,p_2)$. Hence, $\pi_2'$ has the same length as $\pi_2$ and is homotopic to $\pi_2$.

We do the above for all such junction rectangles intersected by $\pi_2$, and let $\pi_3$ be the resulting path, which still connects $a$ to $b$. The length of $\pi_3$ is the same as that of $\pi_2$ and $\pi_3$ is homotopic to $\pi_2$. Further, for each vertical segment of $\pi_3$ that is not incident to either $s$ or $t$, it must be on a vertical side of a junction rectangle.

We assume that $\pi_3$ contains a U-shaped subpath since otherwise the analysis would be similar (but easier).
Consider a U-shape subpath of $\pi_3$ with three segments $s_1$, $s_2$, and $s_3$.
As shown in~\cite{ref:YangRe95}, $s_2$ must contain an obstacle edge $e$ of $\calP$ since otherwise we could shorten the path by dragging $s_2$ towards the direction of $s_1$ and $s_3$. Depending on whether $s_2$ is horizontal or vertical, there are two cases.

\begin{enumerate}
\item
If $s_2$ is horizontal, then $e$ must intersect the interior of an open corridor $\calC$. To see this, on the one hand, $e$ cannot be in a closed corridor because $\pi_1$ (and thus $\pi_2$ and $\pi_3$) does not contain an interior point of any closed corridor. On the other hand, the top or bottom side of each junction rectangle only contains a proper subset of an obstacle edge.

Since $s_2$ is a horizontal segment of $\pi_3$ and $e$ (and thus $s_2$) intersects the interior of $\calC$, $s_2$ intersects both doors of $\calC$, say, at two points $p_1$ and $p_2$. Without loss of generality, we assume the obstacle bounded by $e$ is locally above $e$.
Because $e$ intersects the interior of $\calC$ and $e$ is an obstacle edge, we cannot freely move $\overline{p_1p_2}$ in $\calC$ vertically upwards. This implies that $\overline{p_1p_2}$ is the upper bridge of $\calC$ and thus $p_1$ and $p_2$ are two backbone points of $\calC$. We pick either one of $p_1$ and $p_2$, and call it a {\em breakpoint} of $\pi_3$.

\item
If $s_2$ is vertical, since $s_2$ is between $s_1$ and $s_3$, $s_2$ cannot be incident to either $s$ or $t$. Hence, $s_2$ (and thus $e$) must be on a vertical side of a junction rectangle $R$. Further, since $s_1$ and $s_3$ are toward the same direction, each of $s_1$ and $s_3$ must go inside an open corridor from $s_2$ since otherwise they would have to both go inside $R$ and we could drag $s_2$ to shorten the path.

Let $p$ be the common endpoint of $s_1$ and $s_2$ (e.g., see Fig.~\ref{fig:breakpoint}). Hence, $p$ must be on a door $d_1$ of an open corridor $\calC$. Since $s_1$ goes inside $\calC$, $s_1$ must also intersect the other door $d_2$ of $\calC$. Without loss of generality, we assume $s_1$ is above $s_2$. Since $s_2$ contains an obstacle edge $e$, $d_1$ and $e$ are on the same side of $R$ and $d_1$ is higher than $e$. As $s_1$ intersects both doors of $\calC$, it must be higher than the lower bridge of $\calC$. This implies that $s_2$ must contain the endpoint $p_1$ of the lower bridge of $\calC$ on $d_1$, and we call $p_1$ a {\em breakpoint} of $s_2$.
\end{enumerate}

\begin{figure}[t]
\begin{minipage}[t]{\linewidth}
\begin{center}
\includegraphics[totalheight=1.4in]{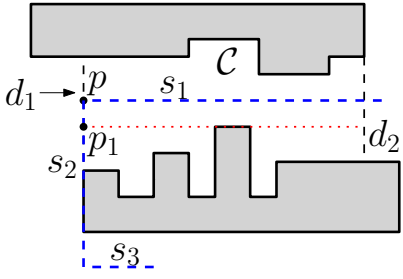}
\caption{\footnotesize Illustrating the case where $s_2$ is vertical.
}
\label{fig:breakpoint}
\end{center}
\end{minipage}
\vspace*{-0.15in}
\end{figure}

In either case above, we show that $s_2$ must contain a backbone point as a breakpoint of $\pi_3$.

If $\pi_3$ has other U-subpaths, then for each of them, the middle segment also contains a backbone point as a breakpoint of $\pi_3$. Hence, if $\pi_4$ is a subpath of $\pi_3$ partitioned by the breakpoints, then $\pi_4$ must be a staircase path and both endpoints of $\pi_4$ must be in $\calB$. Let $a$ and $b$ be the two endpoints of $\pi_4$, respectively. To prove the lemma, it is sufficient to show that $G(\calB)$ has a path connecting $a$ and $b$ with the same length as $\pi_4$ and the path is homotopic to $\pi_4$.

Without loss of generality, we assume that $b$ is to the northeast of $a$.
Based on $\pi_4$, in the following, we obtain another shortest path $\pi_5$ from $a$ to $b$ such that $\pi_5$ has the same length as $\pi_4$ and is homotopic to $\pi_4$. In fact, $\pi_5$ is similar in spirit to the pushed staircase path defined in \cite{ref:YangRe95} (also discussed in Section~\ref{sec:old}) but with respect to the open corridors and the junction rectangles.
If the segment of $\pi_4$ incident to $a$ is horizontal, then let $\alpha$ be the second horizontal segment of $\pi_4$; otherwise let $\alpha$ be the first horizontal segment of $\pi_4$. Unless $\alpha$ is incident to $b$, we
push $\alpha$ upwards until either it
hits a vertex of $\calP$ or it becomes collinear with the next horizontal segment of $\pi_4$. In the latter case, we merge the two horizontal segments and let $\alpha$ refer to the merged segment and we push $\alpha$ upwards again. This procedure stops either when $\alpha$ hits an edge of $\calP$ or becomes incident to $b$.
We do the same for the rest of horizontal segments following their order along the path from $a$ to $b$. Let $\pi_4'$ denote the resulting path. Clearly, $\pi_4'$ has the same length as $\pi_4$ and  is homotopic to $\pi_4$.

Next, we push the vertical segments of $\pi_4'$. If the segment of $\pi'_4$ incident to $a$ is vertical, then let $\beta$ be the second vertical segment of $\pi'_4$; otherwise let $\beta$ be the first vertical segment of $\pi_4'$.
Unless $\beta$ is incident to $b$, we push $\beta$ rightwards until either it
hits a vertex of $\calP$ or it becomes collinear with the next vertical segment of $\pi_4'$. In the latter case, we merge the two vertical segments and let $\beta$ refer to the merged segment and we continue to push $\beta$ rightwards. This procedure stops either when $\beta$ hits a vertex of $\calP$ or becomes incident to $b$. Suppose $\beta$ hits a vertex $v$ of $\calP$. If $v$ is on the boundary of a junction rectangle $R$, in which case $\beta$ is on the right side of $R$, then we do nothing. Otherwise, $v$ must be a vertex in the interior of an open corridor $\calC$, in which case we push $\beta$ {\em leftwards} until it overlaps with the left door of $\calC$ (note that $\beta$ is now on the right side of a junction rectangle). This finishes the push operation for the vertical segment $\beta$. We proceed to do the same for the rest of the vertical segments following their order along the path from $a$ to $b$.
Let $\pi_5$ be the resulting path. Clearly, $\pi_5$ has the same length as $\pi_4'$ and is homotopic to $\pi_4'$.

Consider any 
segment $\alpha$ of $\pi_5$. In the following, we show that $\alpha$ must contain a backbone point of $\calB$. This is obviously true if $\alpha$ is incident to either $a$ or $b$. Below we assume that $\alpha$ is incident to neither $a$ nor $b$. Depending on whether $\alpha$ is horizontal or vertical, there are two cases.

\begin{enumerate}
\item

If $\alpha$ is horizontal, then $\alpha$ contains an obstacle edge $e$ that bounds an obstacle from below. Recall that due to definition of degenerated open corridors, each vertex of $\calP$ must be in a corridor. Let $\calC$ be the corridor that contains the right endpoint of $e$. Note that $\calC$ may be a degenerated open corridor. Since both the vertical segments of $\pi_5$ right before and after $\alpha$ are on vertical sides of junction rectangles, $\alpha$ must travel through $\calC$. Further, $\alpha$ contains the upper bridge of $\calC$ since the portion of $\alpha$ between the two doors of $\calC$ cannot be dragged upwards in $\calC$ due to $e\cap \calC$. Hence, $\alpha$ contains two backbone points that are the two endpoints of the upper bridge of $\calC$.

\item
If $\alpha$ is vertical, then according to our construction of $\pi_5$, $\alpha$ is on the right side of a junction rectangle. Depending on whether $\alpha$ contains an obstacle vertex, there are two cases.

\begin{figure}[t]
\begin{minipage}[t]{0.49\linewidth}
\begin{center}
\includegraphics[totalheight=1.1in]{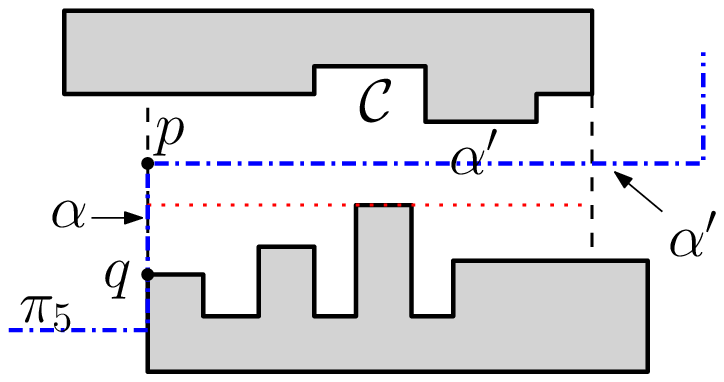}
\caption{\footnotesize The blue dashed dotted path is $\pi_5$, where the two segments $\alpha$ and $\alpha'$ are labeled. The red dotted segment is the lower bridge of $\calC$.
}
\label{fig:case1}
\end{center}
\end{minipage}
\hspace{0.05in}
\begin{minipage}[t]{0.49\linewidth}
\begin{center}
\includegraphics[totalheight=1.1in]{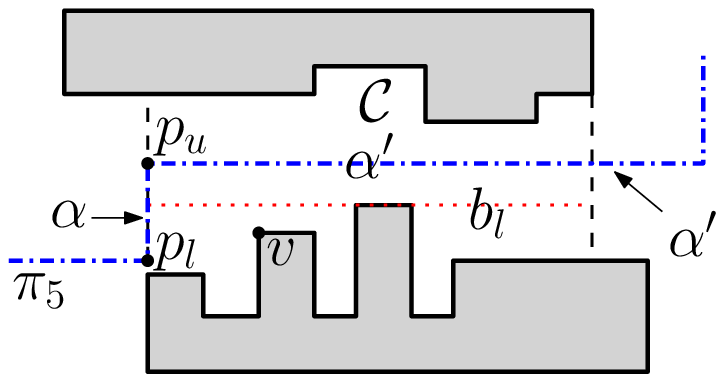}
\caption{\footnotesize The upper and lower endpoints of $\alpha$ are $p_u$ and $p_l$, respectively. The obstacle vertex $v$ is also labeled.
}
\label{fig:case2}
\end{center}
\end{minipage}
\vspace*{-0.15in}
\end{figure}

\begin{enumerate}

\item
If $\alpha$ contains an obstacle vertex, let $p$
be the upper endpoint of $\alpha$ (e.g., see Fig.~\ref{fig:case1}). Then, the next horizontal segment $\alpha'$ of $\pi_5$ starts from $p$ going rightwards inside an open corridor $\calC$, and this
segment travels through $\calC$. Let $d$ be the door of $\calC$
that contains $p$. Let $q$ be the lower endpoint of $d$. Since $\alpha$ contains an obstacle vertex and the upper endpoint of
$\alpha$ is on $d$, $q$ must be on $\alpha$.
Further, since $\alpha'$ travels through $\calC$, $\overline{pq}$ must contain the left endpoint of the lower bridge of $\calC$, which is a backbone point. As $\overline{pq}\subseteq \alpha$, $\alpha$ contains the above backbone point.

\item
Otherwise, according to our construction of $\pi_5$, if we push $\alpha$ rightwards, then we will hit an obstacle vertex $v$ in the interior of $\calC$ (e.g., see Fig.~\ref{fig:case2}). Let $\alpha'$ be the next horizontal segment of $\alpha$. As the above case, $\alpha'$ is going rightwards and travels through $\calC$. This means that $v$ is below $\alpha'$. Note that the lower bridge $b_l$ of $\calC$ must be below $\alpha'$ and above $v$. Also note that $\alpha$ overlaps with the left door $d$ of $\calC$. Let $p_u$ and $p_l$ be the upper and lower endpoints of $\alpha$, respectively. Since $v$ will be hit if we push $\alpha$ rightwards, $v$ is above $p_l$ and below $p_u$. Since $b_l$ is above $v$ and below $\alpha'$ (and thus $p_u$), we obtain that $b_l$ is above $p_l$ and below $p_u$ (e.g., see Fig.~\ref{fig:case2}). Therefore, the left endpoint of $b_l$, which is a backbone point, is on $\alpha$. This proves that $\alpha$ contains a backbone point.
\end{enumerate}

\end{enumerate}

The above shows that each segment of $\pi_5$ contains a backbone point. Hence, each subpath of $\pi_5$ partitioned by all breakpoints on $\pi_5$ must be an L-shaped path. Let $\pi_6$ be any such subpath, and let $x$ and $y$ be its endpoints, which are both in $\calB$. To prove the lemma, it is sufficient to show that $G(\calB)$ contains a path from $x$ to $y$ that has the same length as $\pi_6$ and is homotopic to $\pi_6$, as follows.

Without loss of generality, we assume $y$ is to the northeast of $x$. We also assume that the segment incident to $x$ is horizontal and the one incident to $y$ is vertical. Other cases can be analyzed similarly. Hence, $\pi_6$ consists of a horizontal segment $\overline{xz}$ and a vertical
segment $\overline{zy}$ for some point $z$.

If the rectangle $R_{xy}$ is empty (i.e., $R_{xy}$ in $\calP$), then
by Observation~\ref{obser:recreduced}, $G(\calB)$ contains a staircase
path $\pi_{G(\calB)}(x,y)$ from $x$ to $y$ in $R_{xy}$. Since $R_{xy}$
is in $\calP$,  $\pi_{G(\calB)}(x,y)$ is homotopic to $\pi_6$ with
the same length. Otherwise, there exist backbone points contained in
$R_{xy}$ and they can be ordered as $x=h_1, h_2, \ldots, h_k=y$ such
that $h_{i+1}$ is to the northeast of $h_i$ and $R_{h_i,h_{i+1}}$ is
empty for each $1\leq i\leq k-1$. By
Observation~\ref{obser:recreduced}, for each $1\leq i\leq k-1$, since
the rectangle $R_{h_i,h_{i+1}}$ is empty, $G(\calB)$ has a staircase
path $\pi_{G(\calB)}(h_i,h_{i+1})$ from $h_i$ to $h_{i+1}$. Let
$\pi_{G(\calB)}(x,y)$ be the concatenation of all these staircase
paths $\pi_{G(\calB)}(h_i,h_{i+1})$ for $i=1,2,\ldots,k-1$. Clearly,
$\pi_{G(\calB)}(x,y)$ is a staircase path and thus has the same length
as $\pi_6$. In the following, we show how we find the above sequence
of backbone points, and the way we find them will also imply that
$\pi_{G(\calB)}(x,y)$ is homotopic to $\pi_6$.

\begin{figure}[t]
\begin{minipage}[t]{\linewidth}
\begin{center}
\includegraphics[totalheight=1.2in]{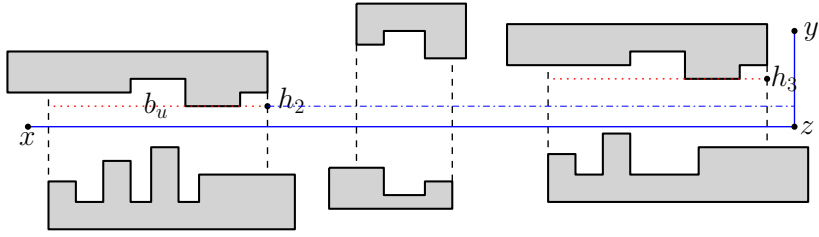}
\caption{\footnotesize Illustrating an example where $\overline{xz}$ travels through three open corridors and the two points $h_2$ and $h_3$ are labeled.
}
\label{fig:sequence}
\end{center}
\end{minipage}
\vspace*{-0.15in}
\end{figure}

Since $\overline{yz}$ is vertical and $y$ is in a junction rectangle, $z$ is also in the same junction rectangle. Refer to Fig.~\ref{fig:sequence}.
As $R_{xy}$ is not empty and both $x$ and $z$ are in junction rectangles, $\overline{xz}$ must travel through some open corridors (maybe degenerated). We push $\overline{xz}$ upwards until it hits a vertex of $\calP$, at which moment, the new segment, denoted by $\alpha$, must contain the upper bridge $b_u$ of an open corridor, and we let $h_2$ refer to the right endpoint of $b_u$ (recall that $h_1=x$). Note that $h_2$ is to northeast of of $h_1=x$ and $R_{h_1,h_2}$ is empty. Next, we consider the L-shaped path $\overline{h_2z}\cup \overline{zy}$. Note that $h_2$ is also in a junction rectangle. Hence, we can use the same way as above to find $h_3$, $h_4$, $\ldots$, until at some moment the pushed horizontal segment contains $y$.

As a summary, $G(\calB)$ contains a path $\pi_{G(\calB)}$ from $s$ to $t$ such that $\pi_{G(\calB)}$ has the same length as the optimal path $\pi$ and is homotopic to $\pi$.
\end{proof}

The following corollary confirms that $G(\calB)$ is indeed a ``path-preserving'' graph.

\begin{corollary}\label{coro:target}
A shortest \st\ path in $G(\calB)$ is a shortest \st\ path in $\calP$.
\end{corollary}
\begin{proof}
Let $\pi$ be a minimum-link shortest \st\ path in $\calP$. By
Lemma~\ref{lem:target}, there is a path $\pi_{G(\calB)}$ from $s$ to
$t$ in  $G(\calB)$ with the same length of $\pi$. On the other hand, any
path in $G(\calB)$ corresponds to a path in $\calP$ with the same
length. Hence,
$\pi_{G(\calB)}$ is a shortest \st\ path in both $G(\calB)$ and $\calP$.
The corollary thus follows.
\end{proof}

\subsubsection{Computing the Graph $G(\calB)$ and the Reduced Domain}
\label{sec:graphconstruct}

We show that the graph $G(\calB)$ can be computed in $O(n+h\log^2 h)$ time. To this end, we will introduce a {\em reduced domain} $\calP_r$, which is a polygonal domain that is a subset of $\calP$ and has $O(h)$ vertices, such that every ordinary edge of $G(\calB)$ is in $\calP_{r}$.

Recall that in Section~\ref{sec:old} the graph $G(\calV)$ with respect to $\calV$ of $n$ points can be constructed in $O(n\log^2 n)$ time~\cite{ref:ClarksonRe87,ref:LeeSh91,ref:YangRe95}. To construct $G(\calB)$, one possible solution is to modify the previous algorithms~\cite{ref:ClarksonRe87,ref:LeeSh91,ref:YangRe95} on the set $\calB$ of $O(h)$ points. However, since we still need to determine whether two points of $\calB$ is visible in $\calP$ in order to determine whether $G(\calB)$ has an edge connecting the two points, even if we can reduce the factor $O(n\log^2 n)$ to $O(h\log^2 h)$, the algorithm may still suffer an $O(n\log n)$ factor in the time complexity. In the following, we propose a different approach.

We assume that the corridor structure of $\calP$ has already been computed. First of all, all backbone points can be easily computed in $O(n)$ time. 
Then, by using the algorithm in~\cite{ref:SchuiererAn96}, all corridor paths and thus the corridor edges of $G(\calB)$ can be computed in $O(n)$ time since the total size of all corridors is $O(n)$. It remains to compute the ordinary edges of $G(\calB)$, as follows.


Consider any ordinary edge $e$ of $G(\calB)$ that connects two vertices $v_1$ and $v_2$. Hence, $e$ is the segment $\overline{v_1v_2}$ that is either horizontal or vertical. Note that all vertices of $G(\calB)$ are in junction rectangles. If $e$ is vertical, an easy observation is that $e$ must be in a junction rectangle.

Suppose $e$ is horizontal and $e$ is not contained in a junction rectangle. Then, $v_1$ and $v_2$ are in two different junction rectangles. Hence, $e$ must travel through some open corridors. Observe that if $e$ travels through an open corridor $\calC$, then $e$ does not contain any point of $\calC$ that is not in the canal of $\calC$. This means that $e$ must be in the union of all junction rectangles and canals of all open corridors.

Define $\calP_{r}$ as the union of all junction rectangles and canals of all open corridors. The above discussions lead to the following lemma.

\begin{lemma}\label{lem:reduceddomain}
Every ordinary edge of $G(\calB)$ is in $\calP_{r}$.
\end{lemma}

Since there are $O(h)$ junction rectangles and open corridors, and each canal of an open corridor is a rectangles, $\calP_{r}$ is essentially a polygonal domain that is the union of $O(h)$ rectangles.  Hence,  $\calP_{r}$ has $O(h)$ vertices and edges. We call $\calP_r$ the {\em reduced domain}. Constructing  $\calP_{r}$ can be easily done in $O(n)$ time from the corridor structure of $\calP$.

By Lemma~\ref{lem:reduceddomain}, we can compute the ordinary edges of $G(\calB)$ with respect to the reduced domain $\calP_r$ of $O(h)$ complexity instead of $\calP$ of $O(n)$ complexity. Consequently, by applying the previous algorithms~\cite{ref:ClarksonRe87,ref:LeeSh91,ref:YangRe95}, we can compute all ordinary edges of $G(\calB)$ in $O(h\log^2 n)$ time.

As a summary, we can compute the graph $G(\calB)$ in $O(n+h\log^2 h)$ time  and $O(n+h\log h)$ space.

\subsection{Computing an Optimal Path Using $G(\calB)$}
\label{sec:compute}

In this section, we compute an optimal \st\ path using $G(\calB)$. Specifically, we show that an optimal \st\ path can be computed by applying the dragging operations as in \cite{ref:YangRe95} on the ordinary edges of $\pi_{G(\calB)}$ and applying a new kind of operations, called {\em through-corridor-path generating operations}, on corridor edges of $\pi_{G(\calB)}$, where $\pi_{G(\calB)}$ is a target path of $G(\calB)$ defined in Lemma~\ref{lem:target}.

The algorithmic scheme is similar to that in
Section~\ref{sec:corrected}. Recall that each ordinary edge of
$G(\calB)$ is either horizontal or vertical. When we advance the
searching process through an ordinary edge, we perform a dragging
operation in exactly the same way as described in  Section~\ref{sec:corrected}
(which is also the way in the YLW algorithm \cite{ref:YangRe95}). If
we are advancing along a corridor edge, then we apply a
through-corridor-path generating operation, which is introduced in the
following. To this end, we first review some results from Schuierer~\cite{ref:SchuiererAn96}.

Consider a closed corridor $\calC$. Let $d$ be a door of $\calC$ and
let $q$ be the backbone point on $d$. Recall that $q$ is an extension
of a window $w$ of the maximal histogram $H$ in $\calC$ with base $d$.

\begin{figure}[t]
\begin{minipage}[t]{0.49\linewidth}
\begin{center}
\includegraphics[totalheight=1.9in]{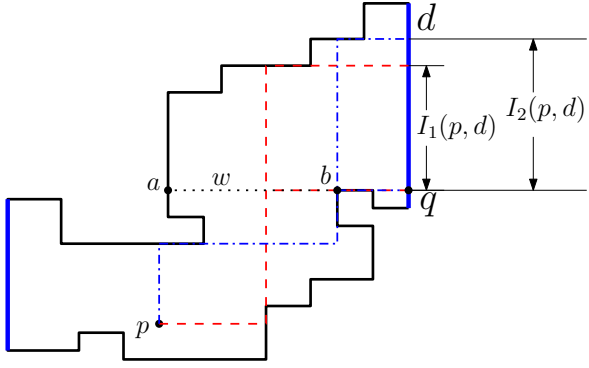}
\caption{\footnotesize Illustrating the two intervals $I_1(p,d)$ and $I_2(p,d)$, where $L_d(p,d)=3$ and $\overline{ab}$ is the window $w$. The two blue segments are doors of the corridor.
}
\label{fig:intervals}
\end{center}
\end{minipage}
\hspace*{0.05in}
\begin{minipage}[t]{0.49\linewidth}
\begin{center}
\includegraphics[totalheight=1.9in]{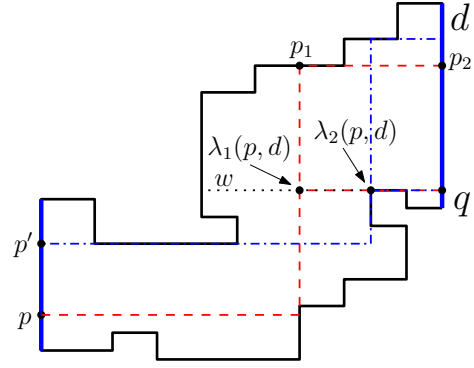}
\caption{\footnotesize Illustrating the two points $\lambda_1(p,d)$ and $\lambda_2(p,d)$ on the window $w$. $p'$ is also a backbone point.
}
\label{fig:newintervals}
\end{center}
\end{minipage}
\vspace*{-0.15in}
\end{figure}

Let $p$ be a point in $\calC$.
Following the terminology in \cite{ref:SchuiererAn96},
a rectilinear path from $p$ to a point on $d$ is called an {\em
admissible path} if the last link is orthogonal to $d$. A {\em
minimum-link admissible path} from $p$ to $d$ is an admissible path
from $p$ to any point of $d$ with the smallest number of links, and we
use $L_d(p,d)$ to denote the number of links in the path. Let $I_1(p,d)$
(resp., $I_2(p,d)$) denote the set of points on $d$ that can be
reached from $p$ with an admissible path of at most $L_d(p,d)$ (resp.,
$L_d(p,d)+1$) links (e.g., see Fig.~\ref{fig:intervals}). It is known
that each of $I_1(p,d)$ and $I_2(p,d)$ is an interval of $d$, and
$I_1(p,d)\subseteq I_2(p,d)$~\cite{ref:SchuiererAn96}. Further, if $p$
is not horizontally visible to $d$, then both intervals have $q$ as one of
their endpoints. By using the histogram partition $\calH(\calC,d)$,
Schuierer~\cite{ref:SchuiererAn96} built a data structure in $O(|\calC|)$ time such that given any point $p\in \calC$, the two intervals $I_1(p,d)$ and $I_2(p,d)$ can be determined in $O(\log |\calC|)$ time. With a little abuse of notation, we also use  $\calH(\calC,d)$ to refer to the above data structure.

Suppose $p$ is a point on the other door of $\calC$ than $d$ (so $p$
is not horizontally visible to $d$).
Then, $I_1(p,d)$ is uniquely determined by a point, denoted by $\lambda_1(p,d)$,
on the window $w$ in the following way~\cite{ref:SchuiererAn96} (e.g., see Fig.~\ref{fig:newintervals}). Recall that $d$ is vertical and thus $w$ is horizontal. Without loss of generality, assume that the
histogram $H$ is locally above $w$ and locally on the left of
$d$. We shoot a ray from $\lambda_1(p,d)$ upwards until a point
$p_1$ on the boundary of $\calC$ and then we project $p_1$
perpendicular to $d$ and let $p_2$ be the projection point. The point
$p_2$ is the other endpoint of the interval $I_1(p,d)$, i.e.,
$I_1(p,d)=\overline{qp_2}$. Note that $p_2$ is above $q$. Let
$I_1'(p,d)$ denote the segment $\overline{\lambda_1(p,d)q}$, which is
on the extension of the window $w$.
We can also understand the two intervals $I_1(p,d)$ and $I_1'(p,d)$
in the following way. There exists an admissible path of
$L_d(p,d)$ links from $p$ to $q$, denoted by $\pi_1(\calC,p,q)$, which is actually a {\em smallest} path from $p$ to $q$, and its last link is
$I_1'(p,d)$; for any point $q'\in I_1(p,d)$, by
dragging the last segment of $\pi_1(\calC,p,q)$ upwards until $q'$, we can obtain an
admissible path of $L_d(p,d)$ links from $p$ to $q'$. The data
structure $\calH(\calC,d)$ can also report $\lambda_1(p,d)$ in $O(\log
n)$ time and the path $\pi_1(\calC,p,q)$ can be output in additional time linear in
the link distance of the path.

The interval $I_2(p,d)$ is uniquely determined by a point
$\lambda_2(p,d)$ on the window $w$ in the similar way as above. Similarly, we define $I_2'(p,d)$ 
and the corresponding admissible path of $L_d(p,d)+1$ links from $p$ to $q$ whose last link is $I_2'(p,d)$, denoted by $\pi_2(\calC,p,q)$, which
is a {\em shortest} path (but not necessarily a smallest path) from $p$ to $q$ in $\calC$~\cite{ref:SchuiererAn96}.
Similarly, the data
structure $\calH(\calC,d)$ can also report $\lambda_2(p,d)$ in $O(\log
n)$ time and the path $\pi_2(\calC,p,q)$ can be output in additional time linear in
the link distance of the path.

In the following, we introduce our through-corridor-path generating operations for advancing paths along corridor edges in our algorithm for searching the graph $G(\calB)$.

Consider a corridor edge $e(q_1,q_2)$ connecting two vertices $q_1$ and $q_2$
of $G(\calB)$. Note that $q_1$ and $q_2$ are two backbone points that are on the two
doors $d_1$ and $d_2$ of a closed corridor $\calC$, respectively. Consider a path $\pi(s,q_1)$ from $s$ to $q_1$ maintained by our algorithm. Suppose we want to
advance $\pi(s,q_1)$ from $q_1$ to $q_2$ along the corridor edge $e(q_1,q_2)$. We perform the following through-corridor-path generating operation that will extend $\pi(s,q_1)$ from $q_1$ to $q_2$ to obtain a path $\pi(s,q_2)$ from $s$ to $q_2$.

Recall that $q_1$ is an extension of a window $w_1$ of the maximal
histogram $H_1$ in $\calC$ with base $d_1$. Hence, $w_1$ divides
$\calC$ into two sub-polygons that contain $d_1$ and $d_2$, respectively. Without loss of generality, we assume that the sub-polygon containing $d_2$ is locally above $w_1$. We also assume that $\calC$ is locally on the right of $d_1$ (e.g., see Fig.~\ref{fig:tcp}).

\begin{figure}[t]
\begin{minipage}[t]{0.49\linewidth}
\begin{center}
\includegraphics[totalheight=1.9in]{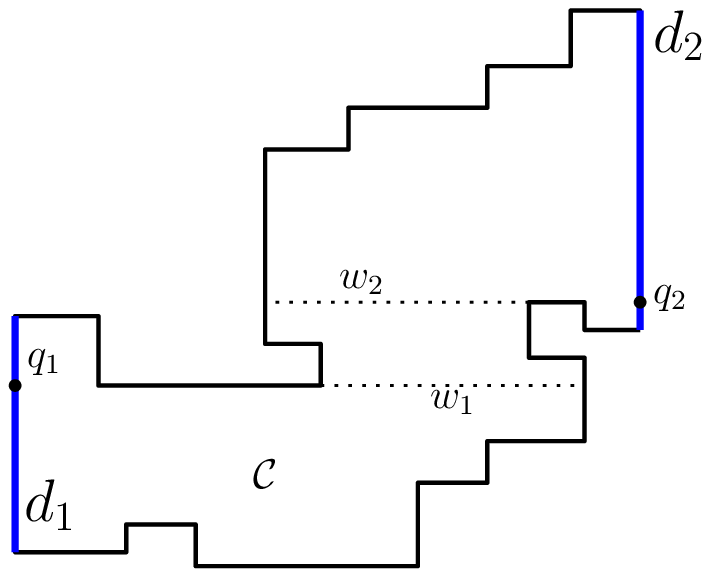}
\caption{\footnotesize Illustrating a closed corridor $\calC$ with two doors $d_1$ and $d_2$. $q_1$ and $q_2$ are the two backbone points, defined by the windows $w_1$ and $w_2$, respectively.
}
\label{fig:tcp}
\end{center}
\end{minipage}
\hspace*{0.05in}
\begin{minipage}[t]{0.49\linewidth}
\begin{center}
\includegraphics[totalheight=1.9in]{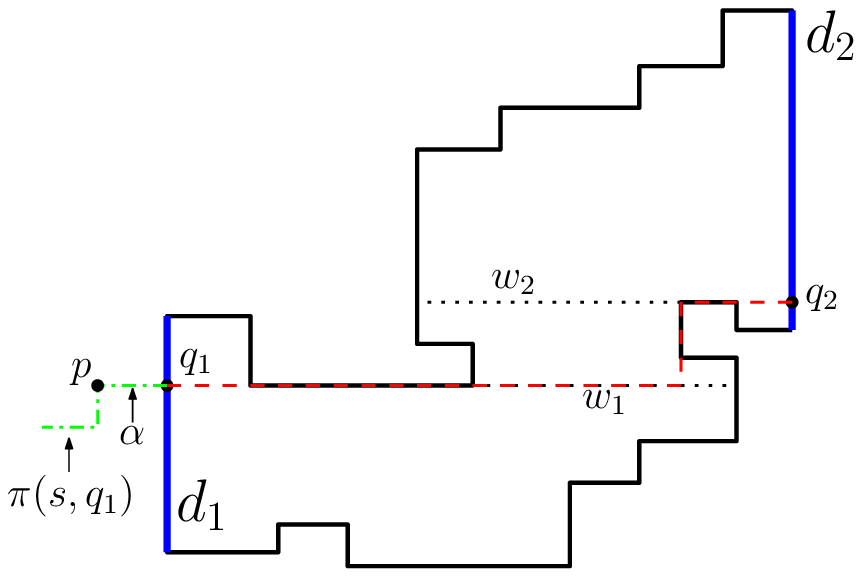}
\caption{\footnotesize Illustrating the through-corridor-path generating operation for the case where $\alpha$ is horizontal. The path $\pi_1(\calC,q_1,q_2)$ are shown with red dashed segments.
}
\label{fig:tcp1}
\end{center}
\end{minipage}
\vspace*{-0.15in}
\end{figure}

Let $\alpha$ be the last segment of $\pi(s,q_1)$ (i.e., the one incident to $q_1$) and let $p$ be the other endpoint of $\alpha$ than $q_1$. Suppose we have already built the data structure $\calH(\calC,d_2)$ for $\calC$ with respect to the door $d_2$.
Depending on whether $\alpha$ is horizontal or vertical, there are two cases.

\begin{enumerate}
\item
If $\alpha$ is horizontal (e.g., see Fig.~\ref{fig:tcp1}), then $p$ must be to the left of $q_1$ since
$\calC$ is locally on the right side of $d_1$. In this case, we use
$\calH(\calC,d_2)$ to determine the path $\pi_1(\calC,q_1,q_2)$ (whose last link is
$I_1'(q_1,d_2)$) and concatenate it with $\pi(s,q_1)$ to obtain $\pi(s,q_2)$. 
We also compute the number of links of $\pi(s,q_2)$ and its length, and store them at $q_2$. Note that $L_1(\pi(s,q_1))$ and $L_d(\pi(s,q_1))$ are already stored at $q_1$.

\item
If $\alpha$ is vertical, then depending on whether $p$ is above $q_1$, there are two subcases.

\begin{enumerate}
\item

If $p$ is above $q_1$, then we use the same approach as above to obtain $\pi(s,q_2)$. Note that in this case the path makes a turn at $q_1$ while there is no turn at $q_1$ in the above case.

\item
If $p$ is below $q_1$, then depending on whether $p$ is on $d_1$, there are further two subcases.

\begin{enumerate}
\item
If $p$ is not on $d_1$, then we use the same approach as above to obtain $\pi(s,q_2)$.

\item
\label{case:except}
If $p$ is on $d_1$, this is the trickiest case. 
We use $\calH(\calC,d_2)$ to determine the path $\pi_1(\calC,p,q_2)$ (whose last link is
$I_1'(p,d_2)$; e.g., see Fig.~\ref{fig:tcp2}). 
We then obtain $\pi(s,q_2)$ by concatenating $\pi_{1}(\calC,p,q_2)$ with the
subpath of $\pi(s,q_1)$ between $s$ and $p$ (thus $\overline{pq_1}$ is not in the resulting path $\pi(s,q_2)$ unless it is contained in $\pi_{1}(\calC,p,q_2)$).
\end{enumerate}
\end{enumerate}

\end{enumerate}

\begin{figure}[t]
\begin{minipage}[t]{\linewidth}
\begin{center}
\includegraphics[totalheight=1.9in]{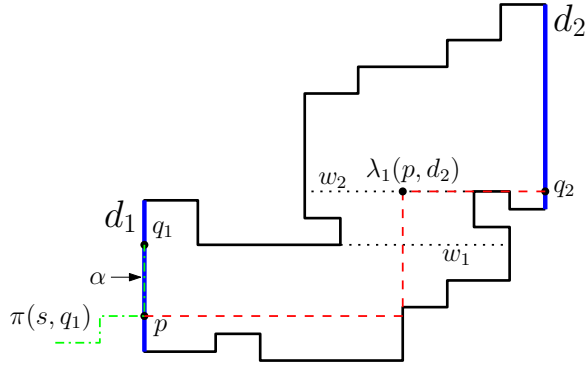}
\caption{\footnotesize Illustrating the through-corridor-path generating operation for the case where $\alpha$ is vertical and $p$ is on $d_1$ below $q_1$. The smallest path $\pi_{opt}(\calC,q_1,q_2)$ are shown with red dashed segments.
Note that $I_1'(p,d_2)=\overline{\lambda_1(p,d_2)q_2}$.
}
\label{fig:tcp2}
\end{center}
\end{minipage}
\vspace*{-0.15in}
\end{figure}

As a summary, to obtain $\pi(s,q_2)$, if Case~\ref{case:except}
happens, then we connect the subpath of $\pi(s,q_1)$ between $s$ and
$p$ with $\pi_{1}(\calC,p,q_2)$; otherwise, we connect $\pi(s,q_1)$
with $\pi_{1}(\calC,q_1,q_2)$. 
In either case, let $\pi'$
be the subpath of $\pi(s,q_2)$ contained in $\calC$.
With the histogram partition $\calH(\calC,d_2)$, we can obtain
$L_1(\pi')$ and $L_d(\pi')$ as well as the first and last links of
$\pi'$
in $O(\log n)$ time (the actual path $\pi'$ can be output in additional
$O(L_d(\pi'))$ time). Hence, we can compute $L_1(\pi(s,p_2))$ and
$L_d(\pi(s,p_2))$ as well as its last link in $O(\log n)$ time,
without explicitly computing the actual path $\pi'$.
Therefore, the through-corridor-path generating operation can be performed in $O(\log n)$ time.

As discussed before, our algorithm works in the same way as the one in
Section~\ref{sec:old} except that we apply through-corridor-path
generating operations on corridor edges of $G(\calB)$ instead of the
dragging operations. We can compute the histogram partitions
for all closed corridors as the preprocessing for performing the
through-corridor-path generating operations, and the total
preprocessing time is $O(n)$ since the size of all corridors is
$O(n)$. After the algorithm finishes, the path $\pi$ stored at $t$
with the smallest measure is an optimal \st\ path. Note that if $\pi$
has some subpaths in closed corridors, then $\pi$ is implicitly
maintained, we can output those subpaths in linear time by using the
histogram partitions on the closed corridors.
The following theorem gives some implementation details and analyzes the time complexities. The
algorithm correctness is proved in the next subsection.

\begin{theorem}\label{theo:singlepath}
We can compute a minimum-link shortest \st\ path
in $O(n+h\log^{3/2} h)$ time and $O(n+h\log h)$ space, and
compute a shortest minimum-link \st\ path or a minimum-cost \st\
path in $O(n+h^2\log^{3/2} h)$ time and $O(n+h^2\log h)$
space.
\end{theorem}
\begin{proof}
We will first show that computing a minimum-link shortest path can be done in
$O(n+h\log^{2} h)$ time and $O(n+h\log h)$ space
and computing other two types of optimal paths can be done in
$O(n+h^2\log^{2} h)$ time and $O(n+h^2\log h)$ space, and then we
will improve the algorithms by utilizing the techniques
in~\cite{ref:ChenOn01,ref:YangRe95} discussed in Section~\ref{sec:old}
as well as the reduced domain $\calP_r$ proposed in
Section~\ref{sec:graphconstruct}.

First of all, as discussed in Section~\ref{sec:graphconstruct},
building the graph $G(\calB)$ takes $O(n+h\log^2 h)$
time and $O(h\log h)$ space. The preprocessing on all closed corridors
take $O(n)$ time in total, so that each through-corridor-path
generating
operation can be performed in $O(\log n)$ time. As in
\cite{ref:YangRe95}, with $O(n)$ time preprocessing, each
dragging operation can be performed in $O(\log n)$ time.

For computing a minimum-link shortest \st\ path, since
$G(\calB)$ has $O(h\log h)$ ordinary edges and $O(h)$ corridor edges,
we only need to apply the dragging operations  $O(h\log h)$
times and apply the through-corridor-path generating operations $O(h)$ times.
Thus, the total time on performing these operations is $O(h\log h\log n)$.
After the algorithm finishes, outputting the optimal path $\pi$ needs
additional
$O(n+h\log n)$ time since $\pi$ travels through at most $O(h)$ closed corridor paths.
Therefore, the total time of the algorithm is $O(n+h\log h\log n)$.
Note that $n+h\log h\log n = O(n+h\log^2 h)$. The space complexity is $O(n+h\log h)$.

For computing other two types of optimal paths,
because each node of $G(\calB)$ may store $O(h)$ paths, the total number of
paths stored in the algorithm is $O(h^2\log h)$. Hence, in the entire
algorithm, the total number of the dragging operations is $O(h^2\log h)$
and the total number of through-corridor-path generating operations is $O(h^2)$. Thus, these operations together take $O(h^2\log h\log n)$ time, and the
algorithm runs in $O(n+h^2\log h\log n)$ time in total. Note that $n+h^2\log
h\log n=O(n+h^2\log^2 h)$. The space complexity is $O(n+h^2\log h)$.

In the sequel, we improve the above algorithms by using the reduced domain
$\calP_r$ proposed in Section~\ref{sec:graphconstruct} and the
techniques in~\cite{ref:ChenOn01,ref:YangRe95}.

We first discuss the problem of finding a minimum-link shortest path.
To reduce the running time, one key issue is to reduce the time on the dragging operations as there are $O(h\log
h)$ such operations in the algorithm. The bottleneck of each such operation
is to answer the following {\em segment dragging queries}: Given an
ordinary edge $e$ of $G(\calB)$ and a direction $\rho$ perpendicular
to $e$, the query asks for the first vertex of $\calV$ hit by $e$
(called the {\em hit vertex} in \cite{ref:YangRe95}) if
we drag $e$ along the direction $\rho$ (such a hit vertex is undefined if $e$ hits an interior of an edge of $\calP$). Note that $e$ is either horizontal or vertical. Each such query
can be answered in $O(\log n)$ time with $O(n)$ time
preprocessing~\cite{ref:ChazelleAn88}.  To reduce the
time, the idea in \cite{ref:YangRe95} is
to compute the results of the segment dragging queries on all edges of
the graph in the preprocessing, so that the hit vertex of each such query can be obtained in $O(1)$ time during the course of the algorithm. To adapt their
techniques, we show below that in our algorithm on $G(\calB)$ we only need to solve those
segment dragging queries with respect to the reduced domain $\calP_r$
instead of $\calP$.

Let $\calP'$ be the union of $\calP_r$ and all closed corridors. An
observation is that the optimal path obtained by our algorithm, i.e.,
by applying the dragging operations on the ordinary edges of a
target path $\pi_{G(\calB)}(s,t)$ and applying the
through-corridor-path generating operations on the corridor edges of
$\pi_{G(\calB)}(s,t)$, must be in $\calP'$. Indeed, this can be verified by checking that the
optimal path $\pi_5$ obtained in the proof of Lemma~\ref{lem:target}
is in $\calP'$. Further, the closed corridors only affect the results
of the through-corridor-path generating operations. Hence, to perform
segment dragging queries (which are only used in the dragging operations), it is sufficient to only consider the domain
$\calP_r$, i.e., finding the hit vertices in $\calP_r$.

With the above discussions, we adapt the techniques of
\cite{ref:YangRe95} in the following way. First, as discussed in
Section~\ref{sec:old}, we construct another path-preserving graph
$G^*(\calB)$ with respect to $\calB$ in the same way as $G^*(\calV)$
with respect to $\calV$, and $G^*(\calB)$  has
of $O(h\log^{1/2} h)$ vertices and $O(h\log^{3/2} h)$ edges. Next, we
insert the $O(h)$ corridor edges to $G^*(\calB)$. As $G(\calB)$, we
can compute all ordinary edges of $G^*(\calB)$ with respect to the
reduced domain $\calP_r$ in $O(h\log^{3/2}h)$ time and space by
using exactly the same algorithm of~\cite{ref:YangRe95} but on $\calB$ and
$\calP_r$. Further, we compute the hit vertices of all ordinary edges of $G^*(\calB)$
in the preprocessing by using the same algorithm in~\cite{ref:YangRe95}, but again
on $\calB$ and the reduced domain $\calP_r$, in $O(h\log^{3/2}h)$ time.

Since $G^*(\calB)$ has $O(h\log^{1/2}h)$ vertices and
$O(h\log^{3/2}h)$ edges, searching the graph using Dijkstra's
algorithm runs in $O(h\log^{3/2}h)$ time. Note that each
through-corridor-path operation still takes $O(\log n)$ time. But
since there are only $O(h)$ corridor edges in the graph, the total
time of the algorithm is bounded by $O(n+h\log^{3/2}h)$. The space
complexity becomes $O(n+h\log^{3/2}h)$ as $G^*(\calB)$ has
$O(h\log^{3/2}h)$ edges. Using the techniques of \cite{ref:ChenOn01}, we can further reduce number of edges of $G^*(\calB)$  to $O(h\log h)$ by representing some edges of the graph implicitly. Some details on maintaining the edges implicitly were provided in \cite{ref:ChenOn01}. In the following, we add more details on computing the hit vertices of all edges of $G^*(\calB)$.
The algorithm FindGG' in~\cite{ref:YangRe95} computes the hit vertices of all ordinary edges of $G^*(\calB)$ in $O(h\log^{3/2}h)$ time and $O(h\log^{3/2}h)$ space.
We modify it in the following way so that the space can be reduced to  $O(h\log h)$ while keeping the same running time asymptotically (the idea should also be used in our $O(n\log^{3/2}n)$ time and $O(n\log n)$ space algorithm for computing the minimum-link shortest paths using the graph $G^*(\calV)$ in Section~\ref{sec:corrected}).

Consider a cut-line $L$ and a {\em horizontal strip} (i.e., a plane region bounded by two horizontal lines) as in the description of FindGG'~\cite{ref:YangRe95}. There is a set $S$ of $k=O(\log h)$ vertices of $\calP_r$ that are horizontally visible to $L$ in the strip. Each vertex of $S$ defines a Steiner point on $L$, so there are $k$ Steiner points on $L$ in the strip. We sort these Steiner points on $L$. For each segment of $L$ divided by these Steiner points in the strip, the algorithm FindGG' computes its hit vertices on its both left and right sides. In the following, we only discuss the right hit vertices. All these hit vertices in all cut-lines and all strips can be computed in $O(h\log^{3/2}h)$ time and $O(h\log h)$ space. One issue is that for every pair of Steiner points (not necessarily adjacent) $a$ and $b$ defined by $S$ on $L$, we need to compute the (right) hit vertex of $\overline{ab}$. To this end, FindGG' uses a table of size $O(k^2)$ to maintain these hit vertices, so that given $a$ and $b$, the hit vertex of $\overline{ab}$ can be obtained in $O(1)$ time. But this table makes the total space of the algorithm become $O(h\log^{3/2}h)$. To reduce the space while still keeping the $O(1)$ query time, we replace the table by an array of size $k+1$ and construct a range-minima data structure on the array~\cite{ref:BenderTh00,ref:HarelFa84}.
Specifically, let $l_i$ be the $i$-th lowest segment of $L$ divided by the Steiner points of $S$. Thus, $L$ has $k+1$ such segments in the strip.
Let $A[1\cdots k+1]$ be an array of $k+1$ elements such that each $A[i]$ represents the $x$-coordinate of the hit vertex of $l_i$ (we also associate the hit vertex with $A[i]$). We build a range-minima data structure on $A$ in $O(k)$ time such that given any $i$ and $j$ with $1\leq i\leq j\leq n$, the minimum value (and its index in $A$) in the subarray $A[i\cdots j]$ can be found in $O(1)$ time~\cite{ref:BenderTh00,ref:HarelFa84}. Given any two Steiner points $a$ and $b$ on $L$ defined by $S$, suppose $a$ is the lower endpoint of $l_i$ and $b$ is the upper endpoint of $l_j$, then the hit vertex of $\overline{ab}$ is exactly the one associated with the minimum value in the subarray $A[i\ldots j]$, which can be found in $O(1)$ time by the range-minima data structure. In this way, we only need $O(k)$ space for each strip. Thus, the total space of the algorithm becomes $O(h\log h)$. The total time of the algorithm is still $O(h\log^{3/2}h)$. Further, given any ordinary edge of $G^*(\calB)$, its hit vertex can still be found in $O(1)$ time.

Therefore, we can compute a minimum-link shortest path
in $O(n+h\log^{3/2}h)$ time and $O(n+h\log h)$ space.

For computing the other two types of optimal paths, we can use the similar
idea as above. The running time is $O(n+h^2\log^{1/2}h)$ and the space
is $O(n+h^2\log h)$. We omit the details.
\end{proof}


\subsection{The Algorithm Correctness}
\label{sec:correct}

In this section, we prove the correctness of our algorithm. As will be seen later, the main effort is to show that our through-corridor-path generating operations are correct.

Let $\pi$ be an optimal \st\ path and let $\pi_{G(\calB)}$
be the corresponding target path obtained in the constructive proof of  Lemma~\ref{lem:target}.
If we can prove the following {\em main claim}: an optimal \st\ path can be
obtained by applying the segment dragging operations and through-corridor-path
generating operations on the edges of $\pi_{G(\calB)}$ in the
order from $s$ to $t$, then by the proof techniques of
Section~\ref{sec:correctold}, we can also show that our algorithms can
correctly compute an optimal \st\ path. Hence, in the following, we
focus on proving the above main claim.

We assume that $\pi$ travels through at least one closed corridor since
otherwise the analysis would be similar (and simpler because we would not need
to consider through-corridor-path generating operations). Along the path $\pi$
from $s$ to $t$, let $\calC$ be the first closed corridor traveled
through by $\pi$. Let $d_1$ be the first door of $\calC$ intersected
by $\pi$ and let $d_2$ be the other door. Let $\pi(a_1,a_2)$ denote
the subpath of $\pi$ in $\calC$ with $a_1\in d_1$ and $a_2\in d_2$
such that the edge of $\pi(a_1,a_2)$ incident to $a_1$ is
perpendicular to $d_1$ and the edge of $\pi(a_1,a_2)$ incident to
$a_2$ is perpendicular to $d_2$.  Refer to Fig.~\ref{fig:cortcp} for an example. Note that such a subpath must exist as $\pi$ travels through $\calC$.
Let $\overline{a_1a_1'}$ and $\overline{a_2a_2'}$ be the first and
last edges of $\pi(a_1,a_2)$, respectively. Let $\overline{aa_1}$ be
the last edge of $\pi(s,a_1)$.

\begin{figure}[t]
\begin{minipage}[t]{0.49\linewidth}
\begin{center}
\includegraphics[totalheight=1.9in]{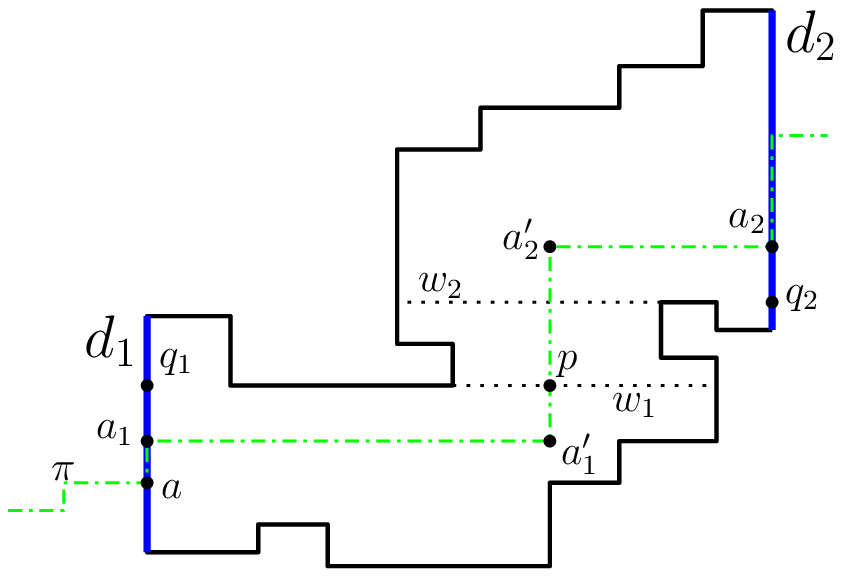}
\caption{\footnotesize Illustrating the definitions of $a$, $a_1$, $a_1'$, $a_2$, and $a_2'$. The (green) dashed dotted path is $\pi$.
}
\label{fig:cortcp}
\end{center}
\end{minipage}
\hspace*{0.05in}
\begin{minipage}[t]{0.49\linewidth}
\begin{center}
\includegraphics[totalheight=1.9in]{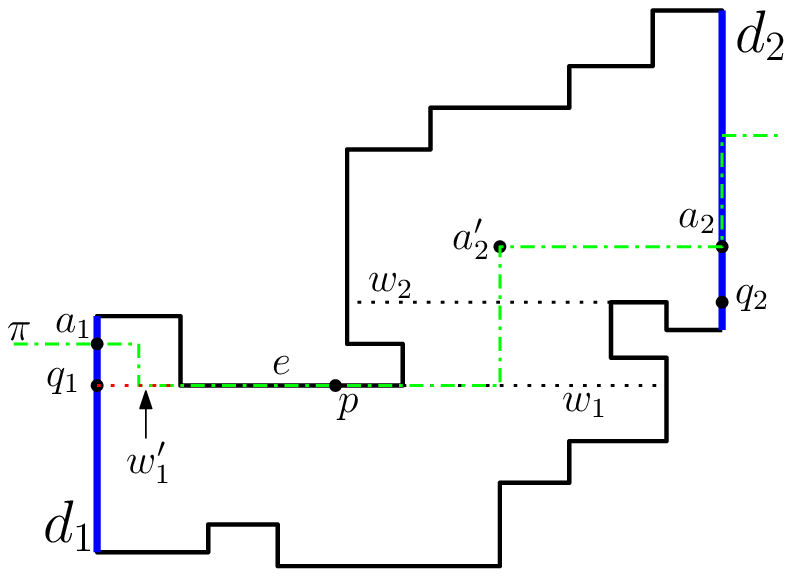}
\caption{\footnotesize Illustrating the case where $a_1$ is above $q_1$.
}
\label{fig:cortcp1}
\end{center}
\end{minipage}
\vspace*{-0.15in}
\end{figure}

Let $q_1$ and $q_2$ be the backbone points on $d_1$ and $d_2$, respectively (e.g., see Fig.~\ref{fig:cortcp}). Recall that $q_1$ is an extension of a window $w_1$ of the maximal histogram $H_1$ in $\calC$ with base $d_1$. Hence, $w_1$ divides
$\calC$ into two sub-polygons that contain $d_1$ and $d_2$, respectively. Without loss of generality, we assume that the sub-polygon containing $d_2$ is locally above $w_1$. We also assume that $\calC$ is locally on the right of $d_1$.

\subsubsection{Obtaining a Special Optimal \st\ Path $\pi'$}

In the following, we obtain another optimal \st\ path $\pi'$ that is homotopic to $\pi$, and $\pi'$ has some special properties that will facilitate our analysis later. First of all, if $\pi$ contains $q_1$, then we let $\pi'=\pi$. Below we assume that $\pi$ does not contain $q_1$. Depending on whether $a_1$ is above or below $q_1$, there are two cases.

\begin{enumerate}
\item
If $a_1$ is above $q_1$ (e.g., see Fig.~\ref{fig:cortcp1}), then $\pi(a_1,a_2)$ must intersect $w_1'$ and $w_1$, where $w_1'$ is the window that contains $q_1$ and is on the extension of $w_1$. This further implies that $\pi(a_1,a_2)$ must contain the edge $e$ of $\calC$ that is between $w_1'$ and $w_1$ (if not, we can shorten the path by making it contain $e$). Let $p$ be any point in the interior of $e$. We replace the subpath $\pi(a_1,p)$ by an L-shaped path $\overline{a_1q_1}\cup \overline{q_1p}$ to obtain a new \st\ path $\pi'$. Clearly, $\pi'$ is homotopic to $\pi$. Below we argue that $\pi'$ is also an optimal \st\ path by showing that $L_1(\pi')\leq L_1(\pi)$ and $L_d(\pi')\leq L_d(\pi)$.

Indeed, since $\overline{a_1q_1}\cup \overline{q_1p}$ is L-shaped, we have $L_1(\pi')\leq L_1(\pi)$. Next, we show that $L_d(\pi')\leq L_d(\pi)$. By the definition of $a_1$, the segment $\overline{a_1a_1'}$ goes rightwards into the interior of $\calC$ from $a_1$. Thus, the subpath $\pi(a_1,p)$ contains at least two turns. The L-shaped path $\overline{a_1q_1}\cup \overline{q_1p}$ contains one turn. However, adding it to $\pi'$ may introduce another turn at $a_1$. Note that there is no additional turn at $p$. To see this, the last segments of both $\overline{a_1q_1}\cup \overline{q_1p}$ and $\pi(a_1,p)$ are horizontal since $p$ is an interior point of $e$ and $e$ is an edge of both $\pi$ and $\pi'$. This implies that $L_d(\pi')\leq L_d(\pi)$.

This proves that $\pi'$ is an optimal \st\ path.

\item
If $a_1$ is below $q_1$, then depending on whether $\overline{aa_1}$ is vertical, there are two subcases.
\begin{enumerate}
\item
Suppose $\overline{aa_1}$ is vertical  (e.g., see Fig.~\ref{fig:cortcp}). By the definition of $w_1$, $\pi$ must intersect a point $p$ on the window $w_1$. We replace $\pi(a_1,p)$ by an L-shaped path $\overline{a_1q_1}\cup \overline{q_1p}$ to obtain a new \st\ path $\pi'$. Clearly, $\pi'$ is homotopic to $\pi$. We argue that $\pi'$ is also an optimal \st\ path by showing that $L_1(\pi')\leq L_1(\pi)$ and $L_d(\pi')\leq L_d(\pi)$. Similar to the above case, $L_1(\pi')\leq L_1(\pi)$ holds. Below, we show that $L_d(\pi')\leq L_d(\pi)$.

Because $a_1$ is strictly below $q_1$ and $\overline{a_1a_1'}$ goes rightwards into the interior of $\calC$, the subpath $\pi(a_1,p)$ contains at least two turns (including the one at $a_1$) in $\pi$. On the other hand, the L-shaped path $\overline{a_1q_1}\cup \overline{q_1p}$ introduces at most two turns to $\pi'$: one at $q_1$ and the other possibly at $p$ (note that there is no turn at $a_1$). This implies that $L_d(\pi')\leq L_d(\pi)$.

\item
Suppose $\overline{aa_1}$ is horizontal   (e.g., see Fig.~\ref{fig:cortcp2}). In this case $a$ must be to the left of $a_1$ since $a$ is outside the corridor $\calC$. Hence,
$\overline{aa_1'}$ is the segment of $\pi$ consisting of both
$\overline{aa_1}$ and $\overline{a_1a_1'}$. Let $\overline{a_1'a_1''}$ be the vertical segment incident to $a_1'$. One can verify that $a_1''$ must be above $a_1'$ since otherwise $\pi$ would not be an optimal \st\ path. Again, $\pi$ must intersect the window $w_1$ at a point $p$.

\begin{figure}[t]
\begin{minipage}[t]{0.49\linewidth}
\begin{center}
\includegraphics[totalheight=1.9in]{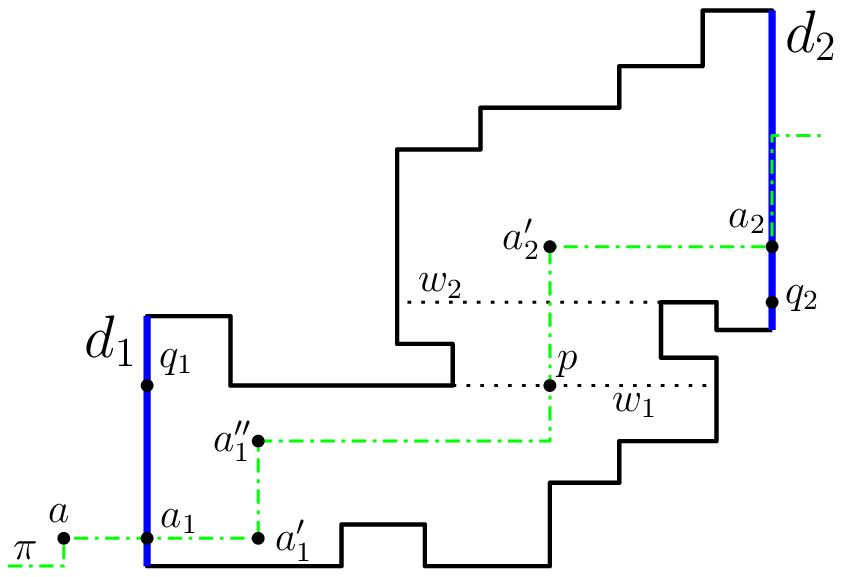}
\caption{\footnotesize Illustrating the case where $\overline{aa_1}$ is horizontal and $\overline{a_1'a_1''}$ does not intersect $w_1$.
}
\label{fig:cortcp2}
\end{center}
\end{minipage}
\hspace*{0.05in}
\begin{minipage}[t]{0.49\linewidth}
\begin{center}
\includegraphics[totalheight=1.9in]{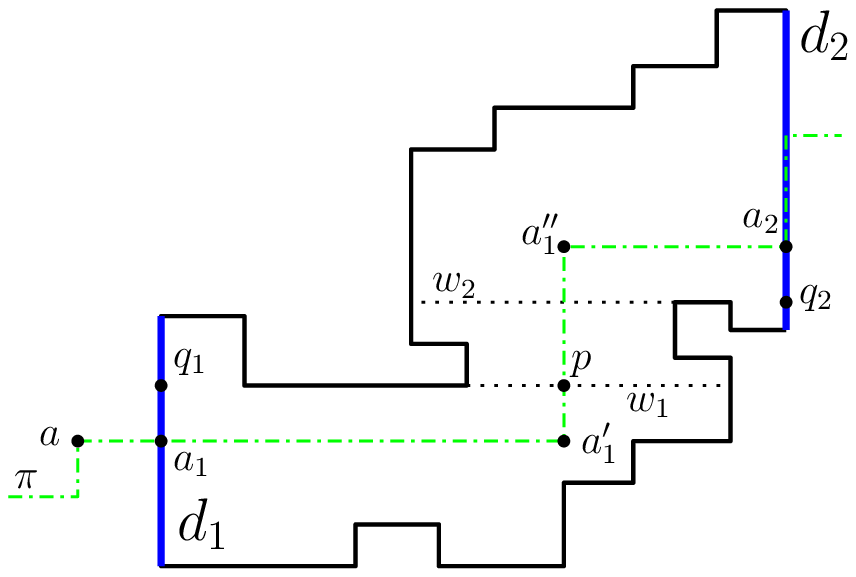}
\caption{\footnotesize Illustrating the case where $\overline{aa_1}$ is horizontal and $\overline{a_1'a_1''}$ intersects $w_1$ (at $p$).
}
\label{fig:tcp3}
\end{center}
\end{minipage}
\vspace*{-0.15in}
\end{figure}

\begin{enumerate}
\item

If $\overline{a_1'a_1''}$ does not intersect $w_1$ (e.g., see Fig.~\ref{fig:cortcp2}),
we replace the subpath $\pi(a_1,p)$ by the L-shaped path $\overline{a_1q_1}\cup \overline{q_1p}$ to obtain another \st\ path $\pi'$. By similar analysis as above, one can verify that $L_1(\pi')\leq L_1(\pi)$ and $L_d(\pi')\leq L_d(\pi)$. We omit the details. Thus, $\pi'$ is also an optimal path.

%
%

\item
If $\overline{a_1'a_1''}$ intersects $w_1$, we let $\pi'=\pi$.
%
%
\end{enumerate}
\end{enumerate}
\end{enumerate}

As a summary, the above obtains an optimal \st\ path $\pi'$, and there are two cases for $\pi'$: either $\pi'$ contains $q_1$, or $a_1$ is below $q_1$ and $\overline{a_1'a_1''}$ intersects the window $w_1$.


Let $b$ refer to the point $q_1$ if $\pi'$ contains $q_1$ and refer to $a_1$ otherwise. Let $\overline{bb'}$ denote the last segment of the subpath $\pi'(s,b)$.


\subsubsection{Obtaining another Optimal \st\ Path $\pi_5$}
In the sequel, we obtain another optimal \st\ path $\pi_5$ by modifying the subpath $\pi'(s,b)$, so that $\pi_5$ has certain special properties that will facilitate our analysis later.

Let $\pi_1=\pi'(s,b)$. The following analysis follows the similar scheme as in the proof of Lemma~\ref{lem:target}.

We assume $\pi_1$ travels through at least one open corridor since otherwise the analysis would be similar (but easier).
Suppose $\pi_1$ travels through an open corridor $\calC$. Hence, $\pi_1$ crosses both doors of $\calC$. Let $p_1$ and $p_2$ be the points on the two doors of $\calC$, respectively, such that the segment incident to $p_1$ and the segment incident to $p_2$ in the subpath $\pi_1(p_1,p_2)$ are both horizontal (and thus perpendicular to the doors). Since $\pi_1$ travels through $\calC$, such two points $p_1$ and $p_2$ must exist. We replace $\pi_1(p_1,p_2)$ by the canonical path $\pi(\calC,p_1,p_2)$ to obtain another path $\pi_1'$, and one can verify that $L_1(\pi_1')\leq L_1(\pi')$ and $L_d(\pi_1')\leq L_d(\pi')$, and thus $\pi_1'$ is still an optimal \st\ path. Note that $\pi_1'$ still contains the point $b$ because $b$ is not in any open corridors.

We do the above for all open corridors traveled through by $\pi_1$.  Let $\pi_2$ denote the new optimal \st\ path. Note that for each horizontal segment of $\pi_2(s,b)$, if it intersects the interior of an open corridor, then it must intersect both doors of the corridor.

Suppose we traverse $\pi_2(s,b)$ from $s$ to $b$. If $\pi_2(s,b)$ intersects a junction rectangle $R$, then let $p_1$ and $p_2$ be the first and last points $\pi_2(s,b)$ intersecting $R$, respectively. We obtain another \st\ path $\pi_2'$ by replacing $\pi_2(p_1,p_2)$ with an L-shaped path connecting $p_1$ and $p_2$ such that $L_d(\pi_2')=L_d(\pi_2)$. Note that such an  L-shaped path must exist. Clearly, $L_1(\pi_2')=L_1(\pi_2)$. Hence, $\pi_2'$ is also an optimal \st\ path.

We do the above for all junction rectangles intersected by $\pi_2(s,b)$, and let $\pi_3$ be the resulting path. Note that each vertical segment of $\pi_3$ must be on a vertical side of a junction rectangle unless it is incident to $s$. Also note that $\pi_3$ still contains $b$.

As shown in the proof of Lemma~\ref{lem:target}, any subpath of $\pi_3(s,b)$ partitioned by the points of $\calB$ on $\pi_3(s,b)$ must be a staircase path. Consider any such a staircase subpath $\pi_3(b_1,b_2)$, where $b_1$ and $b_2$ are the two endpoints. We obtain a pushed path in the same way as in the proof of Lemma~\ref{lem:target}. We do this for all subpaths of $\pi_3(s,b)$ and let $\pi_5$ be the resulting path (we use $\pi_5$ instead of $\pi_4$ to be consistent with the proof of Lemma~\ref{lem:target}), which is still an optimal \st\ path. Again, as shown in the proof of Lemma~\ref{lem:target}, for any segment of $\pi_5(s,b)$, it must contain a point of $\calB$ or it is incident to $b$. Hence, each subpath of $\pi_5(s,b)$ partitioned by the points of $\calB$ must be an L-shaped path. Consider any such subpath $\pi_5(b_1,b_2)$  of $\pi_5(s,b)$.
In the following, we argue the correctness of our algorithm on the subpath $\pi_5(b_1,b_2)$.

\subsubsection{Analyzing the Subpath $\pi_5(b_1,b_2)$}

We first discuss the case where $b_2\neq b$, i.e., it is not the last subpath of $\pi_5(s,b)$.

Without loss of generality, we assume that $b_2$ is to the northeast of $b_1$ and the segment of $\pi_5(b_1,b_2)$ incident to $b_1$ is vertical. As shown in the proof of Lemma~\ref{lem:target}, $G(\calB)$ has a staircase path $\pi_{G(\calB)}(b_1,b_2)$ connecting $b_1$ to $b_2$ and the region between the two paths $\pi_5(b_1,b_2)$ and $\pi_{G(\calB)}(b_1,b_2)$ is empty (because the two paths are homotopic). Hence, when the algorithm processes the horizontal edges of $\pi_{G(\calB)}(b_1,b_2)$, they can be dragged upwards to form $\pi_5(b_1,b_2)$ without hitting any vertices of $\calP$ (similar to the example in Fig.~\ref{fig:Lshaped}).

Next, we discuss the case where $b_2=b$. Recall that $b$ may be either $q_1$ or $a_1$, and if $b=a_1$, then $b$ is on $d_1$ below $q_1$ and $\overline{a_1'a_1''}$ intersects the window $w_1$. In the sequel, we first show that by the dragging operations,  our algorithm will obtain a particular path, denoted by $\pi(s,q_1)$, and later we will use $\pi(s,q_1)$ to argue the correctness of our through-corridor-path generating operations.

If $b=q_1$, then $\pi_5$ contains $q_1$. By the same argument as above and using the  dragging operations, we can obtain $\pi_5(b_1,b_2)$, and thus obtain $\pi_5(s,q_1)$ as well. In this case, we use $\pi(s,q_1)$ to refer to  $\pi_5(s,q_1)$.

If $b=a_1$, then let $\overline{b_2'b_2}$ be the segment of $\pi_5(b_1,b_2)$ incident to $b_2$ (e.g., see Fig.~\ref{fig:vertical}). Depending on whether $\overline{b_2'b_2}$ is horizontal or vertical, there are two cases.

\begin{figure}[t]
\begin{minipage}[t]{0.49\linewidth}
\begin{center}
\includegraphics[totalheight=1.9in]{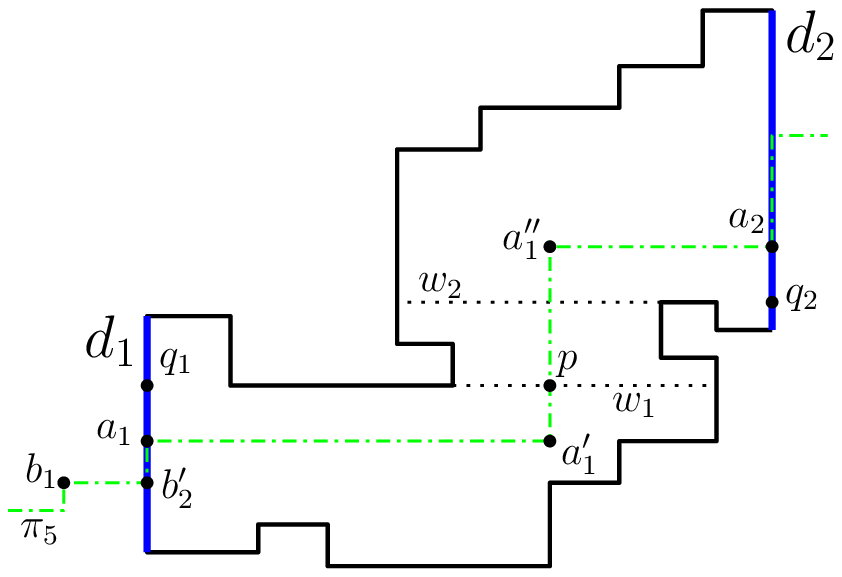}
\caption{\footnotesize Illustrating the case where $\overline{b_2'b_2}$ is vertical. Note that $b_2=b=a_1$. The green dashed dotted path is $\pi_5$.
}
\label{fig:vertical}
\end{center}
\end{minipage}
\hspace*{0.05in}
\begin{minipage}[t]{0.49\linewidth}
\begin{center}
\includegraphics[totalheight=1.9in]{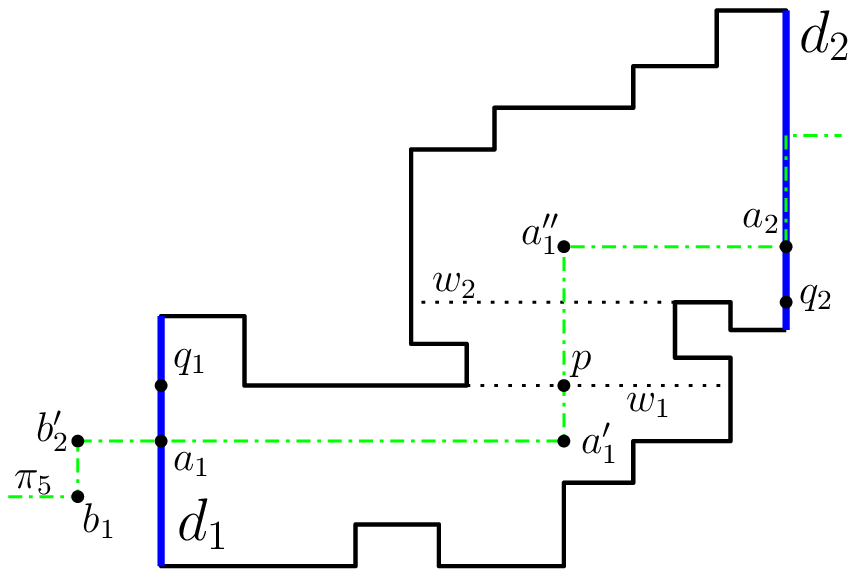}
\caption{\footnotesize Illustrating the case where $\overline{b_2'b_2}$ is horizontal. Note that $b_2=b=a_1$.
}
\label{fig:horizontal}
\end{center}
\end{minipage}
\vspace*{-0.15in}
\end{figure}

If $\overline{b_2'b_2}$ is vertical (e.g., see Fig.~\ref{fig:vertical}), then $\pi^*(b_1,q_1)$ is also L-shaped, where $\pi^*(b_1,q_1)$ is defined to be the concatenation of $\pi_5(b_1,b_2)$ and $\overline{b_2q_1}$. Hence, $G(\calB)$ also has a staircase path $\pi_{G(\calB)}(b_1,q_1)$ connecting $b_1$ to $q_1$. As argued above, by performing the dragging operations on the edges of $\pi_{G(\calB)}(b_1,q_1)$, we can obtain $\pi^*(b_1,q_1)$ and thus obtain a path $\pi(s,q_1)$ from $s$ to $q_1$ that is a concatenation of $\pi_5(s,b_2)$ and $\overline{b_2q_1}$. In this case, we use $\pi(s,q_1)$ to refer to the concatenation of $\pi_5(s,b_2)$ and $\overline{b_2q_1}$.

If $\overline{b_2'b_2}$ is horizontal (e.g., see Fig.~\ref{fig:horizontal}), then $\overline{b_2'a_1'}=\overline{b_2'a_1}\cup \overline{a_1a_1'}$ is a single segment of $\pi_5$. We push $\overline{b_2'a_1'}$ upwards until we hit an obstacle vertex $v$. With a little abuse of notation, we still use $\pi_5$ to denote the new path (which is still an optimal \st\ path) after the push operation, and use $a_1$, $a_1'$, $b_2'$, and $b_2$ to refer to the corresponding new points in the new path. Recall that $\calC$ is the closed corridor that has $d_1$ as a door.
Depending on whether $v$ is in the corridor $\calC$ or not, there are two cases.

\begin{enumerate}
\item
If $v$ is in $\calC$, then since $\overline{a_1'a_1''}$ intersects $w_1$, $v$ must be on the extension of the window $w_1$ and $b_2$ ($=a_1$) is at $q_1$ now. By using the same argument as before, we can obtain  $\pi_5(s,q_1)$ by applying the dragging operations. In this case, $\pi(s,q_1)$ refers to $\pi_5(s,q_1)$.

\item
If $v$ is not in $\calC$, then we claim that $v$ must be on the upper bridge of some open corridor. Indeed, recall that the vertical segment $\overline{b_1b_2'}$ must be on the right side of a junction rectangle. Note that $b_2$ is on the right side of a different junction rectangle. Hence, $\overline{b_2'b_2}$ must travel through some open corridors and $v$ must be at the lowest upper bridge of one of such open corridors. Let $\calC'$ denote the open corridor whose upper bridge contains $v$. Since $\overline{b_2'b_2}$ travels through $\calC'$, $\overline{b_2'b_2}$ contains a backbone point $p$ that is on a door of $\calC'$ (in fact it contains two such backbone points, but one is enough for our argument). The point $p$ breaks the path $\pi_5(b_1,b_2)$ into two subpaths $\pi_5(b_1,p)$ and $\pi_5(p,b_2)=\overline{pb_2}$. Note that $\pi_5(b_1,p)$ is an L-shaped path and $\pi_5(p,b_2)$ is a horizontal segment.

Let $\pi^*(p,q_1)=\overline{pb_2}\cup \overline{b_2q_1}$, which is an L-shaped path. Hence, by using the dragging operations, our algorithm will obtain the path $\pi_5(b_1,p)$ and the path $\pi^*(p,q_1)$, and thus obtain the path $\pi_5(s,b_2)\cup\overline{b_2q_1}$. In this case, we use $\pi(s,q_1)$ to refer to the path $\pi_5(s,b_2)\cup\overline{b_2q_1}$.
\end{enumerate}

As a summary, the above shows that after our algorithm processes the edges of the target path from $s$ to $q_1$ by applying the dragging operations, a path $\pi(s,q_1)$ will be computed at $q_1$ with the following property: if $\pi_5$ contains $q_1$, then $\pi(s,q_1)=\pi_5(s,q_1)$; otherwise, $\pi(s,q_1)$ is $\pi_5(s,a_1)\cup \overline{a_1q_1}$, $a_1$ is below $q_1$ on $d_1$, and $\overline{a_1'a_1''}$
intersects the window $w_1$.

In the following, we argue the correctness of our algorithm on processing the corridor edge $e(q_1,q_2)$ by applying the through-corridor-path generating operation.
Depending on whether $q_1$ is in $\pi_5$, there two main cases as discussed above. We will show that in either case, after the operation, we will obtain a path $\pi(s,q_2)$ with the following property: if we apply the dragging operation on the last edge of $\pi(s,q_2)$ and $\overline{q_2a_2}$ (which is a path of $G(\calB)$), then we can obtain a path $\pi(s,a_2)$ from $s$ to $a_2$ such that the concatenation of $\pi(s,a_2)$ and $\pi_5(a_2,t)$ is an optimal \st\ path, which implies that storing $\pi(s,q_2)$ at $q_2$ is sufficient for obtaining an optimal \st\ path (this further implies the correctness of our through-corridor-path generating operation).

\subsubsection{The First Main Case: $q_1\in \pi_5$}

We begin with the first case where $\pi_5$ contains $q_1$. In
this case, $\pi(s,q_1)=\pi_5(s,q_1)$.
Let $\alpha=\overline{pq_1}$ be the last segment of $\pi_5(s,q_1)$.
Depending on whether $\alpha$ is horizontal or vertical, there are two
cases.

\paragraph{The horizontal case.}
If $\alpha$ is horizontal (e.g., see Fig.~\ref{fig:tcp1}), then $p$ is to the left of $q_1$. In this case, according to our through-corridor-path generating operation,
$\pi(s,q_2)=\pi(s,q_1)\cup \pi_1(\calC,q_1,q_2)$. By applying a dragging operation on the last segment of $\pi(s,q_2)$, we obtain a path $\pi(s,a_2)$ from $s$ to $a_2$, as follows.

Recall the intervals $I_1(q_1,d_2)$ and $I_2(q_1,d_2)$ defined in Section~\ref{sec:compute} (e.g., see Fig.~\ref{fig:intequal}). By the definition of $q_1$, $I_1(q_1,d_2)=I_2(q_1,d_2)$~\cite{ref:SchuiererAn96}. Also recall that
the last segment of $\pi_1(\calC,q_1,q_2)$ is the interval $I_1'(q_1,d_2)$
on the window $w_2$.

If $a_2$ is in the interval $I_1(q_1,d_2)$, then we simply push the last segment of $\pi(s,q_2)$ upwards until $a_2$. Otherwise, we let
$\pi(s,a_2)=\pi(s,q_1)\cup \overline{q_2a_2}$ (i.e., add a vertical segment $\overline{q_2a_2}$ to connect $\pi(s,q_2)$ with $a_2$). Note that the above way of constructing $\pi(s,a_2)$ in either case is consistent with applying the dragging operation on the last segment of $\pi(s,q_2)$ and $\overline{q_2a_2}$. In the latter case, for the purpose of the argument, we conceptually add a horizontal segment of zero length to the end of $\pi(s,a_2)$ to connect $a_2$
such that the last segment of $\pi(s,a_2)$ is also horizontal, and this makes it  consistent with the path $\pi_5(s,a_2)$, whose last segment is also horizontal. Our goal is to show that $\pi(s,a_2)\cup \pi_5(a_2,t)$ is also an optimal \st\
path. To this end, in either case, due to that the last segments of both
$\pi(s,a_2)$ and $\pi_5(s,a_2)$ are horizontal, it is sufficient to prove $L_1(\pi(s,a_2))\leq
L_1(\pi_5(s,a_2))$ and $L_d(\pi(s,a_2))\leq L_d(\pi_5(s,a_2))$.

\begin{figure}[t]
\begin{minipage}[t]{\linewidth}
\begin{center}
\includegraphics[totalheight=1.9in]{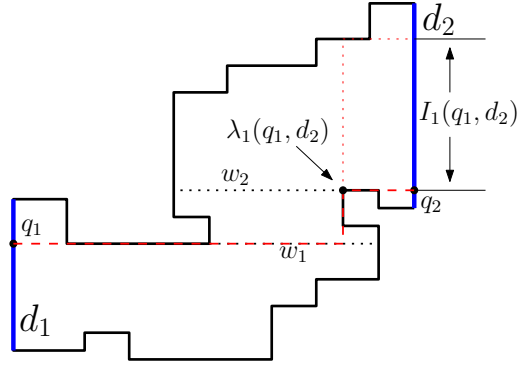}
\caption{\footnotesize Illustrating the intervals $I_1(q_1,d_2)$ and $I_2(q_1,d_2)$, which are equal. The segment $\overline{\lambda_1(q_1,d_2)q_2}$ on the extension of $w_2$ is the interval $I_1'(q_1,d_2)$
}
\label{fig:intequal}
\end{center}
\end{minipage}
\vspace*{-0.15in}
\end{figure}

First of all, by the definition of $\pi_1(\calC,q_1,q_2)$, the subpath of
$\pi(s,a_2)$ between $q_1$ and $a_2$ is a shortest path from $q_1$ to
$a_2$ in $\calC$~\cite{ref:SchuiererAn96}. Thus, $L_1(\pi(s,a_2))\leq L_1(\pi_5(s,a_2))$. In the following,
we prove $L_d(\pi(s,a_2))\leq L_d(\pi_5(s,a_2))$. Let $\pi(q_1,a_2)$ be the subpath of $\pi(s,a_2)$ between $q_1$ and $a_2$. Depending on whether $a_2$ is in $I_1(q_1,d_2)$, there are two subcases.

\begin{enumerate}
\item
If $a_2$ is in $I_1(q_1,d_2)$, then $L_d(\pi(s,a_2))=L_d(\pi(s,q_1))+L_d(\pi(q_1,a_2))-1$, where the minus one is because the path $\pi(s,a_2)$ does not make a turn at $q_1$ (since $\alpha$ is horizontal). As $a_2\in I_1(q_1,d_2)$, $L_d(\pi(q_1,a_2))=L_d(q_1,d_2)$.

On the other hand, $L_d(\pi_5(s,a_2))\geq L_d(\pi_5(s,q_1))+L_d(\pi_5(q_1,a_2))-1$. Since $\pi(s,q_1)=\pi_5(s,q_1)$, $L_d(\pi_5(s,q_1))=L_d(\pi(s,q_1))$. We claim that $L_d(\pi_5(q_1,a_2))\geq L_d(q_1,d_2)$. Indeed, because the last segment of $\pi_5(q_1,a_2)$ is horizontal and $a_2\in d_2$, $\pi_5(q_1,a_2)$ is an admissible path from $q_1$ to $d_2$. Thus, by the definition of $L_d(q_1,d_2)$, $L_d(\pi_5(q_1,a_2))\geq L_d(q_1,d_2)$ holds.

Therefore, we obtain $L_d(\pi(s,a_2))\leq L_d(\pi_5(s,a_2))$.

\item
If $a_2$ is not in $I_1(q_1,d_2)$, then $a_2$ is not on $I_2(q_1,d_2)$ either because $I_1(q_1,d_2)=I_2(q_1,d_2)$~\cite{ref:SchuiererAn96}. By the definition of $I_2(q_1,d_2)$, we obtain
$L_d(\pi_5(q_1,a_2))\geq L_d(q_1,d_2)+2$. Thus, $L_d(\pi_5(s,a_2))\geq L_d(\pi_5(s,q_1))+L_d(\pi_5(q_1,a_2))-1\geq L_d(\pi_5(s,q_1))+ L_d(q_1,d_2)+1$.

On the other hand, $L_d(\pi(q_1,a_2))=L_d(\pi(q_1,q_2))+2=L_d(\pi(q_1,d_2)+2$ (we have ``$+2$'' instead of ``$+1$'' because there is a horizontal segment of zero length at the end of $\pi(q_1,a_2)$). Hence, $L_d(\pi(s,a_2))=L_d(\pi(s,q_1))+L_d(\pi(q_1,a_2))-1=L_d(\pi(s,q_1))+L_d(\pi(q_1,d_2))+1$.

Therefore, we also obtain $L_d(\pi(s,a_2))\leq L_d(\pi_5(s,a_2))$.

\end{enumerate}

\paragraph{The vertical case.}
Next we consider the case where $\alpha=\overline{pq_1}$ is vertical. We obtain our path $\pi(s,a_2)$ in a similar way as before. As in the above horizontal case, one can verify that $L_1(\pi(s,a_2))=L_1(\pi_5(s,a_2))$ holds in all cases below, and thus we will focus on proving $L_d(\pi(s,a_2))\leq L_d(\pi_5(s,a_2))$.
Depending on whether $p$ is above $q_1$, there are two subcases.

If $p$ is above $q_1$, then both $\pi(s,a_2)$ and $\pi_5(s,a_2)$ make a turn at $q_1$.
Thus, $L_d(\pi(s,a_2))=L_d(\pi(s,q_1))+L_d(\pi(q_1,a_2))$ and $L_d(\pi_5(s,a_2))= L_d(\pi_5(s,q_1))+L_d(\pi_5(q_1,a_2))$. The rest of the analysis is similar as the above and we omit the details.

If $p$ is below $q_1$, then depending on whether $p$ is on $d_1$, there are further two subcases.

If $p$ is not on $d_1$, then again both $\pi(s,a_2)$ and $\pi_5(s,a_2)$ make a turn at $q_1$. We also have $L_d(\pi(s,a_2))=L_d(\pi(s,q_1))+L_d(\pi(q_1,a_2))$ and $L_d(\pi_5(s,a_2))= L_d(\pi_5(s,q_1))+L_d(\pi_5(q_1,a_2))$. The rest of the analysis is similar as before and we omit the details.

In the following, we assume that $p$ is on $d_1$. In this case,
according to our through-corridor-path dragging operation (e.g., see Fig.~\ref{fig:tcp2}), $\pi(s,q_2)=\pi(s,p)\cup \pi_{1}(\calC,p,q_2)$, which may not contain $q_1$. By the definition of $q_2$, the last segment of $\pi_{1}(\calC,p,q_2)$ is orthogonal to $d_2$ and $L_d(\pi_{1}(\calC,p,q_2))=L_d(p,d_2)$~\cite{ref:SchuiererAn96}.

Note that $\pi_5(s,a_2)$, which contains $q_1$, makes a turn at $p$ and another turn at $q_1$. Hence, $L_d(\pi_5(s,a_2))= L_d(\pi_5(s,p))+2+L_d(\pi_5(q_1,a_2))\geq L_d(\pi_5(s,p))+2+L_d(q_1,d_2)$.
Observe that $L_d(\pi(p,d_2))\leq 1+L_d(q_1,d_2)$ because for any path from $q_1$ to $d_2$, we can always add $\overline{pq_1}$ to obtain a path from $p$ to $d_2$.

As discussed before, either $I_1(p,d_2)=I_2(p,d_2)$ or $I_1(p,d_2)\subset I_2(p,d_2)$~\cite{ref:SchuiererAn96}. Depending on whether $a_2$ is in $I_1(p,d_2)$, $I_2(p,d_2)$, or not,
there are three cases.

\begin{enumerate}
\item
If $a_2$ is in the interval $I_1(p,d_2)$, then $L_d(\pi(s,a_2))\leq
L_d(\pi(s,p))+L_d(\pi(p,a_2))+1=L_d(\pi_5(s,p))+L_d(p,d_2)+1\leq L_d(\pi_5(s,p))+L_d(q_1,d_2)+2\leq L_d(\pi_5(s,a_2))$.

\item
If $a_2$ is in $I_2(p,d_2)$ but not in $I_1(p,d_2)$, this implies
$I_1(p,d_2)\subset I_2(p,d_2)$. According to Schuierer~\cite{ref:SchuiererAn96}, the first segment of $\pi_{1}(\calC,p,q_2)$ is parallel to the window $w_1$, which is horizontal (e.g., see Fig.~\ref{fig:tcp2}), and further, $L_d(p,d_2)=L_d(q_1,d_2)$. Hence, our path $\pi(s,a_2)$ does not have a turn at $p$ and $L_d(\pi(p,a_2))=L_d(\pi(p,d_2))+2=L_d(q_1,d_2)+2$. Therefore, $L_d(\pi(s,a_2))= L_d(\pi(s,p))+L_d(\pi(p,a_2))=L_d(\pi_5(s,p))+L_d(q_1,d_2)+2\leq L_d(\pi_5(s,a_2))$.

\item
Suppose $a_2$ is not in $I_2(p,d_2)$. 
Since $L_d(\pi(p,d_2))\leq 1+L_d(q_1,d_2)$, $I_1(q_1,d_2)\subseteq
I_2(p,d_2)$ (in fact they are equal~\cite{ref:SchuiererAn96}). Hence, $a_2$ is not in $I_1(q_1,d_2)$.
Since $I_1(q_1,d_2)=I_2(q_1,d_2)$, $a_2$ is not in $I_2(q_1,d_2)$ either. Thus, $L_d(\pi_5(q_1,a_2))\geq L_d(q_1,d_2)+2$.
Therefore, $L_d(\pi_5(s,a_2))= L_d(\pi_5(s,p))+2+L_d(\pi_5(q_1,a_2))\geq
L_d(\pi_5(s,p))+L_d(q_1,d_2)+4$.

For our path $\pi(s,a_2)$, we have $L_d(\pi(s,a_2))\leq
L_d(\pi(s,p))+L_d(\pi(p,a_2))+1$, where ``$+1$'' is due to a possible turn at $p$. Since $a_2\not\in I_2(p,d_2)$, $L_d(\pi(p,a_2))\leq L_d(p,d_2)+2$.
Recall that $L_d(\pi(s,p))=L_d(\pi_5(s,p))$ and
$L_d(p,d_2)\leq 1+L_d(q_1,d_2)$. Hence, we obtain
$L_d(\pi(s,a_2))\leq L_d(\pi_5(s,p))+L_d(q_1,d_2)+1+2+1\leq L_d(\pi_5(s,a_2))$.

\end{enumerate}

Thus, in any case it holds that $L_d(\pi(s,a_2))\leq L_d(\pi_5(s,a_2))$.

\subsubsection{The Second Main Case: $q_1\not\in \pi_5$}

We then consider the second main case where $\pi_5$ does not contain
$q_1$. In this case, $\pi(s,q_1)=\pi_5(s,a_1)\cup \overline{a_1q_1}$, $a_1$ is below $q_1$ on $d_1$, and $\overline{a_1'a_1''}$ intersects the window $w_1$ (e.g., see Fig.~\ref{fig:secondmain}). Hence, the
last segment of $\pi(s,q_1)$ is $\overline{a_1q_1}$, which is
vertical. Since $a_1$ is below $q_1$ and is on the diagonal $d_1$,
according to our corridor-path generating operation,
$\pi(s,q_2)=\pi(s,a_1)\cup \pi_{1}(\calC,a_1,q_2)$. We obtain our path $\pi(s,a_2)$ in the same way as before. Our goal is to show that $L_1(\pi(s,a_2))\leq L_1(\pi_5(s,a_2))$ and
$L_d(\pi(s,a_2))\leq L_d(\pi_5(s,a_2))$. Similarly as before, by the definition of $\pi_{1}(\calC,a_1,q_2)$, the subpath of $\pi(s,a_2)$ between $a_1$ and $a_2$ is a shortest path from $a_1$ to $a_2$ in $\calC$~\cite{ref:SchuiererAn96}, and thus, it holds that $L_1(\pi(s,a_2))\leq L_1(\pi_5(s,a_2))$. In what follows,
we focus on proving $L_d(\pi(s,a_2))\leq L_d(\pi_5(s,a_2))$.

\begin{figure}[t]
\begin{minipage}[t]{\linewidth}
\begin{center}
\includegraphics[totalheight=1.9in]{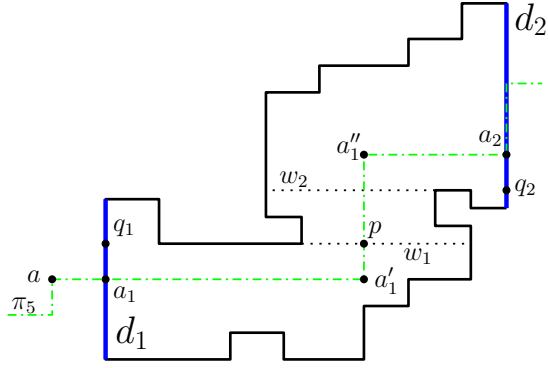}
\caption{\footnotesize Illustrating the intervals $I_1(q_1,d_2)$ and $I_2(q_1,d_2)$, which are equal.
}
\label{fig:secondmain}
\end{center}
\end{minipage}
\vspace*{-0.15in}
\end{figure}

Due to that $\overline{a_1'a_1''}$ intersects $w_1$,
according to \cite{ref:SchuiererAn96}, $L_d(a_1,d_2)=L_d(q_1,d_2)$ and the first segment of $\pi_{1}(\calC,a_1,q_2)$ must be parallel to the window $w_1$ and thus is horizontal, which implies that there is no turn at $a_1$ in our path $\pi(s,a_2)$ (because the last segment of $\pi_5(s,a_1)$ is horizontal). Hence, $L_d(\pi(s,a_2)) =
L_d(\pi(s,a_1))+L_d(\pi(a_1,a_2))$.
Note that $L_d(\pi_5(s,a_2))\geq L_d(\pi_5(s,a_1))+L_d(\pi_5(a_1,a_2))\geq
L_d(\pi_5(s,a_1))+L_d(a_1,d_2)$, and $L_d(\pi(s,a_1))=L_d(\pi_5(s,a_1))$.

If $a_2\in I_1(a_1,d_2)$, then $L_d(\pi(a_1,a_2))=L_d(a_1,d_2)$ and
$L_d(\pi(s,a_2))=L_d(\pi(s,a_1))+L_d(a_1,d_2)$.  Hence,
$L_d(\pi(s,a_2))\leq L_d(\pi_5(s,a_2))$.

Next we consider the case where $a_2\not\in I_1(a_1,d_2)$. In this
case, we have $L_d(\pi(s,a_2))=L_d(\pi(s,a_1))+L_d(a_1,d_2)+2$. Depending
on whether $a_2$ is in $I_2(a_1,d_2)$, there are further two subcases.

If $a_2\in I_2(a_1,d_2)$, then $I_1(a_1,d_2)\neq
I_2(a_1,d_2)$.
According to \cite{ref:SchuiererAn96}, $I_2(a_1,d_2)=I_1(q_1,d_2)$, and
for any path from $a_1$ to $a_2$ of $L_d(a_1,d_2)+1$ links with the last link orthogonal to $d_2$ (i.e., horizontal), the first
segment of the path must be vertical. Recall that by the definition of the point $a_1$, the first segment of $\pi_5(a_1,a_2)$ is horizontal. This implies that
$L_d(\pi_5(a_1,a_2))>L_d(a_1,d_2)+1$, i.e.,
$L_d(\pi_5(a_1,a_2))\geq L_d(a_1,d_2)+2$. Hence, we obtain $L_d(\pi_5(s,a_2))\geq L_d(\pi_5(s,a_1)) + L_d(\pi_5(a_1,a_2))\geq L_d(\pi_5(s,a_1))+L_d(a_1,d_2)+2\geq L_d(\pi(s,a_2))$.

If $a_2\not\in I_2(a_1,d_2)$, then $L_d(\pi_5(a_1,a_2)) \geq L_d(a_1,d_2)+2$. As
above, $L_d(\pi(s,a_2))\leq L_d(\pi_5(s,a_2))$ still holds.


\subsubsection{Wrapping Things Up}

The above shows the correctness of our through-corridor-path generating operations:
if $\pi(s,q_1)$ is stored at $q_1$, then
we can extend $\pi(s,q_1)$ through the corridor $\calC$ to obtain $\pi(s,q_2)$,
which can further be used to obtain an optimal \st\ path by applying the dragging operation on the last segment.

Next we argue that if the path $\pi(s,q_1)$ is not stored at $q_1$, then there must exist another path $\pi'(s,q_1)$ stored at $q_1$ that can also be used to obtain an optimal \st\ path by applying the through-corridor-path generating operation on $\pi'(s,q_1)$.

We first discuss the minimum-link shortest paths. We again consider the two main cases depending on whether $\pi_5$ contains $q_1$.

\begin{enumerate}
\item
We first consider the case where $\pi_5$ contains $q_1$.
Suppose there exists a path $\pi'(s,q_1)$ from $s$ to $q_1$ that causes $\pi(s,q_1)$ not stored at $q_1$.
By Rule ($a_1$), one of the following two cases must happen: either $L_1(\pi'(s,q_1))<L_1(\pi(s,q_1))$,  or
$L_1(\pi'(s,q_1))=L_1(\pi(s,q_1))$ but $L_d(\pi'(s,q_1))\leq
L_d(\pi(s,q_1))-2$.
We apply the through-corridor-path generating operation on $\pi'(s,q_1)$ to obtain a path $\pi'(s,q_2)$ from $s$ to $q_2$, and then obtain a path $\pi'(s,a_2)$ from $s$ to $a_2$ by applying the dragging operation, in the same way as before when we obtained $\pi(s,a_2)$ from $\pi(s,q_2)$ .

\begin{enumerate}
\item
If the first case happens, then $L_1(\pi'(s,a_2))<L_1(\pi(s,a_2))$. If we concatenate $\pi'(s,a_2)$ with $\pi_5(a_2,t)$, we would obtain another \st\ path whose length is strictly smaller than that of $\pi_5$, contradicting with that $\pi_5$ is a (minimum-link) shortest path.

\item
If the second case happens, then by the similar analysis as in Section~\ref{sec:correctold}, one can verify that $L_1(\pi'(s,a_2))= L_1(\pi(s,a_2))$ and $L_d(\pi'(s,a_2))\leq L_d(\pi(s,a_2))$, and we omit the details. Hence, if we concatenate $\pi'(s,a_2)$ with $\pi_5(a_2,t)$, we can obtain another \st\ path $\pi_5'$ with $L_1(\pi_5')=L_1(\pi_5)$ and $L_d(\pi_5')\leq L_1(\pi_5)$. Therefore, 
we can also obtain a minimum-link shortest \st\ path using $\pi'(s,q_1)$.
\end{enumerate}

\item
If $\pi_5$ does not contain $q_1$, then $\pi(s,q_1)=\pi(s,a_1)\cup
\overline{a_1q_1}$. Suppose there is another path $\pi'(s,q_1)$
that is stored at $q_1$ and causes $\pi(s,q_1)$ not stored at $q_1$.
Again, by Rule ($a_1$), one of the following two cases happens: either $L_1(\pi'(s,q_1))<L_1(\pi(s,q_1))$,  or
$L_1(\pi'(s,q_1))=L_1(\pi(s,q_1))$ but $L_d(\pi'(s,q_1))\leq
L_d(\pi(s,q_1))-2$.

\begin{enumerate}
\item
If the first case happens, then as the above analysis, the concatenation of $\pi'(s,q_1)$, $\pi_1(\calC,q_1,a_2)$, and $\pi_5(a_2,t)$ is an \st\ path whose length is strictly smaller than that of $\pi_5$, contradicting with that $\pi_5$ is a (minimum-link) shortest path.

\item
Suppose the second case happens.
Regardless of whether the last segment of $\pi'(s,q_1)$ is horizontal or vertical, due to the extra ``budget'' 2 on the link distance, one can verify that by applying the through-corridor-path generating operation on $\pi'(s,q_1)$ we can obtain a path $\pi'(s,a_2)$ from $s$ to $a_2$ such that if $\pi_5'=\pi'(s,a_2)\cup \pi_5(a_2,t)$, then $L_1(\pi_5')=L_1(\pi_5)$ and $L_d(\pi_5')\leq L_d(\pi_5)$. Thus, using $\pi'(s,q_1)$, we can also  obtain a minimum-link shortest \st\ path.
\end{enumerate}
\end{enumerate}

Other types of optimal paths can be analyzed in a similar way. We omit the details.

The corridor edges partition the target path $\pi_{G(\calB)}$
into subpaths. The above proves that by applying the dragging
operations on the edges of the first such subpath and applying a
through-corridor-path generating operation on the first corridor edge
$e(q_1,q_2)$, we can obtain a path $\pi(s,q_2)$ such that by applying a dragging operation on its last segment,
we can obtain an optimal \st\ path. For the second
subpath of the target path, we use the similar argument. The only
difference is the following. The first subpath starts from $s$, so we do not
need to argue anything.  However, in the second subpath, we have to show
that there exists a path stored at $q_2$ so that by applying a dragging operation on its last segment we can obtain an optimal path.
But this has been proved above.
Hence, by applying the above analysis on each of the subpaths of
$\pi_{G(\calB)}$, we can prove that our algorithm will find an
optimal \st\ path.

This completes the proof of the correctness of our algorithm.

\subsection{The General Cases}
\label{sec:general}

The above discussed the case where both $s$ and $t$ are in junction
rectangles. In this section, we generalize the approach to other cases.
We begin with the most general case where both $s$ and $t$
are in corridors. Let $\calC_s$ and $\calC_t$ be the two corridors
that contain $s$ and $t$, respectively. We first assume $\calC_s\neq \calC_t$.

Consider a door $d$ of $\calC_s$. We define a point $s_d$ on $d$ as
follows. If $s$ is horizontally visible to $d$, then $s_d$ is the
horizontal projection of $s$ to $d$. Otherwise, let $w$ be the
window of the maximal histogram of $\calC$ with base $d$ such that $w$ separates
$s$ and $d$. We define $s_d$ to be the intersection of $d$ and the
extension of $w$  (e.g., see Fig.~\ref{fig:sd}). With $s_d$ thus defined, for any point $p\in d$, there is a shortest path from $s$ to $p$ in $\calC_s$ that is the union of $\overline{ps_d}$ and $\pi_{opt}(\calC,s_d,s)$, where $\pi_{opt}(\calC,s_d,s)$ is the smallest path between $s$ and $s_d$ in $\calC$~\cite{ref:SchuiererAn96}.
We also call $s_d$ a {\em corridor-connection} point of $s$ on $d$.

\begin{figure}[t]
\begin{minipage}[t]{\linewidth}
\begin{center}
\includegraphics[totalheight=1.3in]{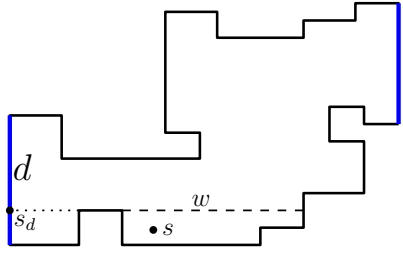}
\caption{\footnotesize Illustrating the definition of $s_d$ on a door $d$.
}
\label{fig:sd}
\end{center}
\end{minipage}
\vspace*{-0.15in}
\end{figure}

Similarly, for each door $d$ of $\calC_t$, we define a
corridor-connection point $t_d$ with respect to $t$ and $\calC_t$.
In this way, there are four corridor-connection points on the
doors of $\calC_s$ and $\calC_t$.

We let $\calB$ now consist of all backbone points and the four corridor-connection points.
We define the graph $G(\calB)$ with respect to $\calB$ in the same way as before. In addition, we add $s$
and $t$ as two new vertices to $G(\calB)$ and also add the following four {\em
corridor-connection} edges. For each corridor-connection point $q$
defined by $s$, we add an edge in $G(\calB)$ connecting $q$ to
$s$. The weight of the edge is $L_1(\pi_{opt}(\calC_s,s,q))$, and the link distance $L_d(\pi_{opt}(\calC_s,s,q))$ as well as the segment of $\pi_{opt}(\calC_s,s,q)$ incident to $q$ is also stored at the edge. Similarly, we add
two corridor-connection edges connecting to $t$. This completes the definition of $G(\calB)$.

\begin{lemma}
There exists an \st\ path $\pi_{G(\calB)}(s,t)$ in $G(\calB)$ that is homotopic to an optimal \st\ path $\pi(s,t)$ and the two paths $\pi_{G(\calB)}(s,t)$ and $\pi(s,t)$ have the same length; we call $\pi_{G(\calB)}(s,t)$ a {\em target path}. This implies that a shortest \st\ path in $G(\calB)$ is also a shortest \st\ path in $\calP$. \end{lemma}
\begin{proof}
Consider any optimal \st\ path $\pi$ in $\calP$. If we traverse on $\pi$ from $s$ to $t$, let $d_s$ be the first door of $\calC_s$ we encounter and let $q_s$ be first point on $d_s$ we encounter. Similarly, if we traverse on $\pi$ from $t$ to $s$, let $d_t$ be the first door of $\calC_t$ we encounter and let $q_t$ be first point on $d_t$ we encounter. Let $p_s$ be the corridor-connection point of $s$ on $d_s$.
Let $p_t$ be the corridor-connection point of $t$ on $d_t$.


Based on $\pi$, we obtain another \st\ path $\pi'$ by replacing the subpath $\pi(s,q_s)$ by $\overline{q_sp_s}\cup \pi_{opt}(\calC_s,p_s,s)$ and replacing the subpath $\pi(t,q_t)$ by $\overline{q_tp_t}\cup \pi_{opt}(\calC_t,p_t,t)$. Clearly,  $L_1(\pi')=L_1(\pi)$ and $\pi'$ is homotopic to $\pi$. Observe that $\pi'$ consists of the following three subpaths: $\pi'(s,p_s)$, which a path from $s$ to $p_s$ in $\calC_s$, $\pi'(p_s,p_t)$, and $\pi'(p_t,t)$, which a path from $p_t$ to $t$ in $\calC_t$. Since both $p_s$ and $p_t$ are in junction rectangles, according to the analysis of Lemma~\ref{lem:target}, the graph $G(\calB)$ has a path $\pi_{G(\calB)}(p_s,p_t)$ from $p_s$ to $p_t$ with the same length as $\pi'(q_s,q_t)$ and  $\pi_{G(\calB)}(p_s,p_t)$ is homotopic to $\pi'(q_s,q_t)$.
As the two subpaths $\pi'(s,p_s)$ and $\pi'(p_t,t)$ correspond to two corridor-connection edges in $G(\calB)$, let $\pi_{G(\calB)}(s,t)$ be the concatenation of the above two corridor-connection edges and $\pi_{G(\calB)}(p_s,p_t)$. According to the above analysis, the length of $\pi_{G(\calB)}(s,t)$ is equal to $L_1(\pi)$ and $\pi_{G(\calB)}(s,t)$ is homotopic to $\pi$. Hence, the first part of the lemma follows.

By using the similar argument as Corollary~\ref{coro:target}, the second part of the lemma can be proved.
\end{proof}

In light of the preceding lemma, we can compute an optimal \st\ path by searching the graph $G(\calB)$. Comparing with the algorithm for the previous case where both $s$ and $t$ are in junction rectangles, one big difference is in the beginning of the algorithm. Here, our algorithm starts to explore the two corridor-connection edges connecting to $s$. For each such edge, say, from $s$ to $s_d$ on a door $d$ of $\calC_s$, we move to the vertex $s_d$ of $G(\calB)$ to obtain two paths and store them at $s_d$, and the two paths are defined as follows. Without loss of generality, we assume $\calC_s$ is locally on the right of $d$. Depending on whether $s$ is horizontally visible to $d$, there are two cases.

If $s$ is not horizontally visible to $d$, recall that in Section~\ref{sec:compute} we defined intervals $I_1(s,d)$, $I_1'(s,d)$, $I_2(s,d)$, $I_2'(s,d)$, and two admissible paths $\pi_1(\calC,s,s_d)$ and $\pi_2(\calC,s,s_d)$ from $s$ to $s_d$ in $\calC$. 
Also, $L_d(\pi_1(\calC,s,s_d))=L_d(s,d)$ and $L_d(\pi_2(\calC,s,s_d))=L_d(s,d)+1$. The last link of $\pi_1(\calC,s,s_d)$ is $I_1'(s,d)$ and the last link of $\pi_2(\calC,s,s_d)$ is $I_2'(s,d)$. We store the two paths $\pi_1(\calC,s,s_d)$ and $\pi_2(\calC,s,s_d)$ at $s_d$.

\begin{figure}[t]
\begin{minipage}[t]{\linewidth}
\begin{center}
\includegraphics[totalheight=1.3in]{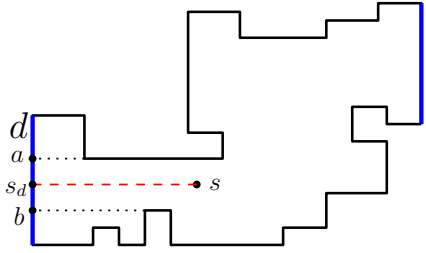}
\caption{\footnotesize Illustrating the interval $I_2(s,d)$, which is the segment $\overline{ab}$ on $d$.
}
\label{fig:visible}
\end{center}
\end{minipage}
\vspace*{-0.15in}
\end{figure}

If $s$ is horizontally visible to $d$, then we can also define $I_1(s,d)$, $I_2(s,d)$, $\pi_1(\calC,s,s_d)$ and $\pi_2(\calC,s,s_d)$ in a way consistent with the above case, as follows. We define both $\pi_1(\calC,s,s_d)$ and $\pi_2(\calC,s,s_d)$  to be the horizontal segment $\overline{ss_d}$. However, we consider $\pi_1(\calC,s,s_d)$ as having a single segment while $\pi_2(\calC,s,s_d)$ as having another vertical segment incident to $s$ with zero length. Hence, we still have $L_d(\pi_1(\calC,s,s_d))=L_d(s,d)$ and $L_d(\pi_2(\calC,s,s_d))=L_d(s,d)+1$. Note that $L_d(s,d)=1$.
We define $I_1(s,d)$ as the single point $s_d$ and define $I_2(s,d)$ as the interval on $d$ swept by $s_d$ if we push $\overline{ss_d}$ vertically in $\calC_s$ (e.g., see Fig.~\ref{fig:visible}). Note that the last segments of both $\pi_1(\calC,s,s_d)$ and $\pi_2(\calC,s,s_d)$ are $\overline{ss_d}$. However,
if a dragging operation is applied, $\overline{ss_d}$ is fixed for $\pi_1(\calC,s,s_d)$ and floating for $\pi_2(\calC,s,s_d)$, i.e., we cannot drag $\overline{ss_d}$ vertically in $\pi_1(\calC,s,s_d)$ but we can do so in $\pi_2(\calC,s,s_d)$ due to the zero-length vertical segment at $s$. We store the two paths $\pi_1(\calC,s,s_d)$ and $\pi_2(\calC,s,s_d)$ at $s_d$.


Next, the algorithm proceeds in the same way as before until when a corridor-connection edge from $t_d$ to $t$ is processed, at which moment, we apply the following {\em concatenation procedure} to concatenate the current paths stored at $t_d$ with the two paths $\pi_1(\calC,t,t_d)$ and $\pi_2(\calC,t,t_d)$ (defined similarly as $\pi_1(\calC,s,s_d)$ and $\pi_2(\calC,s,s_d)$) to obtain an \st\ path. This is done with the help of the two intervals $I_1(t,d)$ and $I_2(t,d)$. The details are given below, which are somewhat similar to the through-corridor-path generating procedure.

Let $\pi(s,t_d)$ be a path stored at $t_d$ and we wish to extend it to $t$. Let $\alpha=\overline{pt_d}$ be the last edge of the path and let $\pi(s,p)$ be the subpath of $\pi(s,t_d)$ between $s$ and $p$. Without loss of generality, we assume that $\calC_t$ is locally on the right of $d$.

If $\alpha$ is horizontal, then we simply concatenate $\pi(s,t_d)$ with $\pi_1(\calC,t,t_d)$ through $t_d$. If $\alpha$ is vertical and $p$ is on $I_1(t,d)$, then we concatenate $\pi(s,p)$ with $\pi_1(\calC,t,t_d)$ through $p$ after dragging the last edge (i.e., the segment incident to $t_d$) of $\pi_1(\calC,t,t_d)$ until $p$. If $\alpha$ is vertical and $p$ is on $I_2(t,d)$, then we concatenate $\pi(s,p)$ with $\pi_2(\calC,t,t_d)$ through $p$ after dragging the last edge of $\pi_2(\calC,t,t_d)$ until $p$. If $p$ is not on $I_2(t,d)$, then we again simply concatenate $\pi(s,t_d)$ with $\pi_1(\calC,t,t_d)$ through $t_d$. By the definitions of the two intervals $I_1(t,d)$ and $I_2(t,d)$, one can verify that the above gives the best solution for extending the path $\pi(s,t_d)$ to $t$.

Once the searching algorithm on the graph $G(\calB)$ is finished, we pick from all the paths stored at $t$ the one with the smallest measure as the optimal solution. The algorithm is applicable to all types of optimal paths and the running times are asymptotically the same as before because computing all corridor-connection edges can be done in additional $O(n)$ time~\cite{ref:SchuiererAn96}.

The above discusses the case where $\calC_s\neq \calC_t$. If
$\calC_s=\calC_t$, then we first compute a smallest path
$\pi_{opt}(\calC,s,t)$ in $\calC_s$ in $O(|\calC_s|)$ time
\cite{ref:SchuiererAn96}. Clearly, if there exists an optimal path
\st\ in $\calC_s$, then $\pi_{opt}(\calC,s,t)$ is a solution.
Otherwise, any optimal \st\ path first goes outside $\calC_s$ through
one door and then gets back to $\calC_s$ through the other door. We
apply the same algorithm as above by conceptually treating $\calC_s$
and $\calC_t$ differently.

If $s$ is in a corridor and $t$ is in a junction rectangle, we can use the similar approach as the above general case. The difference is that the concatenation procedure is not needed any more.

\section{The One-Point Optimal Path Queries}
\label{sec:onepoint}

In this section, we present our results on one-point queries, where
$s$ is the source point and $t$ is the query point. The high-level
scheme of our approach is similar to that in~\cite{ref:ChenOn01},
which is based on the (incorrect) algorithm of~\cite{ref:YangRe95} and
the graph $G(\calV)$ discussed in Section~\ref{sec:old}. Our new
approach is based on our new algorithm and the reduced graph
$G(\calB)$ proposed in Section~\ref{sec:new}.

Let $\calB$ now consist of the source $s$ and all backbone
points. Let $G(\calB)$ be the graph we build in Section~\ref{sec:new} on $\calB$. Note that if $s$ is in a corridor, then the
graph has two corridor-connection edges incident to $s$.
We first consider the minimum-link shortest path queries.

\subsection{The Minimum-Link Shortest Paths}

Consider a query point $t$. We first assume that $t$ is in
a junction rectangle.

As in~\cite{ref:ChenOn01,ref:ChenTw16,ref:ChenSh00},
we define a new graph $G_t(\calB)$ by ``inserting'' $t$ into $G(\calB)$, as follows.
Roughly speaking, $G_t(\calB)$ is the graph defined with respect to $\calB\cup\{t\}$ in the same way as $G(\calB)$ with respect to $\calB$ with the following
constraint: the vertical cut-line $l(t)$ through $t$ is at a leaf node of the
cut-line tree (and thus $l(t)$ does not have any Steiner points). Specifically, let $T(\calB)$ be the cut-line tree of $G(\calB)$.
Since $|\calB|=O(h)$, $T(\calB)$ has $O(h)$ nodes and its height is $O(\log h)$. We first define a set of {\em projection cut-lines} of
$t$. Starting from the root $v$ of $T(\calB)$, if the cut-line $l(v)$
of $v$ is horizontally visible to $t$, then $l(v)$ is a {\em projection
cut-line} of $t$. If $t$ is on the left side of $l(v)$, then we
proceed on the left child of $v$ and the projection cut-lines in the
left subtree of $v$ are defined recursively; otherwise, we proceed on the right
child of $v$. In this way, we can define $O(\log
h)$ projection cut-lines for $t$ because there is at most one projection cut-line at each level of $T(\calB)$.

For each projection cut-line $l(v)$ of $t$, we add a
vertex $v_t$ to $G(\calB)$, where $v_t$ is
a Steiner point that is the horizontal projection of $t$ onto
$l(v)$. Let $a_v$ and $b_v$ be the vertices of $G(\calB)$ on $l(v)$
right above and below $v_t$, respectively, and they are called the
{\em gateways} of $t$. We also add the following three edges to the
graph: $\overline{tv_t}$, $\overline{v_ta_t}$, and $\overline{v_tb_t}$.
Since $t$ has $O(\log h)$ projection cut-lines, we add at most $O(\log
h)$ vertices and edges to $G(\calB)$, and the resulting graph is
$G_t(\calB)$. Let $V_g(t)$ be the set of all gateways of $t$. Clearly,
$|V_g(t)|=O(\log h)$. Intuitively, the gateways ``control'' the paths
from $t$ to all other vertices of $G_t(\calB)$.
Since $G_t(\calB)$ is essentially the graph defined with respect to
$\calB\cup\{t\}$ in the same way as $G(\calB)$ with respect to $\calB$, our algorithm
in Section~\ref{sec:new} can find a minimum-link shortest \st\ path by
searching $G_t(\calB)$.
Based on this observation, we use the following approach to answer the query.

As preprocessing, we apply our algorithm in Section~\ref{sec:new} on
$G(\calB)$ with $s$ as the source. After the algorithm finishes, each vertex $q$ of $G(\calB)$ will (implicitly) store at most sixteen paths $\pi(s,q)$ from $s$ to $q$.
This takes $O(n+h\log^{3/2}h)$ time and $O(n+h\log h)$ space.

Given the query point $t$, we first compute all projection cut-lines
of $t$, which can be done in $O(\log n)$
time~\cite{ref:ChenOn01,ref:ChenTw16,ref:ChenSh00} (e.g., with the
help of the horizontal visibility decomposition of $\calP$). As in
\cite{ref:ChenTw16}, computing
the gateway set $V_g(t)$ can be done in $O(\log h)$ time by searching
the cut-line tree $T(\calB)$ in a top-down manner after building a
fractional cascading data structure on the sorted lists of the
vertices of $G(\calB)$ on all cut-lines of
$T(\calB)$~\cite{ref:ChenTw16}. As the vertices of $G(\calB)$ on all cut-lines can
be sorted in $O(h\log h)$ time, building the fractional cascading data
structure can be done  in $O(h\log h)$ time~\cite{ref:ChazelleFr86}.
Chen et al.~\cite{ref:ChenSh00} provided another (more
involved) way to compute $V_g(t)$  in $O(\log h)$ time.

\begin{figure}[t]
\begin{minipage}[t]{\linewidth}
\begin{center}
\includegraphics[totalheight=1.0in]{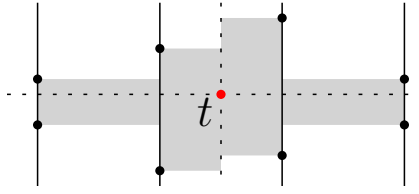}
\caption{\footnotesize Illustrating the gateway region of $t$ (the shaded area). The black points are gateways and the solid vertical lines are cut-lines.
}
\label{fig:gateway}
\end{center}
\end{minipage}
\vspace*{-0.15in}
\end{figure}

Consider a gateway $q$ of $V_g(t)$. For each path $\pi(s,q)$ stored at
$q$, by using the dragging operation we can extend $\pi(s,q)$ to
obtain a path $\pi(s,t)$ from $s$ to $t$. Chen et
al.~\cite{ref:ChenOn01} showed that the dragging operation can be
performed in $O(1)$ time due to the rectilinear convexity of a so-called
{\em gateway region} (e.g., see Fig.~\ref{fig:gateway}). If we extend
the paths stored in all gateways to $t$, then the path with the
smallest measure is a minimum-link shortest path. Since
$|V_g(t)|=O(\log h)$ and the number of paths stored at each gateway is
$O(1)$, we can find such an optimal path in $O(\log h)$ time.

As a summary, if $t$ is in a junction rectangle, computing the measure
of a minimum-link shortest path can be done in $O(\log n)$ time. Note
that outputting an actual path can be done in additional $O(k)$ time
by standard techniques, where $k$ is the link distance of the path. We
omit the details.

If $t$ is in a corridor $\calC$, then we use the idea in
Section~\ref{sec:general}. We first assume $s$ is not in $\calC$.
Hence, an optimal \st\ path must cross a door, say, $d$,
of $\calC$. Let $t_d$ be the corridor-connection point of $t$ on $d$. With
$O(|\calC|)$ time preprocessing on $\calC$ (i.e., building the histogram
partition data structure~\cite{ref:SchuiererAn96}), the following can all be computed in
$O(\log n)$ time: the point $t_d$, the two intervals $I_1(t,d)$,
$I_2(t,d)$, the last segments of the two paths $\pi_1(\calC,t,t_d)$ and
$\pi_2(\calC,t,t_d)$, the measures of the two paths.
Since $t_d$ is in a junction rectangle, we can find
a set $V_g(t_d)$ of gateways in $G(\calB)$. With all above
information, for each path $\pi(s,q)$ stored at a gateway $q$ of
$V_g(t_d)$, we can apply the concatenation procedure to extend
$\pi(s,q)$ to obtain a path $\pi(s,t)$ from $s$ to $t$, and the
measure of $\pi(s,t)$ can be obtained in $O(1)$ time.
In this way, we can obtain $O(\log h)$ candidate paths crossing $d$. We do this for the other door of $\calC$ as well. In total we obtain $O(\log h)$ candidate paths,
and the one with the smallest measure is an optimal \st\ path. Therefore,
the query can be answered in $O(\log n)$ time. If $s$ is also in the
corridor $\calC$, then in addition to the above candidate \st\ paths,
we also need to consider the smallest path from $s$ to $t$ in $\calC$,
which can be obtained in $O(\log n)$ time by the histogram partition data structure
of $\calC$~\cite{ref:SchuiererAn96}.

In summary, we can build a data structure of $O(n+h\log h)$ size in
$O(n+h\log^{3/2}h)$ time such that each one-point minimum-link
shortest path query can be answered in $O(\log n)$ time.

\subsection{The Minimum-Cost Paths}

For the minimum-cost paths, we use the same approach as above. The
difference is that now each vertex of the graph $G(\calB)$ maintains
$O(h)$ paths. Therefore, we need to consider $O(h\log h)$ candidate
paths stored in the $O(\log h)$ gateways of $t$.
Hence, the query time becomes $O(\log n+ h\log h)$. The
preprocessing is the same as those in the algorithm given in Section~\ref{sec:new}, i.e.,
$O(n+h^2\log^{3/2}h)$ time and $O(n+h^2\log h)$ space.

\subsection{The Shortest  Minimum-Link Path Queries}

For the shortest minimum-link paths, following the above approach, we can obtain a result whose complexities are the same as the minimum-cost paths. However, we are able to do better.
The main idea is that we can add more pruning rules to Rule($a_2$) in Section~\ref{sec:new}, so
that for each gateway $q$ of $t$, we can do binary search to find a best path
among all $O(h)$ paths stored at $q$ to connect to $t$, and
consequently we only need to spend $O(\log h)$ time on $q$,
and the total query time becomes $O(\log n+ \log^2 h)$.
Specifically, we replace Rule($a_3$) by the following Rule($a_3'$)

\begin{enumerate}
\item[($a_3'$)]
Let $\pi_1$ be one of $\pi'(s,q)$ and $\pi(s,q)$, and $\pi_2$ the
other. We discard $\pi_2$ if one of the following two cases
happen: (1) $L_1(\pi_1)\leq L_1(\pi_2)$ and $L_d(\pi_1)\leq
L_d(\pi_2)-2$; (2) $L_1(\pi_1)\leq L_1(\pi_2)$,  $L_d(\pi_1) =
L_d(\pi_2)$, the last segments of both paths overlap, and the last segment of
$\pi_1$ is no longer than that of $\pi_2$.
\end{enumerate}

Using the similar analysis as in Section~\ref{sec:correct},
one can verify that with the new rule the previous algorithm still
works. Note that with the new rule each vertex
$q$ still needs to store $\Theta(h)$ paths in the worst case (e.g.,
extending the example in
Fig.~\ref{fig:threepaths} by assuming $L_d(\pi_1)=L_d(\pi_2)=L_d(\pi_3)$
and $L_1(\pi_1)>\L_1(\pi_2)>L_1(\pi_3)$). Based on the new rule,
depending on whether the last segment is from upwards, leftwards,
rightwards, downwards, there are four types of paths stored at $q$.
More importantly to our approach for answering queries, the new rule
guarantees the following property: each type of paths stored at $q$ can
be partitioned into two sets $\Pi_1(q)$ and $\Pi_1(q)$ such that (1)
all paths of $\Pi_1(q)$ (resp., $\Pi_2(q)$) have the same link
distance; (2) the link distance of the paths of $\Pi_1(q)$ is one
larger than that of the paths of $\Pi_2(q)$; (3) the paths of
$\Pi_1(q)$ (resp., $\Pi_2(q)$) can be organized into a sequence
$\pi_1,\pi_2,\ldots, \pi_k$ with $k=O(h)$ such that their lengths are
strictly decreasing and the lengths of their last segments are
strictly increasing (e.g., see Fig.~\ref{fig:threepaths} with
$L_1(\pi_1)>L_1(\pi_2)>L_1(\pi_3)$).

Consider a query point $t$. We first assume that $t$ is in a junction rectangle.
Let $q$ be a gateway of $V_g(t)$. Without loss of generality, we
assume that $q$ is in the first quadrant of $t$. To extend the paths
$\pi(s,q)$ stored at $q$ to $t$, we use the following approach. Note
that we
only need to consider the paths $\pi(s,q)$ whose last segments
are from upwards and rightwards of $q$ (since other paths cannot lead
to an optimal \st\ path). We consider the type of
paths whose last segments are from rightwards of $q$ (the other type of paths can be
processed similarly). Let $\Pi_1(q)$ and $\Pi_2(q)$ be the two sets of
sorted paths. We consider the set $\Pi_1(q)$ (the other set can be
processed similarly). Let the sequence of the paths of $\Pi_1(q)$ be
$\pi_1,\pi_2,\ldots,\pi_k$ as defined above. Suppose we want to find
the best path among above paths to extend it to $t$ with the smallest
measure. Since the lengths of those paths are strictly decreasing, if
we can find the largest index $i\in [1,k]$ such that we can freely drag the
last segment of $\pi_i$ downwards until $t$ without making any extra turn, then
$\pi_i$ is the best path. To find such an index $i$, we can use binary
search as follows. Since the lengths of the last segments of these
paths are strictly increasing, the downward hit vertices of these
segments are actually sorted in increasing order by their
$y$-coordinates. As these hit vertices have already been computed and
associated with these segments, such an index $i$ can be found in
$O(\log h)$ time by binary search on the sequence of the paths. Note that in the case that such an
index $i$ does not exist, for any path $\pi_j$ with $j\in [1,k]$, to extend it to
$t$, we have to make an extra turn, and thus the best path is
$\pi_k$ because its length is the smallest.

The above gives an $O(\log h)$ time algorithm to find among the paths
stored at $q$ the best path to extend to $t$. As there are $O(\log h)$
gateways, the total query time is $O(\log^2 h +\log n)$.

If $t$ is in a corridor, we use a similar approach as above but
on the corridor-connection points of $t$ on the two doors of the corridor, in the
same way as for the minimum-link shortest path queries discussed before. The
query time is still $O(\log^2 h +\log n)$.

Because we need to main the above sorted lists in each vertex $q$ of $G(\calB)$,
we need to modify our preprocessing algorithm.
Fortunately, we can still implement the new algorithm in the same time asymptotically
as before, with the help of van Emde Boas trees~\cite{ref:CLRS09}. The
details are given below.

Consider the list $\Pi_1(q)=\{\pi_1,\pi_2,\ldots,\pi_k\}$ discussed
above. Without loss of  generality, we assume that the last segments of these paths
are all horizontal on the left side of $q$. For each $1\leq i\leq k$, let $q_i$ be the left
endpoint of the last segment of $\pi_i$.
Recall that $q_i$ must be a projection of a vertex $v_i$ of $\calB$
onto the horizontal line through $q$. Let $r(v_i)$ be the {\em rank} of
$v_i$ in $\calB$ in the increasing $x$-coordinate order, i.e., if
$v_i$'s $x$-coordinate is the $j$-th smallest in $\calB$, then
$r(v_i)=j$. We also let $r(q_i)=r(v_i)$, as the rank of $q_i$. Since
$|\calB|=O(h)$, $r(q_i)$ is an integer upper bounded by $O(h)$. We
maintain the ranks of all $q_i$ for $1\leq i\leq k$ in a van Emde Boas
tree $T_{veb}(q)$ of $O(h)$ size, so that each of the following operations can be implemented in $O(\log\log h)$ time~\cite{ref:CLRS09}: search, insert, delete, predecessor, successor, minimum, maximum. Suppose we have a new path $\pi$ from $s$ to $q$ whose last
segment is also horizontal from the left side of $q$ and $L_d(\pi)=L_d(\pi_i)$ for all $i\in [1,k]$. Our task is to update the list $\Pi_1(q)$ with $\pi$ following Rule($a_3'$).

Let $p$ be the left endpoint of the last segment of $\pi$, which is a vertical projection of a vertex $v$ of $\calB$. Let $r(p)$ be the rank of $v$ in $\calB$, which can be obtained in $O(1)$ time if we compute the ranks of all vertices of $\calB$ in the preprocessing. Let $r(q_i)$ be the successor of $r(p)$ in $T_{evb}(q)$, which can be found in $O(\log\log h)$ time. We check whether $L_1(\pi)\geq L_1(\pi_i)$. If yes, then by  Rule($a_3'$), $\pi$ needs not to be stored at $q$ and we are done. Otherwise, we further find the predecessor $r(q_j)$ of $r(p)$ in $T_{veb}(q)$ in $O(\log\log h)$ time.
We check whether $L_1(\pi)> L_1(\pi_j)$. If yes, then  we insert $r(p)$ to $T_{evb}(q)$, and thus insert $\pi$ into the correct position of the list $\Pi_1(q)$. Otherwise, by  Rule($a_3'$), the path $\pi_j$ should be removed, and thus we delete $r(q_j)$ from $T_{evb}(q)$. Next, we find the predecessor of $r(q_j)$ in $T_{evb}(q)$ to check whether the corresponding path should be removed. In this way, updating the set $\Pi_1(q)$ for $\pi$ takes $O((k'+1)\log\log h)$ time, where $k'$ is the number of paths that are removed from $\Pi_1(q)$. Note that once a path is removed it will never be inserted again. Hence, the total sum of such $k'$ in the entire algorithm for all vertices is no more than the total number of paths maintained by the algorithm, which is $O(h^2\log h)$. Therefore, the total time of the algorithm on updating the path lists stored in all vertices of the graph is $O(h^2\log h\log\log h)$, which is still bounded by $O(h^2\log^{3/2}h)$. After the algorithm finishes, in order to facilitate the binary search in our query algorithm, we perform the following ``post-processing'' step: for each vertex $q$, by using the tree $T_{evb}(q)$, we use an array to store the last segments of the  sorted paths of $\Pi_1(q)$. Since the total number of stored paths in all vertices of the graph is $O(h^2\log h)$, the post-processing step can be easily done in $O(h^2\log^{3/2}h)$ time as well.

As a summary,
we can build a data structure of $O(n+h^2\log h)$ size in $O(n+h^2\log^{3/2}h)$ time such that each shortest minimum-link path query can be answered in $O(\log n +\log^2 h)$ time.


The following theorem summarizes our results on one-point path queries.


\begin{theorem}
For the one-point path queries, we have the following results.
\begin{enumerate}
\item
For minimum-link shortest paths,
we can build a data structure of $O(n+h\log h)$ size in $O(n+h\log^{3/2}h)$ time
such that each query can be answered in $O(\log n)$ time.

\item
For minimum-cost paths,
we can build a data structure of $O(n+h^2\log h)$ size in $O(n+h^2\log^{3/2}h)$ time such that each query can be answered in $O(\log n+h\log h)$ time.

\item
For shortest minimum-link paths,
we can build a data structure of $O(n+h^2\log h)$ size in $O(n+h^2\log^{3/2}h)$ time
such that each query can be answered in $O(\log n+\log^2 h)$ time.
\end{enumerate}
\end{theorem}

\section{The Two-Point Optimal Path Queries}
\label{sec:twopoint}

In this section, we present our results for two-point queries, i.e., both $s$ and $t$ are query points. We first give an approach that follows the similar scheme as  in~\cite{ref:ChenOn01}, and then describe another approach that can reduce the query time by a logarithmic factor with slightly more preprocessing. The second approach follows the scheme in~\cite{ref:ChenTw16} for solving two-point $L_1$ shortest path queries in an arbitrary polygonal domain.

Consider any two query points $s$ and $t$. Let $\calB$ be the set of all backbone points and let $G(\calB)$ be the reduced graph proposed in Section~\ref{sec:new}. We ``insert'' both $s$ and $t$ into $G(\calB)$ in the same way as in Section~\ref{sec:onepoint}, and let $G_{st}(\calB)$ be the resulting graph. We have the following lemma.

\begin{lemma}
Unless $\calP$ contains an L-shaped path connecting $s$ and $t$, applying our algorithm in Section~\ref{sec:new} on $G_{st}(\calB)$ can find an optimal \st\ path.
\end{lemma}
\begin{proof}
Our algorithm in Section~\ref{sec:new} works due to the property in Observation~\ref{obser:recreduced}. Hence, the algorithm can find an optimal \st\ path on $G_{st}(\calB)$ if the following property holds: for any two points $p$ and $q$ in $\calB\cup \{s,t\}$, if $R_{pq}$ is empty, then $G_{st}(\calB)$ contains a staircase path connecting $p$ and $q$. If both $p$ and $q$ are in $\calB$, then the property trivially holds due to our way of constructing $G(\calB)$, which is a subgraph of $G_{st}(\calB)$. If only one of $p$ and $q$ is in $\{s,t\}$, e.g., $q=t$ and $p\in \calB$, then according to our discussion in Section~\ref{sec:onepoint}, our way of inserting $t$ into $G(\calB)$ makes sure that the property also holds.

It remains to consider the case where $\{p,q\}=\{s,t\}$. Note that if $R_{st}$ is empty, then there must be an L-shaped \st\ path in $\calP$. Hence, if $\calP$ does not have an L-shaped \st\ path, then $R_{st}$ is not empty and thus the property trivially holds. The lemma thus follows.
\end{proof}

Based on the preceding lemma, we first check whether $\calP$ has an L-shaped \st\ path, which can be done in $O(\log n)$ time~\cite{ref:ChenOn01,ref:ChenTw16,ref:ChenSh00}, e.g., by using the horizontal and vertical visibility decompositions of $\calP$. In the following, we assume that $\calP$ does not have an L-shaped \st\ path.

We first discuss the minimum-link shortest paths.
As in \cite{ref:ChenOn01}, in the preprocessing, for each vertex $p$ of the graph $G(\calB)$, we run our algorithm in Section~\ref{sec:new} on the graph $G(\calB)$ with $p$ as the source. After the algorithm, each vertex $q$ of the graph stores $O(1)$ paths $\pi(p,q)$ from $p$ to $q$. Using the techniques in \cite{ref:ChenOn01}, this can be done in $O(n+h^2\log^2 h)$ time and space for all vertices $p$ of $G(\calB)$.

We assume that both $s$ and $t$ are in junction rectangles.
To answer the query, we first compute the gateway sets $V_g(s)$ and $V_g(t)$ in $O(\log n)$ time. As in \cite{ref:ChenOn01}, we perform the dragging operations on the first segment and the last segment for $O(\log^2 h)$ paths $\pi(p,q)$ with $p\in V_g(s)$ and $q\in V_g(t)$ to obtain an optimal \st\ path, as follows.
Consider a gateway $p$ in $V_g(s)$. For each gateway $q\in V_g(t)$, recall that in the preprocessing, we have stored $O(1)$ paths $\pi(p,q)$ at $q$ with $p$ as the source point. For each such path $\pi(p,q)$, we perform the dragging operation on both its last and first segments to extend the path to obtain an \st\ path. Since $|V_g(t)|=O(\log h)$, there are $O(\log h)$ paths for $p$. Since $|V_g(s)|=O(\log h)$, there are $O(\log^2 h)$ such paths we need to consider. Because we can perform each dragging operation in $O(1)$ time~\cite{ref:ChenOn01}, the total time is $O(\log^2 h)$. Among all obtained \st\ paths, we return the one with the smallest measure as the optimal solution. Hence, the total query time is $O(\log n+\log^2 h)$.

If at least one of $s$ and $t$ is in corridors, then similar to the one-point query problem, we first find the corresponding corridor-connection points and then apply the concatenation procedure in addition to the dragging operations. We still need to consider $O(\log^2 h)$ paths, and the  query time is still $O(\log n+\log^2 h)$.

In summary, with $O(n+h^2\log^2 h)$ time and space preprocessing, each two-point minimum-link shortest path query can be answered in $O(\log n+\log^2 h)$ time.

For the minimum-cost path queries, we use the similar approach. In the preprocessing, for each vertex $p$ of the graph $G(\calB)$, we run our algorithm in Section~\ref{sec:new} (for computing the minimum-cost paths) on $G(\calB)$ with $p$ as the source point. Using the techniques in \cite{ref:ChenOn01}, this can be done in $O(n+h^3\log^2h)$ time and space.
The query algorithm follows the same scheme as above. Since in this problem each vertex stores $O(h)$ paths, we have to consider $O(h\log^2 h)$ paths. Therefore, the query time becomes $O(\log n+h\log^2 h)$.

For the shortest minimum-link path queries, we also use the similar approach. We do the same preprocessing as before by using the algorithm for computing shortest minimum-link paths in Section~\ref{sec:new}. As the minimum-cost path problem, the preprocessing takes $O(n+h^3\log^2h)$ time and space. The query algorithm follows the same scheme as above. In this problem, although each vertex stores $O(h)$ paths, we can use binary search in the same way as in the one-point query problem, and thus the total query time is $O(\log n+\log^3 h)$.

In addition, we also consider the two-point minimum-link path queries since the problem was not studied before. As discussed in Section~\ref{sec:corrected}, Rule($a_4$) makes sure that each vertex of $G(\calB)$ only needs to store $O(1)$ paths. Hence, the algorithm is similar to the one for the minimum-link shortest paths. The preprocessing time and space is $O(n+h^2\log^2 h)$ and the query time is $O(\log n+\log^2 h)$.

\subsection{Reducing the Query Times}

With slightly more preprocessing, we reduce the query time for each problem by a factor of $\log h$. Similar approach was already used in \cite{ref:ChenTw16}. The main idea is to build an {\em enhanced graph} $G_E(\calB)$  of larger size on the backbone points of $\calB$, 
so that we only need a set of $O(\sqrt{\log h})$ gateways for each of $s$ and $t$, which reduces the query time by a factor of $\log h$. The details are given below.

The enhanced graph $G_E(\calB)$ is still built on $\calB$ with respect to the reduced domain $\calP_r$ introduced in Section~\ref{sec:compute}. Comparing with the original graph $G(\calB)$, $G_E(\calB)$ has more Steiner points as vertices and more edges. Specifically, for each vertex $v$ of $\calB$, instead of projecting it to a single cut-line in each level of the cut-line tree $T(\calB)$, it is projected to $O(2^{\sqrt{\log n}})$ cut-lines in every $\sqrt{\log n}$ consecutive levels of the $T(\calB)$ (thus creating $O(2^{\sqrt{\log n}})$ Steiner points), and these cut-lines form a complete binary tree of height $\sqrt{\log n}$. In this way, the graph $G_E(\calB)$ has $O(h\sqrt{\log h}2^{\sqrt{\log h}})$ vertices and edges.
Using $G_E(\calB)$, for any query point, we can define a set of $O(\sqrt{\log h})$ gateways that  ``control'' paths from the query point to all vertices of $G_E(\calB)$. By using the reduced domain $\calP_r$, $G_E(\calB)$ can be built in $O(n+h\log^{3/2} h2^{\sqrt{\log h}})$ time.
Refer to \cite{ref:ChenTw16} for more details.

We first discuss the minimum-link shortest paths. We do the following in the preprocessing. For each vertex $p$ of $G_E(\calB)$, we apply our algorithm in Section~\ref{sec:new} on $G_E(\calB)$ with $p$ as the source point, after which for each vertex $q$ of the graph, it stores $O(1)$ paths $\pi(p,q)$ from $p$ to $q$. This can be done in $O(n+h\log^{3/2} h2^{\sqrt{\log h}})$ time and $O(n+h\log^{1/2} h2^{\sqrt{\log h}})$ space. Running the algorithm for all vertices $p$ of $G_E(\calB)$ takes $O(n+h^2\log^{2} h4^{\sqrt{\log h}})$ time and $O(n+h^2\log h4^{\sqrt{\log h}})$ space.

We assume that both $s$ and $t$ are in junction rectangles. We first compute their gateways sets $V_g(s)$ and $V_g(t)$, which can be done again in $O(\log h)$ time (with $O(n+h\log^{3/2} h2^{\sqrt{\log h}})$ time and $O(n+h\log^{1/2} h2^{\sqrt{\log h}})$ space preprocessing)~\cite{ref:ChenTw16}. Since the sizes of both gateway sets are bounded by $O(\sqrt{\log h})$, we only need to consider $O(\log h)$ paths to extend to connect $s$ and $t$. Therefore, the query time becomes $O(\log h)$. If $s$ or $t$ is in a corridor, then we again need to first compute their corridor-connection points in $O(\log n)$ time, and then follow the approach we discussed before but with only $O(\log h)$ paths to consider. Hence, the total query time is $O(\log n)$.

In summary, with $O(n+h^2\log^{2} h4^{\sqrt{\log h}})$ time and $O(n+h^2\log h4^{\sqrt{\log h}})$ space preprocessing, each two-point query of the minimum-link shortest paths can be answered in $O(\log n)$ time. Note that $h^2\log^{2} h4^{\sqrt{\log h}}=O(h^{2+\epsilon})$ for any $\epsilon>0$.

The same approach is also applicable for other two types of optimal paths. For the minimum-cost paths, the preprocessing has one more $h$ factor on both the time and space. Specifically, with $O(n+h^3\log^{2} h4^{\sqrt{\log h}})$ time and $O(n+h^3\log h4^{\sqrt{\log h}})$ space preprocessing, each query can be answered in $O(\log n+h\log h)$ time. For the shortest minimum-link paths, the preprocessing complexities are the same as the above for the minimum-cost paths, but the query time is $O(\log n+\log^2 h)$. For minimum-link path queries, the complexities of the preprocessing and the query algorithm are all the same as those for the minimum-link shortest paths.

We summarize the two-point query results for all problems in the following theorem.

\begin{theorem}
For the two-point path queries, we have the following results.
\begin{enumerate}
\item
For minimum-link shortest paths or minimum-link paths,
we can build a data structure of $O(n+h^2\log^2 h)$ size in $O(n+h^2\log^2 h)$ time such that each query can be answered in $O(\log n+\log^2 h)$ time; alternatively, we can build a data structure of $O(n+h^2\log h4^{\sqrt{\log h}})$ size in $O(n+h^2\log^{2} h4^{\sqrt{\log h}})$ time such that each query can be answered in $O(\log n)$ time.

\item
For minimum-cost paths, we can build a data structure of $O(n+h^3\log^2 h)$ size in $O(n+h^3\log^2 h)$ time such that each query can be answered in $O(\log n+h\log^2 h)$ time; alternatively, we can build a data structure of $O(n+h^3\log h4^{\sqrt{\log h}})$ size in $O(n+h^3\log^{2} h4^{\sqrt{\log h}})$ time such that each query can be answered in $O(\log n+h\log h)$ time.

\item
For shortest minimum-link paths,
we can build a data structure of $O(n+h^3\log^2 h)$ size in $O(n+h^3\log^2 h)$ time such that each query can be answered in $O(\log n+\log^3 h)$ time; alternatively, we can build a data structure of $O(n+h^3\log h4^{\sqrt{\log h}})$ size in $O(n+h^3\log^{2} h4^{\sqrt{\log h}})$ time such that each query can be answered in $O(\log n+\log^2 h)$ time.
\end{enumerate}
\end{theorem}



\bibliographystyle{plain}
\bibliography{reference}

\begin{thebibliography}{10}

\bibitem{ref:Bar-YehudaTr94}
R.~Bar-Yehuda and B.~Chazelle.
\newblock Triangulating disjoint {Jordan} chains.
\newblock {\em International Journal of Computational Geometry and
  Applications}, 4(4):475--481, 1994.

\bibitem{ref:BenderTh00}
M.~Bender and M.~Farach-Colton.
\newblock The {LCA} problem revisited.
\newblock In {\em Proc. of the 4th Latin American Symposium on Theoretical
  Informatics}, pages 88--94, 2000.


\bibitem{ref:deBergOn91}
M.~de~Berg.
\newblock On rectilinear link distance.
\newblock {\em Computational Geometry: Theory and Applications}, 1:13--34,
  1991.

\bibitem{ref:ChazelleAn88}
B.~Chazelle.
\newblock An algorithm for segment-dragging and its implementation.
\newblock {\em Algorithmica}, 3(1--4):205--221, 1988.

\bibitem{ref:ChazelleFr86}
B.~Chazelle and L.~Guibas.
\newblock Fractional cascading: {I. A} data structuring technique.
\newblock {\em Algorithmica}, 1(1):133--162, 1986.

\bibitem{ref:ChenOn01}
D.Z. Chen, O.~Daescu, and K.S. Klenk.
\newblock On geometric path query problems.
\newblock {\em International Journal of Computational Geometry and
  Applications}, 11(6):617--645, 2001.

\bibitem{ref:ChenTw16}
D.Z. Chen, R.~Inkulu, and H.~Wang.
\newblock Two-point {$L_1$} shortest path queries in the plane.
\newblock {\em Journal of Computational Geometry}, 1:473--519, 2016.

\bibitem{ref:ChenSh00}
D.Z. Chen, K.S. Klenk, and H.-Y.T. Tu.
\newblock Shortest path queries among weighted obstacles in the rectilinear
  plane.
\newblock {\em SIAM Journal on Computing}, 29(4):1223--1246, 2000.

\bibitem{ref:ChenA11}
D.Z. Chen and H.~Wang.
\newblock A nearly optimal algorithm for finding {$L_1$} shortest paths among
  polygonal obstacles in the plane.
\newblock In {\em Proc. of the 19th European Symposium on Algorithms (ESA)},
  pages 481--492, 2011.

\bibitem{ref:ChenL113STACS}
D.Z. Chen and H.~Wang.
\newblock {$L_1$} shortest path queries among polygonal obstacles in the plane.
\newblock In {\em Proc. of 30th Symposium on Theoretical Aspects of Computer
  Science (STACS)}, pages 293--304, 2013.

\bibitem{ref:ClarksonRe87}
K.~Clarkson, S.~Kapoor, and P.~Vaidya.
\newblock Rectilinear shortest paths through polygonal obstacles in {$O(n
  \log^2 n)$} time.
\newblock In {\em Proc. of the 3rd Annual Symposium on Computational Geometry},
  pages 251--257, 1987.

\bibitem{ref:ClarksonRe88}
K.~Clarkson, S.~Kapoor, and P.~Vaidya.
\newblock Rectilinear shortest paths through polygonal obstacles in {$O(n
  \log^{2/3} n)$} time.
\newblock Manuscript, 1988.

\bibitem{ref:CLRS09}
T.~Cormen, C.~Leiserson, R.~Rivest, and C.~Stein.
\newblock {\em Introduction to Algorithms}.
\newblock MIT Press, 3nd edition, 2009.

\bibitem{ref:DasGe91}
G.~Das and G.~Narasimhan.
\newblock Geometric searching and link distance.
\newblock In {\em Proc. of the 2nd Workshop of Algorithms and Data Structures},
  Lecture Notes in Computer Science, pages 261--272. Springer, 1991.

\bibitem{ref:HarelFa84}
D.~Harel and R.E. Tarjan.
\newblock Fast algorithms for finding nearest common ancestors.
\newblock {\em SIAM Journal on Computing}, 13:338--355, 1984.

\bibitem{ref:HershbergerCo94}
J.~Hershberger and J.~Snoeyink.
\newblock Computing minimum length paths of a given homotopy class.
\newblock {\em Computational Geometry: Theory and Applications}, 4(2):63--97,
  1994.

\bibitem{ref:ImaiEf86}
H.~Imai and T.~Asano.
\newblock Efficient algorithms for geometric graph search problems.
\newblock {\em SIAM Journal on Computing}, 15(2):478--494, 1986.

\bibitem{ref:LeeSh91}
D.T. Lee, C.D. Yang, and T.H. Chen.
\newblock Shortest rectilinear paths among weighted obstacles.
\newblock {\em International Journal of Computational Geometry and
  Applications}, 1(2):109--124, 1991.

\bibitem{ref:MitchellAn89}
J.S.B. Mitchell.
\newblock An optimal algorithm for shortest rectilinear paths among obstacles.
\newblock Abstracts of the {\em 1st Canadian Conference on Computational
  Geometry}, 1989.

\bibitem{ref:MitchellL192}
J.S.B. Mitchell.
\newblock {$L_1$} shortest paths among polygonal obstacles in the plane.
\newblock {\em Algorithmica}, 8(1):55--88, 1992.

\bibitem{ref:MitchellMi14}
J.S.B. Mitchell, V.~Polishchuk, and M.~Sysikaski.
\newblock Minimum-link paths revisited.
\newblock {\em Computational Geometry: Theory and Applications}, 47:651--667,
  2014.

\bibitem{ref:MitchellAn15}
J.S.B. Mitchell, V.~Polishchuk, M.~Sysikaski, and H.~Wang.
\newblock An optimal algorithm for minimum-link rectilinear paths in
  triangulated rectilinear domains.
\newblock In {\em Proc. of the 42nd International Colloquium on Automata,
  Languages and Programming (ICALP)}, pages 947--959, 2015.

\bibitem{ref:MitchellMi92}
J.S.B. Mitchell, G.~Rote, and G.~Woeginger.
\newblock Minimum-link paths among obstacles in the plane.
\newblock {\em Algorithmica}, 8:431--459, 1992.

\bibitem{ref:PolishchukkLink05}
V.~Polishchuk and J.S.B. Mitchell.
\newblock {$k$-Link} rectilinear shortest paths among rectilinear obstacles in
  the plane.
\newblock In {\em Proc. of the 17th Canadian Conference on Computational
  Geometry (CCCG)}, pages 101--104, 2005.

\bibitem{ref:SatoA87}
M.~Sato, J.~Sakanaka, and T.~Ohtsuki.
\newblock A fast line-search method based on a tile plane.
\newblock In {\em Proc. of the IEEE International Symposium on Circuits and
  Systems}, pages 588--597, 1987.

\bibitem{ref:SchuiererAn96}
S.~Schuierer.
\newblock An optimal data structure for shortest rectilinear path queries in a
  simple rectilinear polygon.
\newblock {\em International Journal of Compututational Geometry and
  Applications}, 6:205--226, 1996.

\bibitem{ref:WuRe87}
Y.-F. Wu, P.~Widmayer, M.D.F. Schlag, and C.K. Wong.
\newblock Rectilinear shortest paths and minimum spanning trees in the presence
  of rectilinear obstacles.
\newblock {\em IEEE Transactions on Computers}, 36:321--331, 1987.

\bibitem{ref:YangOn92}
C.D. Yang, D.T. Lee, and C.K. Wong.
\newblock On bends and lengths of rectilinear paths: A graph-theoretic
  approach.
\newblock {\em International Journal of Computational Geometry and
  Application}, 02:61--74, 1992.

\bibitem{ref:YangRe95}
C.D. Yang, D.T. Lee, and C.K. Wong.
\newblock Rectilinear path problems among rectilinear obstacles revisited.
\newblock {\em SIAM Journal on Computing}, 24:457--472, 1995.

\end{thebibliography}

%


\end{document}